\documentclass[11pt]{article}
\usepackage[letterpaper, margin=1in]{geometry}  
\usepackage{pgfplots}
\usepackage{pgfplotstable}
\usepackage{tikz}
\usetikzlibrary{calc}
\pgfplotsset{compat=1.18}
\usetikzlibrary{3d, arrows.meta}
\usepackage{tikz-cd}
\usetikzlibrary{fit, positioning, backgrounds}
\usepackage{xcolor}
\usepackage{nicematrix}
\usepackage{witharrows}
\usepackage{setspace}
\usepackage{parskip}        
\DeclareSymbolFont{extraup}{U}{zavm}{m}{n}
\DeclareMathSymbol{\varheart}{\mathalpha}{extraup}{86}
\DeclareMathSymbol{\vardiamond}{\mathalpha}{extraup}{87}
\usepackage{amsmath}
\usepackage{amsthm}
\usepackage{amssymb}
\usepackage{appendix}
\usepackage[hidelinks]{hyperref}
\usepackage[shortlabels]{enumitem}
\usepackage[ruled,lined,linesnumbered]{algorithm2e}
\usepackage[most]{tcolorbox}
\usepackage{quiver} 

\SetCommentSty{mycommfont}
\usepackage{xcolor}
\usepackage{graphicx} 
\usepackage{authblk}
\newtheorem{example}{Example}
\newtheorem{problem}{Problem}
\newtheorem{theorem}{Theorem}
\newtheorem{lemma}{Lemma}

\newtheorem{definition}{Definition}
\newtheorem{remark}{Remark}
\newtheorem{observation}{Observation}
\newtheorem{proposition}{Proposition}
\newtheorem{corollary}{Corollary}

\usepackage[dvipsnames]{xcolor}
\usepackage{ulem}
\newcommand{\simon}[1]{{\color{red}{{}#1}}}
\newtcolorbox[number within=section]{boxalgorithm}[2][]{%
    colframe=black,   
    colback=white,             
    coltitle=black,            
    fonttitle=\bfseries,       
    enhanced,                  
    colbacktitle=white,  
    title={Algorithm~\arabic{algocf}: #2},
    #1  
}
\newif\ifshowline
\showlinetrue

\title{Compression of Voxelized Vector Field Data by Boxes is Hard}
\ifshowline
\author[1]{Simon Zhang\thanks{Department of Computer Science and Engineering, The Ohio State University. Email: zhang.680@osu.edu}}
\else
\author{Anonymous Authors}
\fi
\date{}

\begin{document}

\maketitle
\begin{abstract}
    Voxelized vector field data consists of a vector field over a high dimensional lattice. The lattice consists of integer coordinates called voxels. The voxelized vector field assigns a vector at each voxel. This data type encompasses images, tensors, and voxel data. 
    
    Assume there is a nice energy function on the vector field. We consider the problem of lossy compression of voxelized vector field data in Shannon's rate-distortion framework. This means the data is compressed and then decompressed up to an error bound on the energy distortion at each voxel. 
    
    Our first result is that under certain conditions, lossy compression of voxelized vector fields is undecidable to compute. This is caused by having an infinite number of Euclidean vectors. We formulate this problem instead in terms of clustering the finite number of indices of a voxelized vector field by boxes. We call this problem the $(k,D)$\textsc{-RectLossyVVFCompression} problem. 

    We show four main results about the $(k,D)$\textsc{-RectLossyVVFCompression} problem. The first is that it is decidable. The second is that decompression for this problem is polynomial time tractable. This means that the only obstruction to a tractable solution of the  $(k,D)$\textsc{-RectLossyVVFCompression} problem lies in the compression stage. This is shown by the two hardness results about the compression stage. We show that the compression stage is NP-Hard to compute exactly and that it is even APX-Hard to approximate for $k,D\geq 2$. 
    
    Assuming $P\neq NP$, this shows  that when $k,D \geq 2$ there can be no exact polynomial time algorithm nor can there even be a PTAS  approximation algorithm for the  \\$(k,D)$\textsc{-RectLossyVVFCompression} problem.
\end{abstract}
\section{Introduction}
In 1959 Shannon formulated rate-distortion theory \cite{shannon1959coding}. In this theory we are given a domain $\mathcal{X}$ equipped with a distortion function $\rho: \mathcal{X}\times \mathcal{X} \rightarrow \mathbb{R}_{\geq 0}$. The distortion function on a pair $(x,y) \in \mathcal{X}\times \mathcal{X}$, denoted $\rho(x,y)$, determines the cost of replacing an element $x \in \mathcal{X}$ with the element $y \in \mathcal{X}$. In rate-distortion theory the goal is \textbf{lossy compression} of a random variable $X$ to an estimated random variable $\hat{X}$. The random variable $\hat{X}$ is called the reconstruction and differs from $X$ up to a bound on its distortion. The reconstruction is still supported on the support $\mathcal{X}$. Shannon formulates this as minimizing the \textbf{rate} $R$, or bits of information per sample from $X$ needed for the reconstruction $\hat{X}$~\cite{cover1999elements}.
A robust formulation for a sequence of random variables is expressed by the following optimization problem.
We call this the \textbf{lossy indexed compression} problem (\textsc{LossyIndexedCompression}):

Given an error bound $\varepsilon>0$ and $n$ non-i.i.d. random variables $X_1,...,X_n$, find the optimal rate $R_n^*$ satisfying:
\begin{subequations}
\begin{equation}
    R_n^*=\inf_{f_n,g_n}R_n
\end{equation}
\begin{equation}
   \text{ s.t. }  \rho(X_i, \hat{X}_i) \leq \varepsilon, \forall i=1,...,n 
\end{equation}
\begin{equation}
    (\hat{X}_1,...,\hat{X}_n)=g_n \circ f_n(X_1,...,X_n) 
\end{equation}
\end{subequations}
where 
\begin{equation}
    f_n:\mathcal{X}^n \rightarrow  \{1,...,2^{R_n}\}, g_n: \{1,...,2^{R_n}\} \rightarrow {\mathcal{X}}^n 
\end{equation}
Where the number $R_n$, is the \textbf{rate} we are optimizing and where the set $\{1,...,2^{R_n}\}$ is called the \textbf{codebook}, which consists of codewords.


We call the inequality constraint of the lossy indexed compression problem \textbf{index consistency} for distortion $\rho$ at error bound $\varepsilon$. 

Due to the non-i.i.d. assumption  and the index consistency criterion of the lossy indexed compression problem, the optimal rate is related to the Kolmogorov complexity~\cite{kolmogorov1965three} over index consistent sequences and their reconstructions. 

This is defined as the \emph{smallest number of bits of a code $Q^*$} so that on a prefix univeral Turing machine $U_T$:
\begin{equation}
    U_T(Q^*)=U_T(P;(x_1,...,x_n))=(\hat{x}_1,...,\hat{x}_n),\forall (x_1,...,x_n) \sim P(X_1,...,X_n)
\end{equation}
where $P$ is a program that when run on any $(x_1,...,x_n) \sim P(X_1,...,X_n)$ outputs an index consistent reconstruction $(\hat{x}_1,...,\hat{x}_n) \in \mathcal{X}^n$ for distortion $\rho$ and error bound $\varepsilon>0$.

Assuming the \textbf{decoder is Kolmogorov}, meaning: $g_n:=U_T$, for $U_T$  a prefix universal Turing machine, then we have:
\begin{equation}
    f_n: (x_1,...,x_n)\mapsto Q^*, (x_1,...,x_n) \sim P(X_1,...,X_n)
\end{equation}
\begin{equation}
    g_n: Q^*\mapsto U_T(Q^*)
\end{equation}
Since the length in bits of $Q^*$, denoted $\textsf{size}(Q^*)$ is shortest, it is unique. The codebook $\{1,...,2^{\textsf{size}(Q^*)}\}$ is in bijection with all possible bit strings of length $\textsf{size}(Q^*)$. This is the codomain and domain, $\{1,...,2^{\textsf{size}(Q^*)}\}$, of $f_n,g_n$ respectively. 
\begin{lemma}
In the \textbf{lossy indexed compression} problem with Kolmogorov decoder:
    \begin{equation}
     R_n^* \geq \textsf{size}(Q^*) -C
\end{equation}
where $C$ is dependent only on the choice of the prefix universal Turing machines.
\end{lemma}
\begin{proof}
    This follows by the Fundamental theorem from ~\cite{kolmogorov1965three}. 
\end{proof}
\begin{theorem}\label{thm: undecidable-lossy-compression}
    (Undecidability) 
        If there is a $\varepsilon>0$ and a $x \in \mathcal{X}$ where the set $B_x(\varepsilon)=\{\hat{x} \in \mathcal{X}: \rho(x,\hat{x})\leq \varepsilon\}$ is infinite, 
        then the \textbf{lossy indexed compression} problem with Kolmogorov decoder is undecidable.
\end{theorem}
\begin{proof}
Assume WLOG. $n=1$.

According to Lemma \ref{lemma: undecidable-kolmogorov} in the Appendix, for a single $x \in \mathcal{X}$, computing the optimal $Q^*$ over the infinite set $B_x(\varepsilon)=\{\hat{x} \in \mathcal{X}: \rho(x,\hat{x})\leq \varepsilon\}$ is undecidable.

This is a subproblem of the \textsc{LossyIndexedCompression} problem. This is because the instances are restricted to be of the form:
\begin{equation}
    (\varepsilon,X_1): P(X_1)=\delta_{x}, \lvert B_x(\varepsilon)\rvert=\infty
\end{equation}

Thus, we can Turing reduce this subproblem to the \textsc{LossyIndexedCompression}  problem, see Proposition \ref{prop: subproblem-Treduce} in Appendix. Thus the \textsc{LossyIndexedCompression} is undecidable.
\end{proof}
This result is caused by the infinite nature of $\mathcal{X}$ and the distortion $\rho$. If $\mathcal{X}$ involves real valued measurements over a metric, then the conditions of Theorem \ref{thm: undecidable-lossy-compression} hold true. For example, in this paper we are interested in understanding an indexed list of vectors from a vector field in the form of a voxelized vector field. Each vector consists of infinitely many possible real valued measurements. Their distortion is naturally measured by a distance. 

The index $\{1,...,n\}$ for the data, however, is of finite size. In this paper we will focus on the redundancies in this finite sized index for the lossy indexed compression problem. 

\textbf{Index Clustering: }
The lossy indexed compression problem requires the map of \emph{all $n$ assignments} $i \mapsto X_i$ to be compressed. 
To compress the $n$ assignments, certainly we could have a codebook consisting of codewords that represent a \emph{matching} consisting of exactly $n$ assignments $i\mapsto C_i$, $C_i$ a compression of $X_i$. Any encoder-decoder pair is index consistent with this codebook since none of the $n$ assignments are forgotten. 

However, if we consider compressing over the power set of $\{1,...,n\}$, then we can achieve an optimal compression rate no worse than the compression rate for this codebook. This can be viewed as form of ``clustering." 

This brings about the following Lemma:
\begin{lemma}\label{prop: power-expand}
    In the \textbf{lossy indexed compression problem}, any encoder $f_n: \mathcal{X}^n \rightarrow \{1,...,2^{R_n}\}$ can be expressed as the following \emph{power set expansion} composition:
\begin{equation}\label{eq: power-expand}
    f_n(X_1,...,X_n)=\Lambda(\{(\kappa(S),\phi(S,\{X_i: i \in S\})): S \subseteq \{1,...,n\}\})
\end{equation}
for some $\phi: 2^{\{1,...,n\}}\times 2^{\mathcal{X}^n} \rightarrow \mathcal{Y}$, some $\kappa: 2^{\{1,...,n\}}\rightarrow \mathcal{Z}$ and some $\Lambda: 2^{\mathcal{Z}\times \mathcal{Y}}\rightarrow \{1,...,2^{R_n}\}$. 

Furthermore, the optimal rate $R_n^*$ satisfies:
\begin{equation}
    R_n^*\leq R_n^{\text{match}}
\end{equation}
for encoders $f_n: \mathcal{X}^n \rightarrow \{1,...,2^{R_n^{\text{match}}}\}$ of the form:
\begin{equation}
    f_n(X_1,...,X_n)=\Lambda(\{(\kappa(\{i\}),\phi(\{i\},X_i)): i \in \{1,...,n\}\})
\end{equation}
for some $\phi:\{\{1\},...,\{n\}\}\times 2^{\mathcal{X}^n} \rightarrow \mathcal{Y}$, some $\kappa: \{\{1\},...,\{n\}\}\rightarrow \mathcal{Z}$ and some $\Lambda: 2^{\mathcal{Z}\times \mathcal{Y}}\rightarrow \{1,...,2^{R_n^{\text{match}}}\}$. 
\end{lemma}
\begin{proof}
   For the first part, we can find the maps $\phi,\kappa,\Lambda$ to express $f_n$:  
\begin{equation}
    \phi(S,\{X_i: i \in S\}):= \{X_i: i \in S\}; \forall S \subseteq \{1,..,n\}
\end{equation} 
\begin{equation}
    \kappa(S):= S; \forall S \subseteq \{1,...,n\}
\end{equation}
\begin{equation}
    \Lambda(\{(S,\{Y_i: i \in S\}): S \subseteq \{1,..,n\}\}):= f_n(Y_1,...,Y_n) 
\end{equation}
The second part follows since this is a special case of the first part and the lossy indexed compression problem is a minimization problem. 
\end{proof}

We will show how this expansion over ``clusters" is meaningful for the lossy compression of voxelized vector fields. This will result in our formulation of the $(k,D)$\textsc{-RectLossyVVFCompression} problem, which is no longer undecidable.
\begin{figure}[!h]
    \centering
\begin{tikzpicture}[scale=0.8]
\begin{axis}[
  colormap/viridis,
  point meta min=0,
  point meta max=20,
  width=10cm,
  axis lines=box,
  ticks=none,
  scatter/use mapped color={draw=black, fill=mapped color},
  grid=major,
grid style={dashed, gray!50},
legend style={at={(0.5,-0.05)}, anchor=south}
]

\foreach \x in {-1,0,1}
  \foreach \y in {-1,0,1}
    \foreach \z in {1,2}{
      \pgfmathsetmacro{\vx}{-\y}
      \pgfmathsetmacro{\vy}{\x}
      \pgfmathsetmacro{\vz}{-\z}
      
      \pgfmathsetmacro{\E}{0.5*(\vx*\vx + \vy*\vy + \vz*\vz) + 10*\z}
      
      \pgfmathsetmacro{\norm}{sqrt(\vx*\vx + \vy*\vy + \vz*\vz+1)}
      \pgfmathsetmacro{\ux}{0.5*\vx/\norm}
      \pgfmathsetmacro{\uy}{0.5*\vy/\norm}
      \pgfmathsetmacro{\uz}{0.5*\vz/\norm}
      
      \addplot3[->, thick, quiver={u=\ux, v=\uy, w=\uz, scale arrows=0.5}, 
                mark=*, mark options={scale=1}, scatter,
                point meta=\E, samples=1]
                coordinates {(\x,\y,\z)};
    }
    \addlegendentry{\((m,v_x,v_y,v_z,x_x,x_y,x_z)\mapsto \frac{1}{2}m (v_x^2+v_y^2+v_z^2)+ mg x_z: g=9.8 \text{m/s}\)}
\end{axis}
\end{tikzpicture}
\caption{An illustration of a $(3,7)$\textbf{-dimensional Voxelized Vector Field} with the kinematic energy function providing the coloring at every voxel.}
\end{figure}
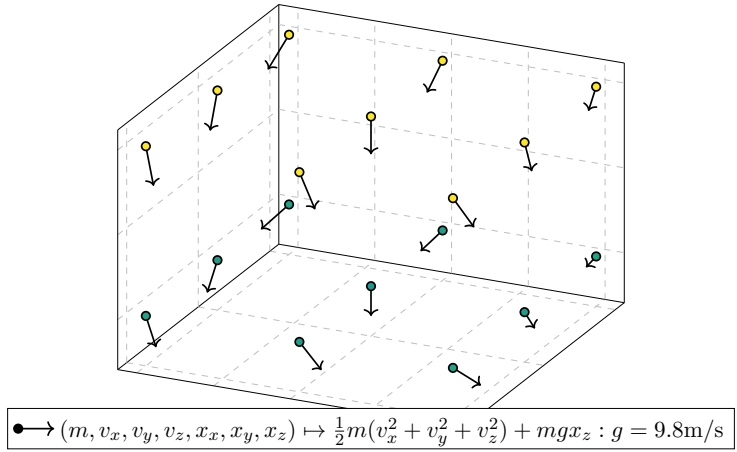
\subsection{Voxelized Vector Field Data }\label{sec: voxelized-vectorfield-data}
The image compression problem was our original motivation for the  \textbf{lossy indexed compression problem}. 
We would like to formulate this problem in terms of an arbitrary indexed set of vectors. 

We call an enumerable set of fixed-length real valued vectors that can be stored, retrieved and managed a \textbf{vector database}. When this enumeration can be multi-indexed, its storage format is a tensor data type, or multi-indexed array of real numbers. This data type includes traditional images~\cite{ahmed2006discrete}, matrices~\cite{strang2022introduction}, as well as voxel data~\cite{lefebvre2006perfect, zhao2020sdrbench} etc.

When the index size scales, the compression problem is known to be hard for tensors.
It is known that compressing a sparse matrix is NP-Hard~\cite{garey2002computers}. It is also known that compressing a 0-1 rectilinear polygon by rectangles is NP-Hard~\cite{garey2002computers}. We would like to investigate the hardness of a closely related compression problem involving multi-indexed vector databases.

In the following, we formulate a \textbf{data type} to represent multi-indexed vector database in the form of a discretized vector field. 

Let $n \in \mathbb{N}$ be the total number of vectors in a vector database $X$. As a vector database, each vector in $X$ can be enumerated  as $X_i \in \mathbb{R}^{D}, i \in [n]$ where $D \in \mathbb{N}$ is the length of each fixed-length vector. This is called \textbf{flat-indexing}. We call the integer $i$ the primary key in the flat-index for the vector record $X_i$.  The flat index up to set bijections is also called the support of $X$, denoted $\textsf{supp}(X)$.

Throughout this exposition, we can assume that the \emph{numerical inputs} $r \in \mathbb{R}$ to all algorithms  are from the rationals. These must have a finite bit string representation. 

Let $\textsf{Bit}(\bullet)$ denote the map from a number to a single bit string representation.  
For a number $r$, let $\textsf{size}(r)$ denote the number of bits returned by $\textsf{Bit}(r)$. 
 We assume that rational numbers  $q<1$ have  $\textbf{size}(q)=\Theta(\log(\frac{1}{q}))$ and that for all $n \in \mathbb{N}$, $\textsf{size}(n)=\Theta(\log(n))$. These include delimiters to denote the start and end of the numerical representation.

We will make the following assumptions on $X$:  
\begin{definition}
  A vector database $X$ is a $(k,D)$-\textbf{dimensional Voxelized Vector Field}, or $(k,D)$-\textsc{Voxelized VF} if the following properties hold:
  \begin{enumerate}
      \item
      (Underlying Vector Field Geometry)
The vector database represents finite samples from an injective vector field $V:\mathbb{R}^k \rightarrow \mathbb{R}^{D}$. Each record $X_i$ of $X$ is a vector $V(x)=X_i \in \mathbb{R}^D$ representing some unique point $x \in \mathbb{R}^k$.
\item
(Compactness of the Range)
    We will assume that each $X_i \in \mathbb{R}^{D}$ belongs to a hypercube denoted as $\Pi_{j=1}^{D} [a_j,b_j] \subseteq \mathbb{R}^{D}$.
\item
\label{ass: voxeldata}(Voxel Domain) 
$X$ is supported on a $k$-dimensional \textbf{voxel grid} $M$, see Appendix Section \ref{sec: more-connected-shapes}. Thus, the index $i \in [n]$ for each $X_i$ is in one-to-one correspondence with a unique coordinate, or \textbf{voxel}, $I(x)=(i_1,...,i_k)$ in $[n_1]\times \cdots \times [n_k]$ where $V(x)=X_i$ for some $x \in \mathbb{R}^k$.
\begin{itemize}
    \item We use the integer $i$ and the integer tuple $I(x)$ interchangeably due to the one-to-one correspondence.
\end{itemize}
  \end{enumerate}
    \end{definition}
    \begin{corollary} (Undecidability of \textsc{LossyIndexedCompression} on \textsc{Voxelized VVF} data) 
    
Let $\rho: \mathbb{R}^D \times \mathbb{R}^D \rightarrow \mathbb{R}_{\geq 0}$ be a nontrivial metric so that \begin{equation}
    \forall x,y \in \mathbb{R}^D: x \neq y\Rightarrow \rho(x,y)>0
\end{equation}
\textsc{Lossy Indexed Compression} with a Kolmogorov decoder and instances $(\varepsilon,(X_1,...,X_n))$ restricted to $(k,D)$\textsc{-Voxelized VF} $X$ with distribution $P(X_i)=\delta_{X_i}$ and $\varepsilon>0$ is undecidable.
\end{corollary}
\begin{proof}
    Since $\rho$ is a nontrivial metric, $B_x(\varepsilon)\subseteq \mathbb{R}^D$ is infinite. The result thus follows by Theorem \ref{thm: undecidable-lossy-compression}.
\end{proof}
We would like to understand the data vectors $X_i=(x_1,...,x_D)$ from $X$ in terms of a multivariate continuous objective function $f: \mathbb{R}^D \rightarrow \mathbb{R}$ of its ``units" $x_j \in \mathbb{R}$. 

We require this function $f$ to be a  \textbf{piecewise multivariate polynomial} in the $D$ variables expressible as a string of symbolic input to an algorithm. This symbolic input should be interpretable as an algebraic expression consisting of a composition of binary sum $(+)$ and binary product $(\cdot)$ arithmetic operations. For example $f(x,y)=x\cdot y+y \cdot y \cdot y$.

Assume the symbolic string representation of a  piecewise multivariate polynomials can be computed in polynomial time in the \textbf{length}, denoted $\textsf{len}(f)$ of the string. We can assume that the degree of $f$ satisfies: $\textsf{deg}(f)\leq \textsf{len}(f)$. For the definition of the degree of a piecewise multivariate polynomial, see the Preliminaries in Section \ref{sec: preliminaries}. The number of bits to represent $f$ satisfies: $\textsf{size}(f)=8\textsf{len}(f)$, assuming an 8 bit ASCII code alphabet. 
\begin{example}
For example, the multivariate polynomial objective function over the vector space of $7$-dimensional vectors could be a kinematic energy function $f: \mathbb{R}^7 \rightarrow \mathbb{R}$, as the following algebraic expression:
\begin{equation}
    f(m,v_x,v_y,v_z,x_x,x_y, x_z) \triangleq \frac{1}{2} \cdot m\cdot v_x \cdot v_x+\frac{1}{2}m\cdot v_y\cdot v_y+\frac{1}{2}m\cdot v_z\cdot v_z+ m \cdot g \cdot x_z ; g\triangleq 9.8 \text{ m/s}
\end{equation}
This expression can be stored as a string of floating points, functional relations, floating point operations, and variable symbols. It can also be evaluated to compute the kinematic energy function in time linear in the string length.
\end{example}

\subsection{Preliminaries}\label{sec: preliminaries}
\textbf{Notation: }
A cartesian product between two sets: 
\begin{equation}
    S_1 \times S_2= \{(s_1,s_2): s_1 \in S_1, s_2 \in S_2\}
\end{equation}
The elements of the cartesian product are called tuples. The cartesian product of a set $S$ over itself $k$ times is written as: 
\begin{equation}
    S^k=\overbrace{(S \times \cdots \times S)}^k
\end{equation}
For any subset $U$ of a cartesian product $S^k$, the $q$-th dimension projection of $U$ is denoted 
\begin{equation}
    (U)_q\triangleq\{s_q \in S: (s_1,...,s_q,...s_k) \in U\}
\end{equation}
A slicing of $U$ in dimension $q \in [k]$ by $x \in S$ is defined as:
\begin{equation}
    \pi_q^x(U)\triangleq \{(s_1,...,x,...,s_k) \in U: (s_1,...,x,...,s_k) \in U\} 
\end{equation}

A cartesian product of a set $S$ can be over another set $I$. If we view each tuple of $S^I$ as an assignment of $i \in I$ to an element of $S$, this forms the set of all functions from $I$ to $S$.

A \textbf{partition} of a set $S$ is defined as a collection of subsets $U_i: U_i \subseteq S, \forall i \in [n]$ so that:
\begin{equation}
    \bigcup_{i\in [n]} U_i=S, U_i\cap U_j=\emptyset, \forall i,j \in [n]
\end{equation}
A \textbf{permutation} on a set $S$ is a bijection from $S$ to itself.

Let $\mathbb{R}$ denote the real numbers. Let $\mathbb{R}_{\geq 0}\triangleq\{x \in \mathbb{R}: x\geq 0\}$. Let $[a,b]\triangleq \{x: a\leq x \leq b\}, a,b \in \mathbb{R}$ denote a closed interval on the reals.
For the other number systems, let $\mathbb{N}\triangleq\{1,...\}$ denote the natural numbers, let $\mathbb{Z}\triangleq\{...,-1,0,1,...\}$ denote the integers, let $\mathbb{Q}\triangleq\{...,\frac{a}{b},...: a,b \in \mathbb{Z}\}$.  We can  similarly define intervals as follows: let $[n] \triangleq\{1,...,n\}, n \in \mathbb{N}, [a,b]\triangleq\{i \in \mathbb{Z}: a\leq i\leq b\}, a,b \in \mathbb{Z}$. We assume the logarithm function is in base 2: $\log_2: \mathbb{R}_{>0} \rightarrow \mathbb{R}$.

A finite \textbf{sequence} of real numbers is denoted by $(x_i)_{i \in  [n]} \in \mathbb{R}^{[n]}$. A \textbf{string} is a sequence of ``symbols." These symbols belong to some common ``alphabet" set. The concatenation of two strings $s_1,s_2$ is denoted $s_1\|s_2$.

A \textbf{(multivariate) polynomial} is a function on $D$ numbers of the form:
\begin{equation}
    (x_1,...,x_D) \mapsto \sum_{i: 0\leq i<T} c_i \Pi_{j=1}^D x_j^{d_{i,j}}, c_i \in \mathbb{R}, d_{i,j} \in \mathbb{N}, T\in \mathbb{N}
\end{equation}
where the numbers $c_i$ are called the coefficients. 
It is called a polynomial if it is in one variable. 
A multivariate polynomial defined on domain $X: X \subseteq \mathbb{R}^D$ is called \textbf{piecewise} if it is defined by $m$ multivariate polynomials $f_i: U_i \rightarrow B$ with $\{U_i\}_{i=1}^m, m \in \mathbb{N}$ forming a partition of $X$. The degree of a piecewise multivariate polynomial $f$ determined by $f_i, i=1,...,m$ is defined as:
\begin{equation}
    \textsf{deg}(f)\triangleq \max_{f_i}\max_{T_p: f_i=\sum_p c_pT_p, c_p \in \mathbb{R}} \sum_{k: T_p= {k=1}^D x_{k}^{j_k}} j_k
\end{equation}
The \textbf{algebraic numbers} are the set: $\tilde{\mathbb{Q}}\triangleq \{x \in \mathbb{R}: p(x)=0 \text{,  $p$ is an integer polynomial}\}$. 
These have a finite symbolic representation. The algebraic numbers satisfy the following:
\begin{equation}
    \mathbb{Q}\subseteq \tilde{\mathbb{Q}} \subseteq \mathbb{R}
\end{equation}

\begin{definition}
A \textbf{graph} is a pair of sets $(V,E)$. The set $V$ consists of nodes and $E \subseteq V \times V$ is called the set of edges.
\end{definition}
\subsubsection{Vector Fields}
A tuple of numbers is a finite sequence of numbers. A real vector is a tuple of real numbers of the same fixed length. The Euclidean space $\mathbb{R}^{k}$ consists of the set of all length $k$ real vectors. The $q$-th entry of a vector $v \in \mathbb{R}^k, q \leq k$ is denoted $v_q \in \mathbb{R}$.

We can map from $\mathbb{R}^k$ to $\mathbb{R}^D$. These maps are called vector fields, which we define below:
\begin{definition}
A \textbf{vector field} is defined as a map from $\mathbb{R}^k$ to a $D$-dimensional vector space $\mathbb{R}^D$. This is denoted $V: \mathbb{R}^k \rightarrow \mathbb{R}^D$.
\end{definition}
A vector field is only a map and does not have to have any algebraic properties.
\subsubsection{Clustering in Metric Spaces}
Let there be a metric $d: \mathbb{R}^{D} \times \mathbb{R}^{D} \rightarrow \mathbb{R}$ on $\mathbb{R}^D$, see Section \ref{sec: metric-spaces} in the Appendix. Define the diameter function $\textsf{diam}: 2^{\mathbb{R}^{D}} \rightarrow \mathbb{R}$ as follows:
\begin{equation}
    \textsf{diam}_d(S)\triangleq \sup_{x,y \in S}d(x,y)
\end{equation}
\begin{definition}
Let $I: J\rightarrow X$ be a bijection from an indexing set $J$ to the set of points $X \subseteq \mathbb{R}^{D}$.

An \textbf{Indexed Vietoris-Rips complex} at threshold $r$ on the indexed set of points  $X$ equipped with a metric $d$ is defined as the following collection of index sets, called index \emph{clusters}:
\begin{equation}
    \textsf{iVR}^d_r(X)\triangleq \{S \subseteq J: \textsf{diam}_d(\{I(j) \in X: j \in S\})\leq r\}
\end{equation}
\end{definition}
\subsection{Voxelized Vector Field Compression}\label{sec: VVFcomp-prob}
In the \textbf{lossy indexed compression problem}, we can define a distortion function $d_f$ as the Euclidean distance between the energy, defined by the piecewise multivariate polynomial $f$, of a vector and its reconstruction:
\begin{equation}
    d_f(X_i,\hat{X}_i)\triangleq \lvert f(X_i)-f(\hat{X}_i)\rvert_2
\end{equation}
This distortion function naturally measures the change in energy between a vector and its reconstruction at a common index $i$. We call this the \textbf{energy gap} distortion function.

We would like to define computable encoders and decoders for the \textbf{lossy indexed  compression problem}. We will thus understand the encoders and decoders in terms of index consistent compression and decompression algorithms respectively. As discussed in Lemma \ref{prop: power-expand}, the output codeword involves clustering assignments $i \mapsto X_i$ atmost at the granularity of $n$ individual assignments. 
We can thus define this compression-decompression process in the following:
\begin{definition}
A \textbf{clustered indexed algorithmic compression scheme} $\mathcal{M}$ over $(k,D)$-voxelized vector fields is a pair of algorithms $(\mathcal{C},\mathcal{D})$, called the \textcolor{purple}{compression} and \textcolor{orange}{decompression} algorithms. They are governed by the following \emph{three step protocol}:
\begin{enumerate}
\item[] Given a $(k,D)$-voxelized vector field $X$, a piecewise multivariate polynomial objective function $f: \mathbb{R}^D \rightarrow \mathbb{R}$, and an error bound $\varepsilon$. 
    \item The \textcolor{purple}{compression} algorithm $\mathcal{C}$ takes a triple $(X,f,\varepsilon)$ as input and outputs a codeword $\mathcal{C}(X,f,\varepsilon)$ from the codebook $\{1,...,2^{R_n}\}$ with
    \begin{itemize}
        \item $\mathcal{C}(X,f,\varepsilon)$  a function of the following clustering:
        \begin{equation}
            \{\{(i,X_i): i \in S\}: S \in \mathcal{S}(X,f,\varepsilon)\}; \mathcal{S}(X,f,\varepsilon)\subseteq 2^{\textsf{supp}(X)}
        \end{equation}
        where $\mathcal{S}(X,f,\varepsilon)$ is called an \textbf{index clustering} of $\textsf{supp}(X)$.
    \end{itemize}
    \item $\mathcal{C}(X,f,\varepsilon)$  is sent over a ``noiseless channel"~\cite{shannon1948mathematical} from $\mathcal{C}$ to $\mathcal{D}$. 
    \item $\mathcal{C}(X,f,\varepsilon)$ is then \textcolor{orange}{decompressed} by the decompression algorithm $\mathcal{D}$ into a $(k,D)$-voxelized vector field.
    \begin{equation}
    \hat{X}=\mathcal{D}\circ \mathcal{C}(X,f,\varepsilon)
    \end{equation}
\end{enumerate}
where the two algorithms $\mathcal{C},\mathcal{D}$ must satisfy: 
\begin{itemize}
    \item (Index Consistency) $(X,\hat{X})$ are \textbf{index consistent} for distortion $d_f$ at error bound $\varepsilon$.  
    \item  (No Learning) $\mathcal{M}$ cannot self update from the data $(X,f,\varepsilon)$.
\end{itemize}
\end{definition}
\begin{remark}
    This definition is a direct instantiation of the power set expansion encoder of Lemma \ref{prop: power-expand}. The index subsets $S \subseteq \{1,...,n\}$, or cluster indices belong to $\mathcal{S}(X,f,\varepsilon)$.
\end{remark}
\begin{remark}
     The lossy indexed compression problem is \emph{inherently a learning problem}. This is because the encoder-decoder pair can use the input data $(X,f,\varepsilon)$ in the optimization problem. 
     
     A clustered indexed algorithmic compression scheme  cannot learn over data. This means that the program size of the compression and decompression algorithms $\mathcal{C}$ and $\mathcal{D}$ are $O(1)$.

     A sufficient condition for the no learning criterion is an \emph{asymmetric causal relationship} from the encoder to the decoder. A learnable encoder would then no longer be able to update its code over samplings $(x_1,...,x_n) \sim P(X_1,...,X_n)$ since it cannot check for index consistency from the decoder.
\end{remark}
\begin{remark}
This definition is a generalized variant of a labeled-compression scheme~\cite{littlestone1986relating}  from classical learning theory. Let:
\begin{equation}
    \mathcal{S}(X,f,\varepsilon):=\{\{i\}: i \in  S, S \subseteq \textsf{supp}(X)\}
\end{equation}
be a collection of singletons.

The ``labeled samples" of a labeled-compression scheme for $X$ are of the form $(i,X_i)$, $i$ being a primary key on the flat index on $X$ and the ``kernel" of the labeled samples is the codeword $\mathcal{C}(X,f,\varepsilon)$ determined by the subset of labeled samples $(i,X_i): \forall \{i\} \in \mathcal{S}(X,f,\varepsilon)$. 
\end{remark}
As discussed in the introduction, the goal is to minimize the rate $R_n$. This can be converted to a maximization of the compression ratio of the codeword $\mathcal{C}(X,f,\varepsilon)$. This is defined as:
\begin{equation}
    \textsf{CR}(\mathcal{C}(X,f,\varepsilon),X)\triangleq \frac{\textsf{size}(X)}{\textsf{size}(\mathcal{C}(X,f,\varepsilon))}
\end{equation}
where the ``size" of a database $X$ is measured by the number of bits needed to store it with flat-indexing. 

Similarly, the ``size" of a codeword is the number of bits to store it as part of the codebook $\{1,...,2^{R_n}\}$ of the outputs of the compression algorithm $\mathcal{C}$. This is the (not necessarily optimal) rate $R_n$ of the clustered indexed algorithmic compression scheme. We can check that $\textsf{CR}(\mathcal{C}(X,f,\varepsilon),X)$ is maximized when ${\textsf{size}(\mathcal{C}(X,f,\varepsilon))}$ is minimized since ${\textsf{size}(X)}$ does not change. 

\subsection{Voxelized Vector Field Compression as a Box Covering Problem}\label{sec: vvfcompression-box-covering}
Let us understand a clustered indexed compression scheme through the \textbf{energy gap}  distortion function $d_f$.

We make the following observation about the \textbf{energy gap} distortion function $d_f$:
\begin{observation}\label{obs: VR-df}
For a $(k,D)$-dimensional voxelized vector field $X$ with objective function $f: \mathbb{R}^D \rightarrow \mathbb{R}$, the \textbf{energy gap} distortion function $d_f: \mathbb{R}^D\times \mathbb{R}^D \rightarrow \mathbb{R}_{\geq 0}$ is a metric and the Indexed Vietoris-Rips complex 
    $\textsf{iVR}^{d_f}_{r}(X)$ over metric $d_f$ is well defined for any $r < \infty$.
\end{observation}
\begin{proof}
    The map $(X_i,X_j) \mapsto \lvert f(X_i)-f(X_j)\rvert_2$ is a metric. Furthermore, no two vectors are equivalent in vector database $X$ by definition of a $(k,D)$\text{-dimensional voxelized vector field}. Thus the Indexed Vietoris-Rips complex can be constructed for any diameter threshold.
\end{proof}
In Observation \ref{obs: VR-df}, we have that $\textsf{iVR}^{d_f}_r(X)$ for any $r<\infty$ is well defined. According to the definition of an Indexed Vietoris-Rips Complex, the clusters in this case must consist of a collection of connected components on $\textsf{supp}(X)$. These consist of orthogonal polyhedra, boxes etc. see Section \ref{sec: Connected-Shapes} for definitions of connected shapes on voxel grids.

Let us specify a restriction on the power expansion of the compression algorithm  belonging to a clustered indexed algorithmic compression scheme. This restriction chooses exact closed forms for the maps $\Lambda,\kappa,\phi$. In particular, these are chosen to restrict the compression algorithm to an index clustering. These restrictions isolate the clustering for the assignments $i \mapsto X_i$ for $i \in \textsf{supp}(X)$. 
\begin{definition}
    A collection of boxes $\mathcal{R} \subseteq 2^{\textsf{supp}(X)}$ \textbf{covers} the indices of a $(k,D)$-dimensional voxelized vector field $X$ if for every $i \in \textsf{supp}(X)$, there is a $R \in \mathcal{R}$ that has $i \in R$.

    Such a collection of boxes is called a \textbf{box cover}.
\end{definition}
Let us define a compression algorithm that outputs an encoding of a meaningful box cover:
\begin{definition}
    A \textbf{box compression algorithm} $\mathcal{C}$ on a $(k,D)$-dimensional voxelized vector field $X$, a piecewise multivariate polynomial $f$ and error bound $\varepsilon$ computes the following index consistent encoder on distortion $d_f$ with the following \emph{power set expansion}:
    \begin{equation}
        \mathcal{C}(X,f,\varepsilon):=\textsf{BitConcatWParams}(\{(\textsf{Encode}(R),\textsf{Summarize}_R(X,f,\varepsilon)): R \in \mathcal{S}(X,f,\varepsilon)\})
    \end{equation}
    where $\mathcal{S}(X,f,\varepsilon)$ is a \textbf{box cover} over $\textsf{supp}(X)$ and:
    \begin{enumerate}
        \item $\textsf{Encode}(R):= (c_1(R), c_2(R))$ where $c_1(R)$ and $c_2(R)$ are the coordinates in $\mathbb{Z}^k$ of the two corners of $R$, see Proposition \ref{prop: 2-corner-box-appendix}. 
        \item $\textsf{Summarize}_R$ maps the cluster of vectors $\{X_i \in \mathbb{R}^D: i \in R\}$ to a summary $\hat{f}(R) \in {\mathbb{Q}}$ satisfying:
         \begin{equation}
             \lvert f(X_i)-\hat{f}(R) \rvert_2 \leq \varepsilon^*, \forall i \in R
         \end{equation}
         by computing $O(\textsf{len}(f))$ number of algebraic operations on the vectors $X_{i} \in \mathbb{R}^D, i \in R$.
        \item $\textsf{BitConcatWParams}:=\textsf{RecTupleBit}\circ \textsf{ConcatParams}$ \text{ where: }
        \begin{equation}
            \textsf{ConcatParams}(\{(R,\bullet_R): R \in \mathcal{S}(X,f,\varepsilon)\}):=   (\overbrace{(\cdots ,( \textsf{Encode}(R),\bullet_R), \cdots )_{R \in \mathcal{S}(X,f,\varepsilon)}}^{\text{any order}}, (\varepsilon^*,\varepsilon, f))
        \end{equation}
        \begin{equation}
            \textsf{RecTupleBit}(L_1,...,L_n):=\begin{cases}
                (\textsf{RetTupleBit}(L_1),...,\textsf{RecTupleBit}(L_n)) & \text{ If }n>1\\
                (\widetilde{\textsf{Bit}}(x_1),...,\widetilde{\textsf{Bit}}(x_n)) & \text{ Else}
            \end{cases} 
        \end{equation}
        with $\widetilde{\textsf{Bit}}(\bullet):=\textsf{Bit}(\bullet)\| 11$ and the comma and the two parentheses are encoded by the length $2$ bit strings $00,01,10$ respectively. These encodings are concatenated all together to form a single bit string.
    \end{enumerate}
    \end{definition}
    \begin{definition}
    We call the tuple before composing with \textsf{RecTupleBit} in a box compression algorithm on input $(X,f,\varepsilon)$: \begin{equation}
        (\overbrace{(\cdots ,( \textsf{Encode}(R),\bullet_R), \cdots )_{R \in \mathcal{S}(X,f,\varepsilon)}}^{\text{any order}}, (\varepsilon^*,\varepsilon, f)) 
    \end{equation} as a \textbf{box cover with box summaries data structure}. 
\end{definition}
Since \textsf{RecTupleBit} is injective, we will interchangeably say that the output of a box  compression algorithm is a box cover with box summaries or a bit string code word in $[2^{R_n}]$
We have the following closed form for the \textbf{size} of a \textbf{box cover with summaries} data structure, which is the length of the bit string returned by a box compression:
\begin{proposition}\label{prop: size-boxcover-compression}
\begin{equation}
    \textsf{size}(\mathcal{C}(X,f,\varepsilon))= \lvert \mathcal{S}(X,f,\varepsilon)\rvert(10r+4\textsf{size}(n)+2v(\textsf{size}(f),n,\textsf{size}(X))+5r+\textsf{size}(\varepsilon^*)+\textsf{size}(\varepsilon)+\textsf{size}(f)
    \end{equation}
    where $r$ is the constant number of bits to represent any delimiter. The number $v$ is the maximum number of bits to represent any vector entry in $X$ using a piecewise multivariate polynomial $f$ for distortion $d_f$ at a lossy bound of $\varepsilon>0$. 
\end{proposition}
\begin{proof}
    See Appendix Proposition \ref{prop: appendix-sizes}.
\end{proof}
We can now identify the following problem:
\begin{tcolorbox}
\begin{problem}\label{prob: lossycompression}$(k,D)$-\text{Lossy Voxelized Vector Field Box Covering Compression Problem}

$((k,D)\textsc{-RectLossyVVFCompression})$
\begin{itemize}
        \item[] \textbf{Input: } For any $(k,D)$-dimensional voxelized vector field $X$ with $n$ vectors, a piecewise multivariate polynomial objective function $f: \mathbb{R}^D \rightarrow \mathbb{R}$, and an $ \varepsilon>0: \varepsilon<1$. (\emph{All real numbers have  finite length bit string representations}
        \item[] \textbf{Output: } A box cover with box summaries data structure $\mathcal{C}(X,f,\varepsilon)$ maximizing:
        \begin{equation}
            \textsf{CR}(\mathcal{C}(X,f,\varepsilon),X)
        \end{equation}
        for some clustered indexed algorithmic compression scheme $\mathcal{M}=(\mathcal{C},\mathcal{D})$ 
    \end{itemize}
    \end{problem}
\end{tcolorbox}
\begin{remark}
    We can convert $(k,D)\textsc{-RectLossyVVFCompression}$ to its \textbf{decision version} by introducing an additional input $K \in \mathbb{R}$ and ask for a feasible clustered indexed algorithmic compression scheme $\mathcal{M}$  so that $\textsf{CR}(\mathcal{C}(X,f,\varepsilon),X)\geq \frac{\textsf{size}(X)}{K}$.
\end{remark}
\begin{theorem}\label{thm: decidable-kd-rectlossyvvfcompression}
    $(k,D)$\textsc{-RectLossyVVFCompression} is decidable.
\end{theorem}
\begin{proof}
    For any $(X,f,\varepsilon) \in I_{(k,D)\textsc{-RectVVFC}}$ since $X$ consists of $n$ vectors, $f$ has finite size $\textsf{size}(f)$, and $\textsf{size}(\varepsilon)$ is finite, there can only be $O(2^n(2n)^{\textsf{size}(f)})$ (a finite number) many possible solutions. 
    
    This is the cartesian product of the power set on $n$ of size $2^n$ and the set of possible polynomials formed by $O(\textsf{len}(f))$ algebraic combinations of the vectors $X_i, i \in [n]$ through the $(+), (\cdot)$ arithmetic operations:
    \begin{equation}
        O((2n)^{\textsf{size}(f)})=\lvert (\{+,\cdot\}\times \{X_i: i\in [n]\})^{O(\textsf{size}(f))}\rvert 
    \end{equation}
    Thus, the problem is decidable.
\end{proof}
We have the following relationship between $(k,D)\textsc{-RectLossyVVFCompression}$ and  \textsc{LossyIndexedCompression}:
\begin{proposition}\label{prop: sublang-rectvvf-lossyindcomp}
    Let $f:\mathbb{R}^D\rightarrow \mathbb{R}$ be a piecewise multivariate polynomial. The language 
    \begin{equation}
    \begin{split}
        L_{(k,D)\textsc{-RectVVFC}}=\{((\varepsilon,(X_1,...,X_n)),K): ((X,f,\varepsilon),K)\text{ an instance of }\\(k,D)\textsc{-RectLossyVVFCompression}: \textsf{CR}(\mathcal{C}(X,f,\varepsilon)\geq \frac{\textsf{size}(X)}{K} \} 
        \end{split}
    \end{equation} 
    associated to the decision version of $(k,D)\textsc{-RectLossyVVFCompression}$ is a sublanguage of \begin{equation}
        L_{\textsc{LossyIndComp}}=\{((\varepsilon,(X_1,...,X_n)),K): R_n^*\leq K\}
    \end{equation}
    associated to the decision version of $\textsc{LossyIndexedCompression}$ for distortion $d_f: \mathcal{X} \times \mathcal{X} \rightarrow \mathbb{R}_{\geq 0}$.
\end{proposition}
\begin{proof}
    We can check that: 
    \begin{equation}
         L_{(k,D)\textsc{-RectVVFC}}\subseteq L_{\textsc{LossyIndComp}}
    \end{equation}
    since a clustered indexed algorithmic compression scheme is a feasible  encoder-decoder pair for \textsc{LossyIndComp}.
\end{proof}
The globally optimal solution may not be a box cover with box summaries. This can be seen with Lemma \ref{prop: power-expand}. A box cover with box summaries data structure does not have a post-encoding $\Lambda$ of the box clusters. Thus, it is not fully expressive of all possible encoders. Nonetheless, this problem isolates the need to cluster jointly compressible assignments $i \mapsto X_i$. 

\textbf{Why Boxes?} 
Of course the compression algorithm $\mathcal{C}$ can output more complicated shapes beyond boxes. In \cite{reckhow1987covering}, general coverings of point sets by polygons on the plane are investigated and encoded by ``signatures." These can also be used as an encoding for $2$ dimensional voxel grid covers. We would like to, however, investigate the simplest possible shape of a box. This requires a simple $O(1)$ encoding. A box can be encoded by only two corner points, see Appendix Section \ref{sec: more-connected-shapes}. For more on shapes in voxel grids, see Section \ref{sec: Connected-Shapes} and the Appendix Section \ref{sec: more-connected-shapes}.
\section{Outline of Main Results}
We summarize here the main results of the paper:

We have shown that the problem of \textsc{LossyIndexedCompression} is undecidable under the conditions in Theorem \ref{thm: undecidable-lossy-compression}, restated here:
\begin{theorem}
    (Undecidability) 
        If there is a $\varepsilon>0$ and a $x \in \mathcal{X}$ where the set $B_x(\varepsilon)=\{\hat{x} \in \mathcal{X}: \rho(x,\hat{x})\leq \varepsilon\}$ is infinite, 
        then the \textbf{lossy indexed compression} problem with Kolmogorov decoder is undecidable.
\end{theorem}
For voxelized vector fields, we have shown the following undecidability result:
\begin{corollary}
     (Undecidability of \textsc{LossyIndexedCompression} on \textsc{Voxelized VVF} data) 
    
Let $\rho: \mathbb{R}^D \times \mathbb{R}^D \rightarrow \mathbb{R}_{\geq 0}$ be a nontrivial metric so that \begin{equation}
    \forall x,y \in \mathbb{R}^D: x \neq y\Rightarrow \rho(x,y)>0
\end{equation}
\textsc{LossyIndexedCompression} with a Kolmogorov decoder and instances $(\varepsilon,(X_1,...,X_n))$ restricted to $(k,D)$\textsc{-Voxelized VF} $X$ with distribution $P(X_i)=\delta_{X_i}$ and $\varepsilon>0$ is undecidable.
\end{corollary}
We thus introduced a decidable sublanguage to the \textsc{LossyIndexedCompression}, the   $(k,D)$\textsc{-RectLossyVVFCompression} problem. The decidability can be stated here:
\begin{theorem}
    $(k,D)$\textsc{-RectLossyVVFCompression} is decidable.
\end{theorem}
The remainder of the paper is about the $(k,D)$\textsc{-RectLossyVVFCompression} problem. 

We show a polynomial time tractability result about decompression for the\\ $(k,D)$\textsc{-RectLossyVVFCompression} problem and two hardness results about compression. We will use terminology from complexity theory, see Section \ref{sec: algorithms-appendix} in the Appendix.

The tractability result is shown in Section \ref{sec: decompression-easy} and can be summarized as the following:
\begin{theorem}
    For any $\mathcal{C}$ box compression algorithm solution for $(k,D)$\textsc{-RectLossyVVFCompression}.
    
    For any  $(X,f,\varepsilon)$ instance of $\mathcal{C}$. We have the \textbf{box cover with box summaries} data structure on $(X,f,\varepsilon)$:
    \begin{equation}
        \mathcal{C}(X,f,\varepsilon)=(\{(R,\hat{f}_R): R \in \mathcal{S}(X,f,\varepsilon)\}, (\varepsilon^*,\varepsilon,f))
    \end{equation}
    
    Then there is a decompression algorithm $\mathcal{D}$ running in time:
    \begin{equation}
       O(\lvert \mathcal{S}(X,f,\varepsilon)\rvert (\textsf{size}(f)\log(\frac{1}{\varepsilon^*}))^{O(D^2)})
    \end{equation}
    so that:
    \begin{equation}
    \lvert f(X_i)-f((\mathcal{D}\circ \mathcal{C}(X,f,\varepsilon))_i)\rvert_2 \leq \varepsilon, \forall i \in \textsf{supp}(X)
\end{equation}
\end{theorem}
We also show in Section \ref{sec: decompression-easy} that there is no obstruction to knowing the space of all possible feasible solutions of the $(k,D)$-\textsc{RectLossyVVFCompression} in polynomial time. 

This is stated in Lemma \ref{lemma: boxcoversum-sol}. It roughly states that it is always possible to find a number $\varepsilon^{*}\in \mathbb{Q}$ in polynomial time with size polynomial in the input size so that any possible solution is contained in an Indexed Vietoris-Rips complex of threshold determined by $\varepsilon^{*} \in \mathbb{Q}$.  

This Lemma along with the tractability result on decompression, shows that any possible compressed data structure can be decompressed in polynomial time. Thus, the only obstruction to obtaining a polynomial time tractable clustered indexed compression scheme is in solving for an optimal clustering.

The first hardness result is about the obstruction to finding this optimal compressed data structure. This is shown in Section \ref{sec: np-hard} and stated here:
\begin{theorem}
    $(k,D)$\textsc{-RectLossyVVFCompression} is NP-hard if $k\geq 2, D\geq 1$.
\end{theorem}
It is natural to desire an approximation for this NP-Hard problem. Our second result shows that even approximately solving the $(k,D)$\textsc{-RectLossyVVFCompression} problem is hard.  This is shown in Section \ref{sec: apx-hard} and stated here:
\begin{theorem}    $(k,D)$\textsc{-RectLossyVVFCompression} is APX-Hard for $k,D\geq 2$
\end{theorem}
Assuming $P\neq NP$, our two hardness results show that when $k,D \geq 2$ there can be no exact polynomial time algorithm nor can there even be a PTAS  approximation algorithm for the $(k,D)$\textsc{-RectLossyVVFCompression} problem.

Our proof outline for the three hardness results are given in the following diagram: 
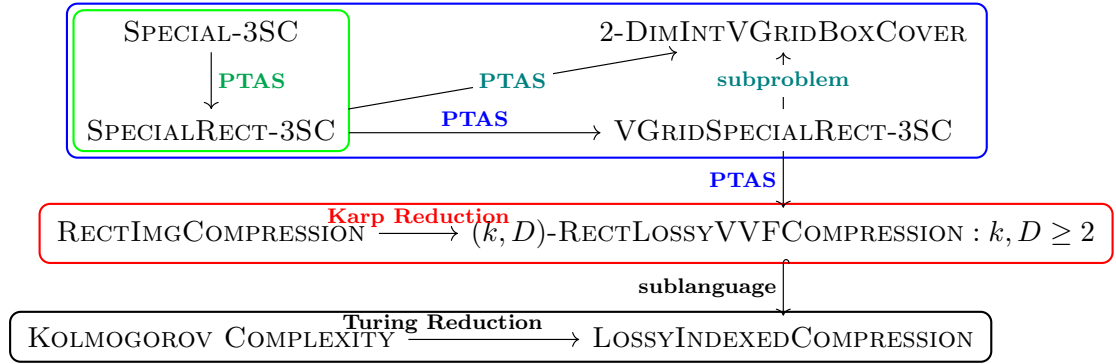
\[
\begin{tikzcd}
[
  scale=0.2,
  every node/.style={transform shape},
  row sep=2.0em, column sep=1.75em,
  remember picture,
  execute at end picture={
    \begin{scope}[on background layer]
      \node[draw, color=green, thick, rounded corners, fit=(special-3sc)(specialrect-3sc), inner sep=0.01pt, label=left:{}] {};
      \node[draw, thick, rounded corners, fit=(kolmogorov-complexity)(lossyindexedcompression), inner sep=1.5pt, label=left:{}] {};
      \node[draw, color=red, thick, rounded corners, fit=(rectimgcompression)(kd-rectlossyvvfcompression), inner sep=1.5pt, label=left:{}] {};
      \node[draw, color=blue, thick, rounded corners, fit=(special-3sc)(specialrect-3sc)(2-dimintvgridboxcover)(vgridspecialrect-3sc), inner sep=2.5pt, label=left:{}] {};
    \end{scope}
    }]
	&& |[alias=special-3sc]|{\textsc{Special-3SC} } & |[alias=2-dimintvgridboxcover]|{2\textsc{-DimIntVGridBoxCover}} \\
	&& |[alias=specialrect-3sc]|{\textsc{SpecialRect-3SC} } & |[alias=vgridspecialrect-3sc]|\textsc{VGridSpecialRect-$3$SC} \\
	{} && |[alias=rectimgcompression]|{\textsc{RectImgCompression}} & |[alias=kd-rectlossyvvfcompression]|{(k,D)\textsc{-RectLossyVVFCompression}: k,D\geq2} \\
	&& |[alias=kolmogorov-complexity]|\textsc{Kolmogorov Complexity} & |[alias=lossyindexedcompression]|\textsc{LossyIndexedCompression}
	\arrow["{{{{{\textcolor{Green}{\textbf{PTAS}}} }}}}", from=1-3, to=2-3]
	\arrow["{{{{{\textcolor{teal}{\textbf{PTAS}}}}}}}"{description}, from=2-3, to=1-4]
	\arrow["{{{{{\textcolor{blue}{\textbf{PTAS}}}}}}}", from=2-3, to=2-4]
	\arrow["{{{{{\textcolor{teal}{\textbf{subproblem}}}}}}}"{description}, from=2-4, to=1-4]
	\arrow["{{{{{\textcolor{blue}{\textbf{PTAS}}}}}}}"', from=2-4, to=3-4]
	\arrow["{{{{\textcolor{red}{\textbf{Karp Reduction}}}}}}", from=3-3, to=3-4]
	\arrow["{{{\textbf{sublanguage}}}}"', hook, from=3-4, to=4-4]
	\arrow["{{{\textbf{Turing Reduction}}}}"{pos=0.2}, from=4-3, to=4-4]
\end{tikzcd}
\]
The \textbf{undecidability} of the \textsc{LossyIndexedCompression} problem is shown in Theorem \ref{thm: undecidable-lossy-compression}. Its relationship to the $(k,D)$\textsc{-RectLossyVVFCompression} problem is proven in Proposition \ref{prop: sublang-rectvvf-lossyindcomp}. The decidability of $(k,D)$\textsc{-RectLossyVVFCompression} is proven in Theorem \ref{thm: decidable-kd-rectlossyvvfcompression}.

The Karp reduction, denoted in \textcolor{red}{\textbf{red}}, from \textsc{RectImgCompresion} to \\$(k,D)$\textsc{-RectLossyVVFCompression} for $k,D \geq 2$ is proven in Section \ref{sec: np-hard} in Theorem \ref{thm: np-hard-wproof}.

The PTAS reductions from $3\textsc{VC}$ to \textsc{SpecialRect-$3$SC} are known results from~\cite{chan2014exact}. These are denoted in \textcolor{Green}{\textbf{green}}. This is discussed at the beginning of Section \ref{sec: apx-hard}. We then proceed with  a sequence of PTAS reductions from \textsc{SpecialRect-$3$SC} to $(k,D)$\textsc{-RectLossyVVFCompression}, denoted in \textcolor{blue}{\textbf{blue}}. These are proven in Section \ref{sec: apx-hard} from Theorem \ref{thm: R2ZPTAS-reduction} to Theorem \ref{thm: vgrid2rectlossyvfcomp}. The subproblem $2$\textsc{-DimIntVGridBoxCover} is a natural consequence of Theorem \ref{thm: R2ZPTAS-reduction}. The connecting reductions are denoted with  \textcolor{teal}{\textbf{teal}}.

\section{Decompressing is Easy}\label{sec: decompression-easy}
The  data structure $\mathcal{C}(X,f,\varepsilon)$ is a box cover over the indices of a $(k,D)$-dimensional voxelized vector field. Our objective in the $(k,D)$\textsc{-RectLossyVVFCompression} is to maximize the compression ratio. In order to achieve this, we must find redundant indices up to the inequality constraint from Problem \ref{prob: lossycompression}. This can be achieved by finding  large clusters of vectors $X_i$ whose indices can be jointly summarized up to the lossy constraint. 

Let us study the space of all possible clusters formed by the metric inequality constraint of Problem \ref{prob: lossycompression}. An Indexed Vietoris-Rips Complex is defined by such an inequality constraint. The following observation says that any clustering of a vector database $X$ up to a distortion bound of $\varepsilon^*$ is contained in some  Indexed Vietoris-Rips Complex of threshold $2\varepsilon^*$.\begin{proposition}\label{prop: sol-2epsilon}
   Let $\varepsilon^*>0$ and let there be a $(k,D)$-dimensional voxelized vector field $X$ with objective function $f: \mathbb{R}^D \rightarrow \mathbb{R}$. 
   
   For any cluster $C \in \textsf{iVR}^{d_f}_{2\varepsilon^*}(X)$, each vector $X_i: i  \in C$ can be replaced by the \textbf{\textsf{mid-range}} summary value:
\begin{equation}
 \textbf{\textsf{mid-range}}(f,C)\triangleq \frac{\max_{i \in  C}f(X_i)+\min_{i \in C}f(X_i)}{2}
\end{equation} on $C$ without affecting the $\varepsilon^*$ distortion bound: 
\begin{equation}
    \lvert f(X_i)-\textbf{\textsf{mid-range}}(f,C) \rvert_2 \leq \varepsilon^*, \forall i \in C
\end{equation}
\end{proposition}

\begin{proof}
    By Observation \ref{obs: VR-df} the Indexed Vietoris-Rips Complex $\textsf{iVR}^{d_f}_{2\varepsilon^{*}}(X)$ is well defined.
    
Certainly, for any $C \in \textsf{iVR}^{d_f}_{2\varepsilon^*}(X)$
\begin{subequations}
\begin{equation}
\lvert \textbf{\textsf{mid-range}}(f,C)-f(X_i)\rvert_2 = \lvert \frac{(\max_{j \in C} f(X_j)+\min_{j \in C} f(X_j))}{2}-f(X_i)\rvert_2 \end{equation}
\begin{equation}\leq \lvert \frac{(\max_{j \in C} f(X_j)+\min_{j \in C} f(X_j))}{2}-\min_{j \in C} f(X_j) \rvert_2 
\end{equation}\begin{equation}=\lvert \frac{\max_{j \in C} f(X_j)-\min_{j \in C} f(X_j)}{2}\rvert_2 \leq \varepsilon^*, \forall i \in C
\end{equation}
\end{subequations}
\end{proof}
\begin{proposition}\label{prop: max-VR-threshold}
For the conditions of Problem \ref{prob: lossycompression}, the maximum threshold $r$ to determine clusters $C \in \textsf{iVR}^{d_f}_{r}(X)$ where  $\lvert f(X_i)-\hat{f}(C)\rvert_2 \leq \varepsilon^*$ for some $\hat{f}(C) \in \mathbb{R}$ is $r= 2\varepsilon^*$.
\end{proposition}
\begin{proof}
    We prove by contradiction. 
    
    Suppose $r>2\varepsilon^*$. Let there be some $(k,D)$-dimensional voxelized vector field $X$ where there are two vectors $X_i,X_j: i,j \in  C$ for some $C \in \textsf{iVR}^{d_f}_{r}(X)$  where 
    \begin{equation}
        \lvert f(X_i)-f(X_j)\rvert_2=r > 2\varepsilon^*.
    \end{equation}
    If we can find a single vector $\hat{f}(C) \in \mathbb{R}$ to summarize cluster $C$ under the lossy constraint
    \begin{equation}
        \lvert \hat{f}(C)-f(X_i)\rvert_2 \leq \varepsilon^*,
    \end{equation}
    then we show that there is a contradiction. 
    
    If $\lvert \hat{f}(C)-f(X_i)\rvert_2 \leq \varepsilon^*$ and $\lvert \hat{f}(C)-f(X_j)\rvert_2 \leq \varepsilon^*$, then:
    \begin{equation}
        \lvert f(X_i)-f(X_j)\rvert_2 \leq \lvert \hat{f}(C)-f(X_i)\rvert_2+ \lvert \hat{f}(C)-f(X_j)\rvert_2 \leq 2\varepsilon^*,
    \end{equation}
    contradicting that $X_i,X_j$ are witness to:
    \begin{equation}
        \lvert f(X_i)-f(X_j)\rvert_2=r > 2\varepsilon^*
    \end{equation}
\end{proof}
\begin{proposition}\label{prop: maximum-epsilon}
    For a \textbf{box compression algorithm} $\mathcal{C}$ with:
    \begin{equation}
        \mathcal{C}(X,f,\varepsilon)=(\{(R,\bullet_R): R \in \mathcal{S}(X,f,\varepsilon)\}, (\varepsilon^*,\varepsilon,f))
    \end{equation}
    we must have:
    \begin{equation}
        \mathcal{S}(X,f,\varepsilon)\subseteq \textsf{iVR}_{2\varepsilon}^{d_f}(X)
    \end{equation}
\end{proposition}
\begin{proof}
We know that $\varepsilon^*<\varepsilon$. Thus, by Proposition \ref{prop: max-VR-threshold}, with $\varepsilon^*$ set to $\varepsilon$, we must have $2\varepsilon$ acting as the upper bound on all Indexed Vietoris-Rips thresholds determined by all $\varepsilon^*<\varepsilon$ .
\end{proof}
   Let \begin{equation}\label{eq: df*}
    D(X,f,2\varepsilon)\triangleq\max_{i,j \in \textsf{supp}(X): d_f(X_i,X_j) <  2\varepsilon}d_f(X_i,X_j)
\end{equation} 
be the largest pairwise distance with respect to $d_f$ on $X$ less than $2\varepsilon$.
\begin{observation}\label{obs: dstar-complexity}
    The number $D(X,f,2\varepsilon) \in \mathbb{Q}$ can be computed in time $O(n^2\log(n)+\textsf{size}(f)n^2)$.
\end{observation}
\begin{proof}
    The time complexity is determined by sorting the set of all pairwise distances on $X$ with respect to $d_f$. To determine all pairwise distance requires  $O(\textsf{size}(f)n^2)$ amount of time.
\end{proof}

We have the following guarantee about the solution space for $(k,D)$\textsc{-RectLossyVVFCompression}.
\begin{lemma}\label{lemma: boxcoversum-sol}
For any instance $(X,f,\varepsilon)$ of  $(k,D)$\textsc{-RectLossyVVFCompression}. 
Let: 
\begin{equation}
\varepsilon^{*}:=\frac{D(X,f,2\varepsilon)}{2}
 \end{equation}
 where $D(X,f,2\varepsilon)$ is as in Equation \ref{eq: df*}.
 
The codeword $\mathcal{C}^*(X,f,\varepsilon)$ for the optimal algorithm $\mathcal{C}^*$ to the $(k,D)$\textsc{-RectLossyVVFCompression}  problem has: 
\begin{equation}
    \mathcal{S}^*(X,f,\varepsilon) \subseteq \textsf{iVR}^{d_f}_{2\varepsilon^{*}}(X)
\end{equation}
where $\varepsilon^{*}$ satisfies:
\begin{enumerate}
    \item (Feasibility) $\varepsilon^{*} <\varepsilon$
    \item (Maximal) $\textsf{iVR}^{d_f}_{2\varepsilon^{*}}(X) =\textsf{iVR}^{d_f}_{2\varepsilon}(X)$
    \item (Polynomial Time Computable) $O(n^2\log(n)+\textsf{size}(f)n^2)$
    \item (Finite Bit Representation) $\log(\frac{1}{\varepsilon^{*}})=O(\frac{\textsf{size}(X)}{n})$
\end{enumerate}
\end{lemma}
\begin{proof}
Using Proposition \ref{prop: sol-2epsilon} and Proposition \ref{prop: max-VR-threshold} on ${\varepsilon}^*>0: \varepsilon^*< \varepsilon<1$, any box cover from $\textsf{iVR}^{d_f}_{2\varepsilon^*}(X)$ has each $R \in \textsf{iVR}^{d_f}_{2\varepsilon^*}(X)$ summarizable with the number $\textbf{\textsf{mid-range}}(f,R)$ so that:
    \begin{equation}
        \lvert f(X_i)-\textbf{\textsf{mid-range}}(f,R)\rvert_2 \leq \varepsilon^*<\varepsilon, \forall i \in R
        \end{equation}
    Thus, we can let our data structure be:
    \begin{equation}
         (\{(R,\textbf{\textsf{mid-range}}(f,R)): R \in \mathcal{S}(X,f,\varepsilon)\}, (\varepsilon^*,\varepsilon,f))
    \end{equation}
    By monotonicity of $\textsf{iVR}^{d_f}_{\bullet}(X)$, we must have:
    \begin{equation}
        \textsf{iVR}^{d_f}_{2\varepsilon^*}(X) \subseteq \textsf{iVR}^{d_f}_{2\varepsilon}(X), \forall \varepsilon^*: \varepsilon^*< \varepsilon<1
    \end{equation}
 In order to find $\varepsilon^{*}: \varepsilon^{*}<\varepsilon<1$, so that the converse is true, we must utilize the fact that there are finitely many pairwise distances on $X$.
 
 We know that $d_f^*(X,d_f,2\varepsilon) <2\varepsilon$. This proves Feasibility, Property 1. 
 
 We know that for
 \begin{equation}
     r \in [D(X,f,\varepsilon),2\varepsilon]
 \end{equation}
 we have:\begin{equation}\label{eq: equal-VR}
     \textsf{iVR}_r^{d_f}(X)=\textsf{iVR}_{2\varepsilon}^{d_f}(X)
 \end{equation}
 which follows by definition of an Indexed Vietoris-Rips Complex.

 The Maximal property, property 2,  follows since $r$ can be set to its lower bound.

 The Time complexity property, property 3, to find $\varepsilon^{*}$ follows by  Observation \ref{obs: dstar-complexity}.

The finite bit representation property, property 4, follows since $\varepsilon^{*}$ is a composition of a  constant number of algebraic operations on the vectors from $X_{\bullet}$.
\end{proof}
According to Lemma \ref{lemma: boxcoversum-sol}, a feasible $\varepsilon^{*}$ for a \textbf{box cover with box summaries} data structure can be found in polynomial time. Thus, we can always compress the largest possible Indexed Vietoris-Rips Complex $\textsf{iVR}^{d_f}_{2\varepsilon^{*}}(X)=\textsf{iVR}^{d_f}_{2\varepsilon}(X)$ according to Proposition \ref{prop: maximum-epsilon}. 
\subsection{A Decompression Algorithm}
We show in the following that because $\mathcal{C}(X,f,\varepsilon)$ is a box cover with box summaries the decompression algorithm can run in polynomial time. By showing this, we show that the only \textbf{obstruction} left to efficiently solving the $(k,D)$\textsc{-RectLossyVVFCompression} problem is in finding an optimal box cover.

We design a decompression algorithm in Algorithm \ref{alg: decompress}. It takes as input a \textbf{box cover with box summaries} data structure and computes a $(k,D)$-dimensional voxelized vector field $\hat{X}$. Its output $\hat{X}$ satisfies the inequality constraint of the $(k,D)$\textsc{-RectLossyVVFCompression}  problem. 
\begin{algorithm}[!h]
\SetAlgoLined
\SetKwComment{Comment}{/* }{ */}
\caption{A \textcolor{orange}{Decompression} Algorithm }\label{alg: decompress}
\KwData{Let the \textbf{box cover with box summaries} data structure $(\{(R,\hat{f}_R): R \in \mathcal{S}(X,f,\varepsilon)\}, (\varepsilon^*,\varepsilon,f))$, with $\mathcal{S}(X,f,\varepsilon)$ a box cover of $\textsf{supp}(X)$ be the output of a box compression algorithm on instance $(X,f,\varepsilon)$ of $(k,D)$\textsc{-RectLossyVVFCompression} with a piecewise $f$ on $f_j$ and error bound $\varepsilon$. } 
\KwResult{$\hat{X}$ a $(k,D)$-dimensional voxelized vector field with $\hat{X}_i \in \mathbb{R}^D, i \in \textsf{supp}(X)$}
\For{$R \in \mathcal{R}$}{
Run the algorithm $\mathcal{A}$ from  \cite{grigor1988solving} on the polynomial inequality (rescale coefficients to the integers if needed):
\begin{equation}
    x \in \tilde{\mathbb{Q}}^D : (f_j(x)-\hat{f}_R)^2 \leq {(\varepsilon-\varepsilon^*)^2}, \text{ where }f_j \text{ is defined on }\hat{f}_R
\end{equation}
$\tau \gets \text{Representative Set Output by }\mathcal{A}$ \\
\For{$i \in R$}{
$\hat{X}_i\gets \text{Some  vector }\hat{x} \in \tau$ \Comment*[r]{Entries of $\hat{x}$ are Algebraic Numbers from $\tilde{\mathbb{Q}}^D$. }
}
}
\KwRet $\hat{X}$
\end{algorithm}
We claim that for a \textbf{box cover with box summaries} data structure as given in Lemma  \ref{lemma: boxcoversum-sol}, Algorithm \ref{alg: decompress} is a polynomial time decompression algorithm for some clustered indexed compression scheme for the $(k,D)$-\textsc{RectLossyVVFCompression} problem. 
\begin{theorem}\label{thm: rect-covering}
    Let $(X,f,\varepsilon)$ be an instance of $(k,D)$\textsc{-RectLossyVVFCompression} and let 
    \begin{equation}
        \mathcal{C}(X,f,\varepsilon)=(\{(R,\hat{f}_R): R \in \mathcal{S}(X,f,\varepsilon)\}, (\varepsilon^*,\varepsilon,f))
    \end{equation}
    be some feasible solution.
    
    Then there is a decompression algorithm $\mathcal{D}$ running in time:
    \begin{equation}
       O(\lvert \mathcal{S}(X,f,\varepsilon)\rvert (\textsf{size}(f)\log(\frac{1}{\varepsilon^*}))^{O(D^2)}+\sum_{R \in \mathcal{S}(X,f,\varepsilon)}\lvert R\rvert  )
    \end{equation}
    so that:
    \begin{equation}
    \lvert f(X_i)-f((\mathcal{D}\circ \mathcal{C}(X,f,\varepsilon))_i)\rvert_2 \leq \varepsilon, \forall i \in \textsf{supp}(X)
\end{equation}
\end{theorem}
\begin{proof}
The decompression algorithm is given in Algorithm \ref{alg: decompress}. 

\textbf{Correctness: }

According to Algorithm \ref{alg: decompress}, the output $\hat{X}=\mathcal{D}\circ \mathcal{C}(X,f,\varepsilon)$ must satisfy for every $R \in \mathcal{R}$:
\begin{equation}
    (f(\hat{X}_i)-\hat{f}_R)^2 \leq {(\varepsilon-\varepsilon^*)^2}, \forall i \in R
\end{equation}
Taking square root on both sides, this is equivalent to:
\begin{equation}
    \lvert f(\hat{X}_i)-\hat{f}_R\rvert_2 \leq \varepsilon-\varepsilon^*, \forall i \in R
\end{equation}
By triangle inequality, we have the following bound:
\begin{equation}
    \lvert f(\hat{X}_i)-{f}(X_i)\rvert_2\leq \lvert f(\hat{X}_i)-\hat{f}_R\rvert_2 +\lvert \hat{f}_R-{f}(X_i)\rvert_2\leq {\varepsilon}, \forall i \in R
\end{equation}
Since $R \in \mathcal{R}$ was arbitrary and \begin{equation}
    \bigcup_{R \in \mathcal{R}}R=\textsf{supp}(X)
\end{equation}
every $i \in \textsf{supp}(X)$ satisfies this inequality.

This proves the desired lossy constraint on $R \in \mathcal{R}$.

\textbf{Complexity: } 

The asymptotic runtime of the algorithm in \cite{grigor1988solving} is: 
\begin{equation}
     O( \textsf{size}(f)(\log(\frac{1}{\varepsilon^*}))^{O(D^2)})
\end{equation}
Since the algorithm is run $\lvert \mathcal{S}(X,f,\varepsilon)\rvert$  times and it takes $m = O(\textsf{size}(f))$ time to find which $f_j$ of $f$ is defined on $\hat{f}_R \in \mathbb{R}$, we get the first term of the desired asymptotic complexity for Algorithm \ref{alg: decompress}. The second term comes from the number of assignments of a vector to an index. This costs the aggregation of the number of indices from all rectangles from $\mathcal{S}(X,f,\varepsilon)$.
\end{proof}
Of course, as shown in Lemma \ref{lemma: boxcoversum-sol}, the compressed data structure can represent $\textsf{iVR}^{d_f}_{2\varepsilon}(X)$. Thus, Algorithm \ref{alg: decompress} can decompress any feasible \textbf{box cover with box summaries} data structure.
\begin{remark}
    The requirement of:
\begin{equation}
        \lvert f(X_i)-\bullet_R\rvert_2 \leq {\varepsilon^*}<\varepsilon, \forall i \in R
\end{equation}
    is sufficient to have an inequality constraint in the decompression algorithm. 
    
    If the upper bound of ${\varepsilon^*}$ were replaced by $\varepsilon$, then this would require the decompression algorithm to solve for a root of a multivariate polynomial exactly. This is known to be intractable.
    
    In the Blum-Shub-Smale model~\cite{blum1989theory}, which is a computational model for real numbers, solving for a root exactly is undecidable~\cite{calvert2011noncomputable}. Grobner bases~\cite{10.1145/1088216.1088219} can also be used to solve for the roots of a multivariate polynomial. This computation, however, is asymptotically exponential in the number of variables. 
\end{remark}
\section{$(k,D)$\textsc{-RectLossyVVFCompression} is NP-Hard}\label{sec: np-hard}
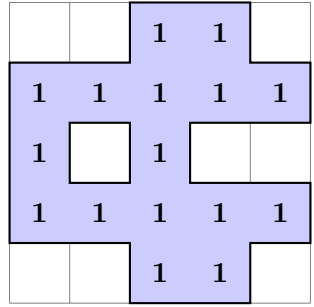
\begin{figure}[!h]
    \centering
\begin{tikzpicture}[scale=0.8]
\def\xmax{5}
\def\ymax{5}

\draw[step=1cm,gray,very thin] (0,0) grid (\xmax,\ymax);

\begin{scope}
\clip
  (0,1) -- (2,1) -- (2,0) -- (4,0) -- (4,1) -- (5,1) -- (5,2) -- (3,2) -- (3,3) -- (5,3) -- (5,4) -- (4,4) -- (4,5) -- (2,5) -- (2,4) -- (0,4) -- cycle;

\foreach \x in {0,...,\numexpr\xmax-1}
  \foreach \y in {0,...,\numexpr\ymax-1} {
  \pgfmathparse{(\x!=1 || \y!=2)} 
  \ifnum\pgfmathresult=1
    \fill[blue!20] (\x,\y) rectangle ++(1,1);
    \node at (\x+0.5,\y+0.5) {\textbf{1}};
    \fi
  }
\end{scope}

\draw[thick]
  (0,1) -- (2,1) -- (2,0) -- (4,0) -- (4,1) -- (5,1) -- (5,2) -- (3,2) -- (3,3) -- (5,3) -- (5,4) -- (4,4) -- (4,5) -- (2,5) -- (2,4) -- (0,4) -- cycle;
\draw[thick]
  (1,2) -- (2,2) -- (2,3) -- (1,3) -- cycle;

\end{tikzpicture}
\caption{A $0$-$1$ matrix instance of the \textsc{RectiLinearImgCompression} Problem. This problem is NP-Complete.}\label{fig: rectlinearimg}
\end{figure}
We show that the $(k,D)$\textsc{-RectLossyVVFCompression} problem is NP-Hard.  This means that finding an exact globally optimal solution would not be a tractable endeavour. 

We will show this by a Karp reduction from the \textsc{RectiLinearImgCompression} Problem \cite{garey2002computers}. In this problem, a matrix of zeros and ones is provided and in the optimization version the objective is to cover it with a minimum number of rectangular submatrix blocks. An illustration is shown in Figure \ref{fig: rectlinearimg}. 
\begin{theorem}\label{thm: np-hard-wproof}
    The optimization version of $(k,D)$\textsc{-RectLossyVVFCompression} is NP-hard if $k\geq 2, D\geq 1$.
\end{theorem}
\begin{proof}
We reduce the decision version of $(2,1)$\textsc{-RectLossyVVFCompression} from the decision version of the Rectilinear Image Compression Problem
(\textsc{RectImgCompression}), which is NP-Complete~\cite{garey2002computers}. 
This shows that the optimization version of $(2,1)$\textsc{-RectLossyVVFCompression} is NP-Hard. by Proposition \ref{prop: k'kD'Dsubprob}, $(k,D)$\textsc{-RectLossyVVFCompression} must be NP-Hard.

Given a matrix $M \in (\mathbb{Z}_2)^{m\times m}$ of $n=m^2$ pixels where we denote $M_i$ as a pixel at some index on $M$ and an integer $K' \leq n$. Let $\textsf{nnz}(M)=\lvert \{ i \in [m]^2: M_i\neq 0\}\rvert $. 

We transform $(M,K')$, an instance of the decision version of \textsc{RectImgCompression}, into an instance of the decision version of  $(2,1)$\textsc{-RectLossyVVFCompression}:

Let $\varepsilon=0.1$ and for a constant $r \in \mathbb{N}$:
\begin{equation}
    K= (K'+(n-\textsf{nnz}(M)))(10r
    +4\textsf{size}(n)+2v(\textsf{size}(f),n,\textsf{size}(X)))+5r+\textsf{size}(\varepsilon^*)+\textsf{size}(\varepsilon)+\textsf{size}(f)
\end{equation}
Since matrices are $(2,1)$-dimensional vector fields over $[m]^2$, we can let $X: [m]^2 \rightarrow \{0,1\}$ be defined as:
\begin{equation}
    X_i=M_i, \forall i \in [m]^2
\end{equation}
where the number of bits to represent $X_i$ is exactly $1$.

Let the objective function $f: \mathbb{R}^1 \rightarrow \mathbb{R}$ be defined as follows: 
\begin{equation}
    f(X_i)=\begin{cases}
        1 & \text{ if $X_i=1$}\\
        -\eta(i) & \text{ otherwise}
    \end{cases}, \forall i \in [m]^2
\end{equation} 
The function $\eta: [m]^2\rightarrow \mathbb{N}$ is defined from the voxel indices to an integer enumeration, or primary key, of the voxel vector field of $X$. This enumeration can be a row-major order, for example. 
\begin{equation}
    \eta((j,k))=mj+k \in \mathbb{Z}, \forall (j,k) \in [m]^2
\end{equation}
The function $f$ is $1$ on a general orthogonal polygon. This is because the nonzeros of $M$ form this general orthogonal polygon  over its indices. For the complement, the total order provided by $\eta$ separates the indices of the complement from each other by atleast $1$ as well as from $M=1$ by atleast $1$. 

\textbf{The box cover of $M$ of size $\leq K'$ $\Leftrightarrow$ there is a box cover with box summaries data structure of size $\leq K$ for $(X,f,\varepsilon=0.1)$:  }

$(\Rightarrow):$  
For the constraint $\lvert f(X_i)-f(X_j) \rvert_2 \leq 0.1$, the $(2,1)$-dimensional voxelized vector field separates into the general orthogonal polyhedron of $f=1$ and a disconnected set of indices from its exterior: $f\neq 1$. 

Thus, an clustered indexed algorithmic compression scheme $\mathcal{M}=(\mathcal{C},\mathcal{D})$ of $(X,f,\varepsilon=0.1)$ must have its box cover $\mathcal{S}(X,f,\varepsilon)$ be a union of a \textbf{box cover}  for the indices of $f=1$ and a box cover of all the individual indices of $f\neq 1$. This means:
\begin{equation}
    \lvert \mathcal{S}(X,f,\varepsilon) \rvert =\lvert \mathcal{S}'(M)\rvert +(n-\textsf{nnz}(M))
\end{equation}

Let  $\mathcal{S}'(M)$ be a box cover of $M$ with $\lvert \mathcal{S}'(M)\rvert \leq K'$. The box cover $\mathcal{S}'(M)$ covers the indices of $f=1$. These indices are the support of $M$. Thus, using the definition of $\mathcal{C}(X,f,\varepsilon)$ and $K$, we obtain the following inequality:
\begin{subequations}
\begin{equation}
    \textsf{size}(\mathcal{C}(X,f,\varepsilon))= 
    \end{equation}
    \begin{equation}
    \lvert \mathcal{S}(X,f,\varepsilon)\rvert(10r+4\textsf{size}(n)+2v(\textsf{size}(f),n,\textsf{size}(X)))+5r+\textsf{size}(\varepsilon^*)+\textsf{size}(\varepsilon)+\textsf{size}(f)
    \end{equation}
    \begin{equation}
        = \lvert \mathcal{S}'(M)+(n-\textsf{nnz}(M))\rvert(10r+4\textsf{size}(n)+2v(\textsf{size}(f),n,\textsf{size}(X)))+5r+\textsf{size}(\varepsilon^*)+\textsf{size}(\varepsilon)+\textsf{size}(f)
    \end{equation}
    \begin{equation}
    \begin{split}
    \leq  (K'+(n-\textsf{nnz}(X)))(10r+4\textsf{size}(n)+2v(\textsf{size}(f),n,\textsf{size}(X)))+5r\\+\textsf{size}(\varepsilon^*)+\textsf{size}(\varepsilon)+\textsf{size}(f)=K
    \end{split}
\end{equation}
\end{subequations}
Using the definition of the size of $\mathcal{C}(X,f,\varepsilon)$, the following inequality is now true:
\begin{equation}
    \textsf{CR}(\mathcal{C}(X,f,\varepsilon),X)= \frac{\textsf{size}(X)}{\textsf{size}(\mathcal{C}(X,f,\varepsilon))}
\geq \frac{\textsf{size}(X)}{K}
\end{equation}
$(\Leftarrow): $
Conversely, if \begin{equation}
    \textsf{CR}(\mathcal{C}(X,f,\varepsilon),X) = \frac{\textsf{size}(X)}{\textsf{size}(\mathcal{C}(X,f,\varepsilon))}\geq \frac{\textsf{size}(X)}{K},
\end{equation}
then using the definition of $K$ and $\textsf{size}(\mathcal{C}(X,f,\varepsilon)) $, we obtain the following inequality:
\begin{subequations}
    \begin{equation}
    \begin{split}
    \lvert \mathcal{S}(X,f,\varepsilon)\rvert(10r+4\textsf{size}(n)+2v(\textsf{size}(f),n,\textsf{size}(X)))+5r\\+\textsf{size}(\varepsilon^*)+\textsf{size}(\varepsilon)+\textsf{size}(f)
    \end{split}
    \end{equation}
    \begin{equation}
        \begin{split}
    =\lvert \mathcal{S}'(M)+(n-\textsf{nnz}(M))\rvert (10r+4\textsf{size}(n)+2v(\textsf{size}(f),n,\textsf{size}(X)))+5r\\+\textsf{size}(\varepsilon^*)+\textsf{size}(\varepsilon)+\textsf{size}(f)
    \end{split}
    \end{equation}
    \begin{equation}
        =\textsf{size}(\mathcal{C}(X,f,\varepsilon)) 
    \end{equation}
    \begin{equation}
    \begin{split}
    \leq K= (K'+(n-\textsf{nnz}(M)))(10r
    +4\textsf{size}(n)+2v(\textsf{size}(f),n,\textsf{size}(X)))\\+5r+\textsf{size}(\varepsilon^*)+\textsf{size}(\varepsilon)+\textsf{size}(f)
    \end{split}
\end{equation}
\end{subequations}
Then, we must have that:
\begin{equation}
    \lvert \mathcal{S}'(M)\rvert \leq K'
\end{equation}
Because the sets of voxel indices $X=1$ and $X\neq 1$ are disjoint  and because $\{i: X_i=1\}=\{i: M_i=1\}$,  we must have that the box cover of $M$ must satisfy: \begin{equation}
\mathcal{S}'(M)=\{R \in \mathcal{S}(X,f,\varepsilon): (f\circ X)\mid_{R}=1\}
\end{equation} 

This proves that the decision problem \textsc{RectImgCompression} polynomial time reduces to the decision problem $(2,1)$\textsc{-RectLossyVVFCompression} by a Karp reduction~\cite{garey2002computers}. 
\end{proof}
\section{$(k,D)$\textsc{-RectLossyVVFCompression} is APX-Hard}\label{sec: apx-hard}
In order to prove that the  $(k,D)$\textsc{-RectLossyVVFCompression} problem is APX-Hard, we need to find a PTAS reduction from an APX-Hard problem. A potential candidate problem is the \textbf{rectangle cover problem}, also known as the ``geometric set cover problem by rectangles" in the computational geometry literature~\cite{haussler1986epsilon}.

In the \textbf{rectangle cover problem}, the instances are \emph{range spaces}, which are pairs consisting of a set $X$ and a collection of subsets of $X$. The range space for the \textbf{rectangle cover problem} is a set of points on the plane and a collection of rectangles on the plane.  The goal is to find the minimum cardinality rectangle cover of the points. It is known to be APX-Hard~\cite{chan2014exact}. 

In fact, this is APX-Hard \emph{even when the collection of rectangles on $\mathbb{R}^2$ only intersect either emptily or at their boundaries at exactly $4$ points}. We call this problem the $2$-\textsc{DimSpecialBoxCover} problem.  We would like, however, an even more simplified version of the problem to fully exploit planarity. 

We can start from a combinatorial set system problem, unrelated to geometry, known to be APX-Hard~\cite{chan2014exact}. It is called the \textsc{Special-3SC} problem. It is a special case of the minimum vertex cover problem on graphs of degree at most three and can be reduced from the minimum vertex cover problem on degree $3$ regular graphs ($3$\textsc{VC}). The $3$\textsc{VC} problem is known to be APX-Hard~\cite{chan2014exact}.

We define this combinatorial set system problem here:
\begin{tcolorbox}
\begin{problem}
(Special 3SC Problem): 
        (\textsc{Special-3SC})\\
\textbf{Input: } A pair $(U,S)$ consisting of: 
\begin{enumerate}
     \item A set $U= A\cup B, $ where: 
        \begin{enumerate}
            \item $A=\{a_1,...,a_n\}$\item $B= W \cup X \cup Y\cup Z$ with: \\$W=\{w_1,...,w_m\}$
            , $X=\{x_1,...,x_m\}$
            , $Y=\{y_1,...,y_m\}$
            , $Z=\{z_1,...,z_m\}$
            \item $2n=3m$
        \end{enumerate}
        \item A set system $S$ of size $5m$ so that:
        \begin{enumerate}
            \item $S=\bigcup_{t \in [m]} S_t$ s.t. 
        \begin{equation}
           S_t \triangleq \{\{a_i, w_t\}, \{w_t, x_t\}, \{a_j , x_t, y_t\}, \{y_t, z_t\},  \{a_k, z_t\}: \exists i,j,k:  1\leq i<j<k\leq n \}
        \end{equation}
        \item $\forall i \in [n]$, $a_i \in A$\text{ belongs to exactly two sets $s_1,s_2$ in $S$}. The pair of sets are totally ordered as $s_1$ before $s_2$.
        \end{enumerate}
        \end{enumerate}
         \item[] \textbf{Output: } A subcollection $T: T \subseteq S$ so that $\lvert T\rvert $ is minimized.
\end{problem}
\end{tcolorbox}
For any instance  $(U,S)$ of the $\textsc{Special-3SC}$ problem. We define a map $\Phi: U \rightarrow \mathbb{R}^2$ that determines the PTAS reduction from \cite{chan2014exact} to the $2$\textsc{-DimSpecialBoxCover}.
\begin{definition}\label{def: phi}
The following map
    \begin{equation}
        \Phi: \bullet_i' \mapsto \bullet_i
    \end{equation}
    for \textsc{Special-$3$SC}$\leq_{PTAS}$$2$\textsc{-DimSpecialBoxCover} 
    satisfies: 
    \begin{enumerate}
    \item $\Phi(A)=\{a_1,...,a_n\}$ is totally ordered on the horizontal line: 
    \begin{enumerate}
        \item  $(a_i)_2=(a_j)_2=1, \forall i,j \in [n]$
        \item \begin{equation}
                (a_1)_1<(a_2)_1<\cdots <(a_n)_1
            \end{equation}
    \end{enumerate}  
            \item $\Phi(B)=\{w_1,...,w_m, x_1,...,x_m, y_1,...,y_m, z_1,...,z_m\}$ is totally ordered in the second dimension:
            \begin{equation}
                (w_1)_2<\cdots <(w_t)_2<(x_t)_2<(y_t)_2<(z_t)_2<(w_{t+1})_2<\cdots<(z_m)_2
            \end{equation}
            \item $\forall i\in [n]$ for the two sets (in order) $s_1,s_2 \in S:  s_1,s_2\ni a_i$ we have that:
            \begin{equation}
                \forall \vardiamond_t \in s_1\setminus \{a_i\}, \forall \varheart_{\tilde{t}} \in s_2\setminus \{a_i\}: (\Phi(\vardiamond_t))_2<(\Phi(\varheart_{\tilde{t}}))_2
            \end{equation}
        \end{enumerate}
\end{definition}
We state a problem that has instances in the form of the output of the PTAS reduction from \textsc{Special-3SC} to instances of the $2$\textsc{-DimSpecialBoxCover}. We call this the \textsc{SpecialRect-3SC} problem.  
\begin{tcolorbox}
\begin{problem}
        (Special Rectangle 3SC Problem): 
        (\textsc{SpecialRect-3SC})
\begin{itemize}
    \item[] \textbf{Input: } 
        A range space $(X,\mathcal{R})$ where:
        \begin{enumerate}
            \item $X=\Phi(U)$ where $\Phi$ is from Definition \ref{def: phi}.
            \item $\mathcal{R}$ consists of $5m$ boxes on $\mathbb{R}^2$ so that:
            \begin{enumerate}
          \item $\forall s \in S$\text{, there is a unique }$R \in \mathcal{R}$\text{ that contains each set of points $\Phi(s)$. }
        \item Any pair of rectangles in $\mathcal{R}$ intersect at either $4$ or $0$ points between their boundaries.
        \end{enumerate}
        \end{enumerate}
        for some instance $(U,S)$ of \textsc{Special-$3$SC}.
    \item[] \textbf{Output: } A minimum cardinality subcollection $\mathcal{S}$ of $\mathcal{R}$ that covers $X$.
    \end{itemize}
\end{problem}
\end{tcolorbox}
We will see in Section \ref{sec: embed-2N} that because the $\textsc{SpecialRect-$3$SC}$ problem is a strict subproblem of $2$\textsc{-DimBoxCover}, that \textsc{SpecialRect-3SC} is more amenable for our proof of the APX-Hardness of $(k,D)$\textsc{-RectLossyVVFCompression} compared to directly reducing from $2$\textsc{-DimBoxCover}. 

\textbf{Proof Outline: }


To prove the APX-Hardness result, we will introduce an intermediary problem called the $k$-Dimensional Integer Coordinate Voxel Grid Box Cover (\textsc{$k$-DimIntVGridBoxCover}) problem. We state it as an optimization problem in the following:
\begin{tcolorbox}
\begin{problem}
    ($k$-Dimensional Integer Coordinate Voxel Grid Box Cover Problem):
    
    (\textsc{$k$-DimIntVGridBoxCover}) 
\begin{itemize}
    \item[] \textbf{Input: } A range space $([q]^k,\mathcal{R})$ consisting of a set of:
\begin{enumerate}
        \item A voxel grid $[q]^k \subseteq \mathbb{Z}^k$ of $n$ points with $q^k=n$ and
        \item A collection $\mathcal{R}$ of $m$ boxes on $[q]^k$ with $\bigcup_{R \in \mathcal{R}} R \supseteq [q]^k$.
    \end{enumerate}

   \item[]  \textbf{Output: } A minimum cardinality subcollection $\mathcal{S}$ of $\mathcal{R}$ that covers $[q]^k$.
   \end{itemize}
\end{problem}
\end{tcolorbox}
The first step will be to reduce from \textsc{SpecialRect-3SC} to $2$\textsc{-DimIntVGridBoxCover}:
\begin{equation}
\textsc{SpecialRect-3SC}\leq_{\text{PTAS}}\textsc{$2$-DimIntVGridBoxCover}
\end{equation}
This PTAS reduction produces a subproblem of \textsc{$2$-DimIntVGridBoxCover}. We call this subproblem the  
\textsc{VGridSpecialRect-3SC} problem. Later, in Section \ref{sec: dimintvgridboxcover2rectlossyvfcomp}, we find a PTAS reduction from this subproblem to $(k,D)\textsc{-RectLossyVVFCompression}$ for $k,D\geq 2$:
\begin{equation}
   \textsc{VGridSpecialRect-3SC} \leq_{\text{PTAS}} (k,D)\textsc{-RectLossyVVFCompression}
\end{equation}
The composition of these two PTAS reductions provides a PTAS reduction from \textsc{SpecialRect-3SC} to $(k,D)$\textsc{-RectLossyVVFCompression}. Since \textsc{SpecialRect-3SC} is APX-Hard, showing this would be enough to show that $(k,D)$\textsc{-RectLossyVVFCompression} is APX-Hard when $k,D\geq 2$.

To prove the first reduction, we need to design an efficient polynomial time algorithm to embed a range space on the plane to a $2$D voxel grid. We show that this can be done with an efficient polynomial time algorithm. The algorithm has its output size up to a constant multiple of the input size. This will be a sufficient condition for the existence of a PTAS reduction. These are described in Section \ref{sec: embed-2N} and Section \ref{sec: sweep-line}. 

For the second reduction, we convert an instance of \textsc{VGridSpecialRect-3SC} to a \emph{disjoint union}  of orthogonal polygons and rectangles. Picking an error bound of $\varepsilon=0.1$, we can then assign values to each of these disjoint shapes. We can then assign constant values to this disjoint union to form a piecewise multivariate polynomial. This constructs an instance of $(2,2)\textsc{-RectLossyVVFCompression}$. These algorithms all run in polynomial time. We prove that due to the disjointness of the input instance, they do not cause any distortion in the output size. We discuss this in detail in Section \ref{sec: dimintvgridboxcover2rectlossyvfcomp}. By Proposition \ref{prop: k'kD'Dsubprob}, this shows a reduction to $(k,D)$\textsc{-RectLossyVVFCompression} for $k,D \geq 2$.
\subsection{A Collection of Rectangles on  the Plane can be Embedded into $[2N]^2$}\label{sec: embed-2N}
We first make some observations about how large the voxel grid needs to be in terms of the range space size for such an algorithm.

We first analyze the case of one dimension, namely the real line and a collection of intervals.
We consider a mapping from a collection of intervals to the real line. We would like such a mapping to preserve the relative ordering of all the endpoints on the real line. This kind of mapping is monotonic on the set of all endpoints:\begin{equation}
    S\triangleq \bigcup_{[a_i,b_i] \in \mathcal{R}}\{a_i,b_i\}
\end{equation}
We call these interval collection embeddings. This is defined formally as follows:
\begin{definition}
Let $\mathcal{R}=\{[a_i,b_i]: a_i, b_i \in \mathbb{R}\}$ be a collection of $N$ intervals on $\mathbb{R}$. We call a map $\phi: \mathbb{R} \rightarrow \mathbb{R}$ an \textbf{interval collection embedding} of $\mathcal{R}$ if it is monotonic on $S$:
\begin{equation}
    s_1<s_2 \text{ iff }\phi(s_1)< \phi(s_2), \forall s_1,s_2 \in S
\end{equation}
\end{definition}
Certainly an interval collection embedding can map the endpoints of its input intervals to the integers. We show that this can be done in near linear time for $N$ intervals and with all intervals lying in the range $[2N]$.
\begin{algorithm}[!h]
\SetAlgoLined
\SetKwComment{Comment}{/* }{ */}
\setstretch{1.2}
\caption{Integer Embedding of a Collection of Intervals}\label{alg: interval-embedding}
\KwData{$N$ intervals $[a_i,b_i] \subseteq \mathbb{R}, i=1,...,N$ with $\phi(a_i), \phi(b_i) \in \mathbb{Z}$}
\KwResult{An interval collection embedding $\phi: \mathbb{R} \rightarrow [1,2N]$}
\Comment{Form a Partition on the $N$ Intervals.}
$S \gets \bigcup_{[a_i,b_i] \in \mathcal{R}} \{a_i,b_i\}$ \\
$(c_1,...,c_{\lvert S\rvert })\gets \textsf{Sort}(S)$\\
$\phi((1-t)c_j+tc_{j+1})\gets t(j+1)+(1-t)j, \forall t \in [0,1], \forall j \in [\lvert S \rvert-1]$\Comment{Linear Interpolation}
\Return $\phi$
\end{algorithm}
\begin{proposition}\label{prop: interval-embedding}
    Let there be a collection $\mathcal{R}$ of $N$ intervals  on $\mathbb{R}$. Algorithm \ref{alg: interval-embedding} correctly computes an \textbf{interval collection embedding mapping} the endpoints to $[2N]$ and can be computed in time $O(N\log N).$
\end{proposition}
\begin{proof}
\textbf{Correctness: }

An interval consists of two endpoints. Since there are only $N$ intervals, there can be at most $2N$ interval endpoints. Thus, there are atmost $2N$ elements in $S$. On each $c_j \in S$, we have that $\phi(c_j)=j$. Furthermore, the parameterized linear interpolation over $c_j \in S$ is well defined. This means the map $\phi: \mathbb{R} \rightarrow [1,2N]$ is well defined. 

Certainly $\phi: \mathbb{R} \rightarrow [1,2N]$ is an interval collection embedding of $\mathcal{R}$ since it is a linear interpolation over the increasing sequence $c_1<c_2<...<c_{\lvert S\rvert}$

\textbf{Complexity: }

The complexity is in the sorting, which is $O(N \log N)$. The interpolation only requires mapping each $c_j$ to $j$ and storing a linear interpolation equation. An interpretable symbolic expression should suffices. This takes time $O(N)$.
\end{proof}
We can generalize the interval collection embedding to a box collection embedding as follows:
\begin{definition}
Let $\mathcal{R}$ be a collection of $N$ boxes in $\mathbb{R}^{k}$. 
A \textbf{box collection embedding} is the map $\phi: \mathbb{R}^{k} \rightarrow \mathbb{R}^{k}$ with:
\begin{equation}
        \phi: (x_1,...,x_k) \mapsto (\phi_j(x_j))_{j=1}^{k}
\end{equation}
 where $\phi_j: \mathbb{R} \rightarrow \mathbb{R}$ are \textbf{interval collection embedding}s. 
\end{definition}
We show in the following lemma that a collection of $N$ boxes from $\mathbb{R}^{k}$ can be embedded in near linear time into  $\mathbb{R}^{k}$.
\begin{lemma}\label{lemma: integer-emb}
   We can embed a collection $\mathcal{R}$ of $N$ boxes in $\mathbb{R}^{k}$ into $\mathbb{Z}^{k}$ in $O(k N \log N)$ time so that all boxes lie in $[2N]^{k}$.
\end{lemma}
\begin{proof}
For the collection $\mathcal{R}$ and a dimension $j \in [k]$, let $(R_i)_j$ be the interval $[a_i,b_i] \subseteq \mathbb{R}$ formed by the projection of $R_i \in \mathcal{R}$ onto the $j$-th dimension. Let \begin{equation}
    (\mathcal{R})_j\triangleq \{(R_i)_j, R_i \in \mathcal{R}\}
\end{equation}
be the collection of $N$ intervals. 

According to Proposition \ref{prop: interval-embedding}, we can compute an interval collection embedding $\phi_j: \mathbb{R} \rightarrow \mathbb{R}$ of $(\mathcal{R})_j$. The endpoints of each interval $\phi_j((R_i)_j), R_i \in \mathcal{R}$ are integers according to Proposition \ref{prop: interval-embedding}.

The dimension $j \in [k]$ is arbitrary and independent of each other dimension. Thus, we can repeat the above process for each $j \in [k]$. This forms rectangles in $\mathbb{Z}^{k}$ with all corners belonging to $\mathbb{Z}^{k}$. 

Since the algorithm is run $k$ times to get an algorithm running in time $O(kN\log N)$.
\end{proof}
\subsection{A Sweep-Line Algorithm to Cover the Complement}\label{sec: sweep-line}
As Lemma \ref{lemma: integer-emb} shows, we can embed a range space instance $(X,\mathcal{R})$ of \textsc{SpecialRect-$3$SC} into $[2N]^2$. This, however, is not an instance of $2$\textsc{-DimIntVGridBoxCover}. It is only a cover of the union of all boxes in $\mathcal{R}$. We need a covering of  all of $[2N]^2$. 

We can define the set of integer coordinates that are missing. 
Assuming a collection of boxes $\mathcal{R}$ in a $2$-dimensional voxel grid $[q]^2$. 
Let us define the \textbf{complement} of the union over $\mathcal{R}$:
\begin{equation}
    C\triangleq[q]^2 \setminus \bigcup_{R \in \mathcal{R}} R
\end{equation}
We would like to compute a covering of this complement by boxes. Taking the union of $\mathcal{R}$ and such a covering of the complement would give an instance to $2$\textsc{-DimIntVGridBoxCover}. 

If we can compute a box covering $H$ that satisfies:
\begin{equation}
   \exists K \in \mathbb{R}_{\geq 0}:  \lvert H \rvert \leq K\lvert \mathcal{R}\rvert, 
\end{equation}
then we will show in Theorem \ref{thm: R2ZPTAS-reduction} that the constant $K$ can be used to preserve a PTAS for \\\textsc{SpecialRect-$3$SC} to a PTAS for $2$\textsc{-DimIntVGridBoxCover}.

We show this explicitly in Algorithm \ref{alg: cover-complement}. It is a kind of sweep-line algorithm~\cite{shamos1976geometric} over a $2$D voxel-grid. In such an algorithm, a horizontal ``sweep-line" scans from bottom to top of the $2$D voxel grid. The goal is to record rectangles that fully cover the partial voxel grid \begin{equation}\label{eq: Ck}
    C_k=\{(i,j) \in C: j\leq k\}
\end{equation}
of $C$ scanned up to height $k$ until the entire voxel grid is scanned. 

The algorithm requires sweep-line \textbf{events}.  At each of these \textbf{events}, the algorithm records an update to the collection of rectangles $H$ covering the complement $C$. These events can be represented by key-value pairs where the key determines an integer height at which the sweep-line performs an update and the value is the rectangle which is intersecting with the sweep-line.
\begin{tcolorbox}[enhanced, colback=white]\label{def: sweep-line-events}
We define our sweep-line \textbf{events} as the following key-value pairs:
\begin{enumerate}
    \item \textcolor{purple}{\textbf{(Bottom Event}): } $(\min((R)_2), R)$ 
    \item \textcolor{blue}{\textbf{(Top Event}): }$(\max((R)_2), R)$ 
\end{enumerate}
We will have two events that always occur in our sweep line, the first is an initial starting event and the second is the final ending event:
\begin{enumerate}
    \item $(1,[q] \times 1)$ is a starting \textcolor{blue}{\textbf{Top Event}}
    \item $(q,[q]\times q)$ is an ending \textcolor{purple}{\textbf{Bottom Event}}
\end{enumerate}
These are the bottom and top of the voxel grid and thus are considered as not being a part of the complement $C$.
\end{tcolorbox}
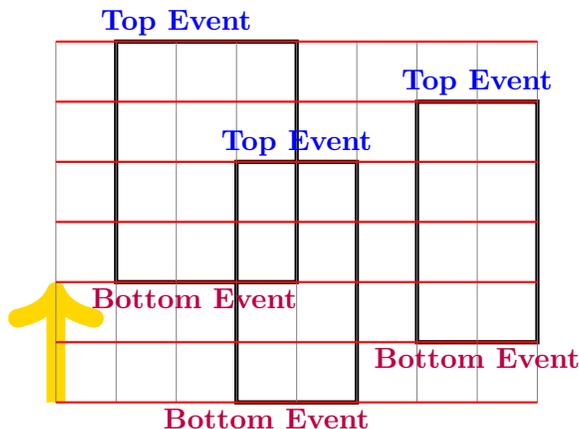
\begin{figure}[!h]
    \centering
\begin{tikzpicture}[scale=0.8]
    \definecolor{gold}{RGB}{255, 215, 0}
    \draw[->, thick, gold, line width=2.5mm] (0, 1) -- (0, 3);
        \draw[ultra thick] (3, 1) rectangle (5, 5);
        \draw[ultra thick] (6, 2) rectangle (8, 6);
        \draw[ultra thick] (1, 3) rectangle (4, 7);

        \draw[step=1cm,gray,very thin] (0,1) grid (8,7);

        \draw[thick, red] (0,1) -- (8,1);
        \draw[thick, red] (0,2) -- (8,2);
        \draw[thick, red] (0,3) -- (8,3);
        \draw[thick, red] (0,4) -- (8,4);
        \draw[thick, red] (0,5) -- (8,5);
        \draw[thick, red] (0,6) -- (8,6);
        \draw[thick, red] (0,7) -- (8,7);
        
        \node[anchor=south] at (3.5, 0.4) {\textcolor{purple}{\textbf{Bottom Event}}};
        \node[anchor=south] at (7, 1.4) {\textcolor{purple}{\textbf{Bottom Event}}};
        \node[anchor=south] at (2.3, 2.4) {\textcolor{purple}{\textbf{Bottom Event}}};
        \node[anchor=north] at (4, 5.7) {\textcolor{blue}{\textbf{Top Event}}};
        \node[anchor=north] at (7, 6.7) {\textcolor{blue}{\textbf{Top Event}}};
        \node[anchor=north] at (2, 7.7) {\textcolor{blue}{\textbf{Top Event}}};
    \end{tikzpicture}
    \caption{Illustration of a Sweep Line Algorithm with its \textcolor{blue}{\textbf{Top Events}} and \textcolor{purple}{\textbf{Bottom Events}}}
    \label{fig: sweep-line}
\end{figure}
The algorithm, similar to the sweep-line algorithm for the planar point location problem~\cite{sarnak1986planar} and the rectilinear polygon cover problem~\cite{franzblau1989performance}, requires tracking intervals covering $C_k$ from Equation \ref{eq: Ck} that close up later to form rectangles in $C$. 

This is an example of a persistence algorithm, \cite{edelsbrunner2002topological}. The intervals tracked from $C_k$ at each sweep line height form ``1-dimensional homology generators" that close up to form $2$-dimensional rectangles during an \textbf{event}. 

In ~\cite{sarnak1986planar}, these ``1-dimensional homology generators" are tracked with a \emph{partially persistent} search tree data structure. A data structure is called \emph{partially persistent}~\cite{driscoll1986making} if it can update its present state and query any of its past states. In the case of sweep-line algorithms, the state is indexed by the \textbf{events} ordered by height. The partial persistence is necessary to look back in time to the past homological generators.

Our algorithm can be summarized as the following procedure:
\stepcounter{algocf}
\begin{boxalgorithm}[label={alg: sweep-line-alg}]{Sweep-Line Algorithm To Cover $C$}
\textbf{Input: } $([q]^2,\mathcal{R})$ a range space. \\
\textbf{Output: } A collection $H$ of rectangles covering $C$.
\begin{enumerate}
    \item An empty collection $H$.
    \item \textcolor{ForestGreen}{\textbf{(Initial Creation): }} Store $(1,[q]\times 1)$ as a creation interval.
    \item For each \textbf{event} $(y,R)$ from low to high key values:
    \begin{enumerate}
    \item \textcolor{Goldenrod}{\textbf{(Query): }} Find adjacent interval(s) $I_l,I_r$ (if they exist) on either side of $\pi_1^y(R)$ 
    \item If the event is a  \textcolor{purple}{\textbf{Bottom Event}}:
    \begin{enumerate}
    \item \textcolor{brown}{\textbf{(Destruction): }} Delete/Mark-off any created intervals that are contained in $I_l \cup \pi_1^y(R) \cup I_r$
        \begin{enumerate}
            \item Store their corresponding rectangles in $H$.
        \end{enumerate}
        \item \textcolor{ForestGreen}{\textbf{(Creation): }} Store the intervals $I_l,I_r$ as created.
    \end{enumerate}
    \item Otherwise, if the event is a  \textcolor{blue}{\textbf{Top Event}}:
    \begin{enumerate}
        \item \textcolor{brown}{\textbf{(Destruction): }} Delete/Mark-off any created intervals that are contained in $I_l \cup \pi_1^y(R) \cup I_r$
        \begin{enumerate}
            \item Store their corresponding rectangles in $H$.
        \end{enumerate}
        \item \textcolor{ForestGreen}{\textbf{(Creation): }} Store the intervals $I_l,I_r$ as created.
    \end{enumerate}
    \end{enumerate}
    \item \textcolor{brown}{\textbf{(Final Destruction): }} All remaining created intervals form rectangles with the top of the voxel grid. Store these in $H$.
\end{enumerate}
\end{boxalgorithm}
The algorithm performs a sweep-line from the bottom to the top of the voxel grid. It tracks intervals adjacent to an event. These are ``created" homological generators, represented by intervals, for some ``homology group"~\cite{edelsbrunner2002topological, may1999concise}. The previously created intervals covered by them close off with a ``destruction" of the homological generators.

We show in Lemma \ref{lemma: cover-the-complement}, that Algorithm \ref{alg: sweep-line-alg} has an output size bound up to a constant factor with the number of boxes $N$. 
\begin{lemma}\label{lemma: cover-the-complement}
Let $\mathcal{R}$ be a collection of $N$ boxes, each box a subset of $[q]^2$. 
There is a box cover of $C$ with at most $4N$ boxes consisting of boxes contained in $C$.

This can be computed in time $O(N\log N)$ by implementing Algorithm \ref{alg: sweep-line-alg} with an interval tree, see Algorithm \ref{alg: cover-complement} and  Section \ref{sec: interval-tree}, Section \ref{sec: algorithms-appendix} in the Appendix.
\end{lemma}
\begin{proof}
\textbf{Correctness of Algorithm \ref{alg: cover-complement}: }

The algorithm forms horizontal ``sweep segments" at every sweep-line \textbf{event}. We define our sweep-line \textbf{event} by the value of the second dimension (vertically) of the upper and lower sides of the rectangles in $\mathcal{R}$.

For each sweep-line \textbf{event} in order from low to high, it must ``create" atmost two intervals which must be closed up to form a $2$D rectangle at a later event. 

\textbf{Every Created Interval is Paired with some Destruction Event: }

Certainly each created interval must eventually have some sweep-line event above it. Such sweep-line events destroy a created interval. 

Say that there are still created intervals that do not have a sweep-line event above them. We know, however, that the last sweep-line event destroys all remaining created intervals. This proves the claim.

\textbf{There are Atmost $2N$ Creation Events: }

Since there are $N$ rectangles and each can create atmost $2$ intervals, this proves the claim.

Algorithm \ref{alg: cover-complement} is an implementation of Algorithm \ref{alg: sweep-line-alg} using an interval tree. The interval tree acts as the storage data structure to track the created intervals. When an interval is created it is inserted into the interval tree. When an interval is destroyed, it is deleted from the interval tree. This is a direct correspondence between the memory requirements of the algorithm and the operations on the data structure. Thus, it does not change correctness. 

\textbf{Complexity of Algorithm \ref{alg: cover-complement}: }

Algorithm \ref{alg: cover-complement} requires two sorting operations and $O(N)$ number of query/insert/delete operations on the interval tree. These operations are $\log(N)$ in complexity, see Section \ref{sec: interval-tree} in the Appendix. Thus the algorithm runs in time $O(N\log N)$. 
\end{proof}
\begin{algorithm}[!h]
\SetAlgoLined
\SetKwRepeat{Do}{do}{while}
\SetKwComment{Comment}{/* }{ */}
\caption{Box Covering of the Complement of a Box Cover by Sweep Line using an Interval Tree}\label{alg: cover-complement}
\KwData{$\mathcal{R}=\{R\subseteq [q]^2: R\text{ is a box}\}; \lvert \mathcal{R} \rvert = N$}
\KwResult{A box cover $H$ of $[q]^2 \setminus \bigcup_{R \in \mathcal{R}}R$ consisting of  boxes from  $[q]^2 \setminus \bigcup_{R \in \mathcal{R}}R$}
$\mathcal{T} \gets \text{Empty Interval Tree}$\\
$\mathcal{S}_1\gets \emptyset,\mathcal{S}_2 \gets \emptyset$\\
\For{$i=1,2$}{
$\mathcal{S}^{\text{left}}_i\gets \{(\min((R)_i), (\pi^{\min((R)_i)}_{i}(R),-1 )):R \in \mathcal{R}\}$\\
$\mathcal{S}^{\text{right}}_i\gets \{(\max((R)_i), (\pi^{\max((R)_i)}_{i}(R), +1)):R \in \mathcal{R}\}$
\\
$\mathcal{S}_i\gets \mathcal{S}^{\text{left}}_i \cup \mathcal{S}^{\text{right}}_i$\Comment*[r]{A set of key value pairs}
}
$(X,\mathcal{I}_1)\gets \textsf{SortByKey}(\mathcal{S}_1)$ \Comment*[r]{In Increasing Order, with $\lvert X \rvert =\lvert \mathcal{I}_1\rvert$}
$(Y,\mathcal{I}_2)\gets \textsf{SortByKey}(\mathcal{S}_2)$ \Comment*[r]{In Increasing Order, with $\lvert Y \rvert =\lvert \mathcal{I}_2\rvert$}
$\textsf{Insert}(\mathcal{T}, [1,q])$\\
\For{$(y,(I,s)) \in (Y,\mathcal{I}_2) \text{ in increasing order of $Y$ }$}{
$(C_l,C_r) \gets \textsf{AdjacentNeighbors}(I,\mathcal{I}_1)$ \Comment*[r]{See Section \ref{sec: algorithms-appendix} in Appendix. }
$\mathcal{D}\gets \textsf{Query}(\mathcal{T},[\min(C_l),\max(C_r)], \supseteq)$\Comment{Query for any created intervals that are contained in $[\min(C_l),\max(C_r)]$.}
\For{$D \in \mathcal{D}$}{ 
$\textsf{Delete}(\mathcal{T}, D)$\Comment{These created intervals are now $2$D and can be deleted.}
}
\Comment{Newly created intervals}
\uIf{$s=-1$}{
$\textsf{Insert}(\mathcal{T}, C_l)$\\
$\textsf{Insert}(\mathcal{T}, C_r)$
}\Else{
$C \gets [\min(C_l),\max(C_r)]$\\
$\textsf{Insert}(\mathcal{T}, C)$
}
}
\Return $S$
\end{algorithm}
\subsection{$\textsc{SpecialRect-3SC} \leq_{PTAS} \textsc{VGridSpecialRect-3SC}$}
In order to reduce $\textsc{SpecialRect-$3$SC}$ to $\textsc{VGridSpecialRect-3SC}$, we would like to use a box collection embedding as in  Lemma \ref{lemma: integer-emb}  along with a cover of the complement so that the corresponding optimal solutions across the two problems are in one-to-one correspondence by the embedding.

We can describe such a reduction algorithm here:
\stepcounter{algocf}
\begin{boxalgorithm}[label={alg: reduction-special2vgrid}]{Reduction from \textsc{SpecialRect-3SC} to $\textsc{VGridSpecialRect-3SC}$}
\textbf{Input: } $(X,\mathcal{R})$ a range space  instance of \textsc{SpecialRect-3SC} with $N=\lvert \mathcal{R}\rvert$\\
\textbf{Output: } An instance of \textsc{VGridSpecialRect-3SC}
\begin{enumerate}
\item $\mathcal{R}''\gets \{\phi(R) \cap \mathbb{Z}^2: R \in \mathcal{R}\}$ where ${\phi}$ is a box collection embedding of Lemma \ref{lemma: integer-emb}.
\item $H(\mathcal{R}'')\gets $\text{ Box cover from the sweep-line algorithm on }$\mathcal{R}''$.
\item \textbf{Return } $([2N]^2,\mathcal{R}''\cup H(\mathcal{R}'')$
 \end{enumerate}
\end{boxalgorithm}
We make the following observation about the relationship between $H(\mathcal{R}'')$, $\mathcal{R}''$  and their solutions.
\begin{proposition}\label{prop: H-R-relationship}
Let $(X,\mathcal{R})$ be an instance of \textsc{SpecialRect-$3$SC}. 
Let $([2N]^2,\mathcal{R}''\cup H(\mathcal{R}''))$ be an instance of \textsc{VGridSpecialRect-$3$SC}.

Furthermore, the corresponding solution $\mathcal{S}$ for \textsc{SpecialRect-$3$SC} and $\mathcal{S}'$ of \textsc{VGridSpecialRect-$3$SC} are related by:
\begin{equation}
    \lvert \mathcal{S}'\rvert =\lvert \mathcal{S}\rvert + \lvert H^*\rvert 
\end{equation}
 where $H^*=H(\mathcal{R}'')$ is the minimal subcover of $H(\mathcal{R}'')$ covering the complement $C$ of $\bigcup_{R \in \mathcal{R}''}R$ in $[2N]^2$. 
\end{proposition}
\begin{proof}
    Since $C$ is the complement to $\bigcup_{R \in \mathcal{R}''} R$ with respect to $[2N]^2$ and $H(\mathcal{R}'')$ covers $C$ using rectangles on $C$, we can say that the optimal cover $\mathcal{S}'$ will be the disjoint union of the optimal subcover of $H(\mathcal{R}'')$ and the optimal subcover of $\mathcal{R}''$. 
    
    We notice that $H(\mathcal{R}'')$ is a collection of disjoint rectangles only intersecting on the boundaries. This is because in the horizontal direction the rectangles in $H(\mathcal{R}'')$ are always separated by some rectangle in $\mathcal{R}''$. In the vertical direction the rectangles only intersect on the boundary when there is an event (either top of bottom). Thus, its minimal subcover is itself.
    This gives the desired equality.
\end{proof}
We show in the following that the ``planarity"-like property of the instances of \textsc{SpecialRect-$3$SC} are sufficient to provide a lower bound to the size of the solution $\mathcal{S}'$ of the $\textsc{VGridSpecialRect-3SC}$ instance $([2N]^2,\mathcal{R}''\cup H(\mathcal{R}''))$. 
\begin{lemma}\label{lemma: lower-bound}
Let $(X,\mathcal{R})$ be an instance of \textsc{SpecialRect-$3$SC}. 

Let the instance $([2N]^2,\mathcal{R}''\cup H(\mathcal{R}''))$ of the $\textsc{VGridSpecialRect-3SC}$ problem be the output of Algorithm \ref{alg: reduction-special2vgrid}. 

The optimal minimizing solution $\mathcal{S}'$ for $\textsc{VGridSpecialRect-3SC}$ satisfies the lower bound:
    \begin{equation}
        \lvert \mathcal{S}'\rvert \geq \frac{\lvert \mathcal{R}''\rvert}{3}+ \lvert {H}^*\rvert
    \end{equation}
    where $H^*=H(\mathcal{R}'')$ is the optimal subcover of $H(\mathcal{R}'')$ covering the complement $C$ of $\bigcup_{R \in \mathcal{R}''}R$ in $[2N]^2$. 
\end{lemma}
\begin{proof}
According to Proposition \ref{prop: H-R-relationship}, we know that $H(\mathcal{R}'')$ and $\mathcal{R}''$ are disjoint.

Thus, it suffices to lower bound the optimal minimal cover of the instance:
 \begin{equation}
     (\bigcup_{R' \in \mathcal{R}''}(R' \cap \mathbb{Z}^2),\mathcal{R}'')
 \end{equation}
 as a geometric set cover instance.

 The collection $(X,\mathcal{R})$ is an instance of \textsc{SpecialRect-$3$SC}. We would like to show that the minimal cover cannot be too small. 
 
 We know that the problem \textsc{SpecialRect-$3$SC} is reduced from \textsc{Special-$3$SC}. The ``low-degree" property of the instances of \textsc{Special-$3$SC} shows up as a ``planarity property" in the instances  of \textsc{SpecialRect-$3$SC}. We will use the ``low degree property" to show that a small number of rectangles cannot cover too many points at once. 

 Let $(U,S)$ be an instance of \textsc{Special-$3$SC}. 
 According to the constraint for \textsc{Special-$3$SC},  every set $s \in S$ has atmost $3$ points in it.

 Let us bijectively map the sets in $S$ to a set of nodes $V$ and the points $U$ bijectively to a set of edges $E$ where the set-point incidence is maintained. 
 
 We can check that the edges are well defined. This means that the set-point incidence is maintained as a node-edge incidence since each point belongs to exactly $2$ sets. 
 
 Using the bijection, the minimal vertex cover $V^* \subseteq V$ on graph $(V,E)$ of degree at most $3$ is thus in one-to-one correspondence with a minimal set cover $S^* \subseteq S$.   

 Since a minimal vertex cover $V^* \subseteq V$ already covers all the edges $E$, we can say that for $V^*$: 
 \begin{equation}
     3\lvert V^*\rvert \geq \lvert E \rvert = 4m+n
 \end{equation}
 Thus, we have that the optimal solution $S^* \subseteq S$ satisfies:
 \begin{equation}
     \lvert S^* \rvert \geq \frac{4m+n}{3}\geq \frac{5m}{3}=\frac{\lvert S \rvert}{3} 
 \end{equation}
 because $\lvert S^*\rvert = \lvert V^*\rvert, \lvert E \rvert = \lvert U \rvert$ and $2n=3m$.

 According to Proposition \ref{prop: H-R-relationship}, the minimal subcover $H^*$ of  $H(\mathcal{R}'')$  is itself:
 \begin{equation}\label{eq: opt-H}
     H^*=H(\mathcal{R}'')
 \end{equation}
 
 Since we have that: 
 \begin{equation}
     \bigcup_{R' \in \mathcal{R}''}(R' \cap \mathbb{Z}^2) \supseteq \phi\circ \Phi(U), 
 \end{equation}
 the optimal solution $\mathcal{S}'$ of \textsc{VGridSpecialRect-$3$SC} must contain the following union of two collections whose rectangles do not intersect:
 \begin{equation}
     \{\phi\circ \Phi(t) \cap \mathbb{Z}^2 : t \in S^*\} \cup H^*
 \end{equation}
 We also know that $\lvert \mathcal{R}'' \rvert=\lvert S\rvert$ since the sets in $S$ are bijectively mapped to boxes in $\mathcal{R}$ through $\Phi$ and the box collection embedding $\phi$ is a bijection. Thus we can continue bounding the previous lower bound:
 \begin{equation}
     \lvert \mathcal{S}' \rvert \geq \lvert S^* \rvert +\lvert H^*\rvert \geq \frac{\lvert S \rvert}{3}+\lvert H^*\rvert=\frac{\lvert \mathcal{R}'' \rvert}{3}+\lvert H^*\rvert= \frac{\lvert \mathcal{R}'' \rvert}{3}+\lvert H(\mathcal{R}'')\rvert
 \end{equation}
 where the last equality follows from Equation \ref{eq: opt-H}. 
 This gives the desired lower bound.  
\end{proof}
\begin{theorem} \label{thm: R2ZPTAS-reduction}
    $\textsc{SpecialRect-3SC} \leq_{PTAS} \textsc{VGridSpecialRect-3SC}$
\end{theorem}
\begin{proof}
We define the following PTAS reduction: 
\begin{enumerate}
    \item Let the reduction map $m: I_{\textsc{SpecialRect-3SC}} \rightarrow I_{\textsc{VGridSpecialRect-3SC}}$ map as follows:
    \begin{equation}
        m: (X,\mathcal{R})\mapsto ([2\lvert \mathcal{R}\rvert ]^2,\mathcal{R}'' \cup H(\mathcal{R}''))
    \end{equation}
    be Algorithm \ref{alg: reduction-special2vgrid} acting as a map from its input to its output. 
    \item Let $\alpha(\epsilon)\triangleq \frac{\epsilon}{13}$
    \item Let $x=(X,\mathcal{R})$ be an instance of \textsc{SpecialRect-$3$SC}. 
    
    Let $s=\mathcal{S}$ be a feasible solution to the \textsc{VGridSpecialRect-3SC} on instance $m(x)$. 
    
    Certainly since $\mathcal{S}\subseteq \mathcal{R}''\cup H(\mathcal{R}'')$.
    We can then define:
    \begin{equation}
        g(x,s)\triangleq \phi^{-1}(s\setminus H(\mathcal{R}''))
    \end{equation}
    This map $g$ is polynomial time computable since the set difference can be computed by enumeration and $\phi^{-1}$ only takes $O(1)$ to compute for each rectangle in the set difference.
\end{enumerate}

\textbf{PTAS preservation from \textsc{VGridSpecialRect-3SC} to \textsc{SpecialRect-3SC}: }

Let $\mathcal{A}$ be a $(1+\alpha(\epsilon))$-approximation algorithm to \textsc{VGridSpecialRect-3SC}. This means: 
\begin{equation}
    \mathcal{A}(m(x)) \leq (1+\alpha(\epsilon))\lvert \mathcal{S}'\rvert 
\end{equation}
where $\mathcal{S}'$ is an optimal minimizing solution to \textsc{VGridSpecialRect-3SC}. 

We then have that:
\begin{DispWithArrows*}[format=ll,wrap-lines,fleqn,mathindent = 2cm]
    \lvert g(x,s)\rvert &\Arrow{ \text{ (By Definition of $g$)}}\\
    = \lvert \mathcal{A}(m(x))\rvert -\lvert H(\mathcal{R}'')\rvert &\Arrow{\text{ (By $\mathcal{A}$ as a $(1+\alpha(\epsilon))$-approximation) }} \\
    \leq  (1+\alpha(\epsilon)) \lvert \mathcal{S}'\rvert -\lvert H(\mathcal{R}'')\rvert &\Arrow{\text{ (By Proposition \ref{prop: H-R-relationship})}}\\
    = (1+\alpha(\epsilon)) (\lvert \mathcal{S}\rvert+\lvert H(\mathcal{R}'')\rvert) -\lvert H(\mathcal{R}'')\rvert \\
    = (1+\alpha(\epsilon))\lvert \mathcal{S}\rvert +\alpha(\epsilon)\lvert H(\mathcal{R}'')\rvert &\Arrow{\text{ (By Lemma  \ref{lemma: cover-the-complement})}}\\
    \leq (1+\alpha(\epsilon))\lvert \mathcal{S}\rvert +4\alpha(\epsilon)\lvert \mathcal{R}''\rvert &\Arrow{\text{ (By Lemma \ref{lemma: lower-bound})}}\\
    \leq (1+\alpha(\epsilon))\lvert \mathcal{S}\rvert +12\alpha(\epsilon)({\lvert \mathcal{S}'\rvert-\lvert H(\mathcal{R}'')\rvert}) &\Arrow{\text{ (By Proposition \ref{prop: H-R-relationship}) }}\\
    =  (1+\alpha(\epsilon))\lvert \mathcal{S}\rvert +12\alpha(\epsilon){\lvert \mathcal{S}\rvert} &\Arrow{ (By Definition of $\alpha(\epsilon)$) }\\
    \leq (1+\frac{\epsilon}{13})\lvert \mathcal{S}\rvert +\epsilon\frac{12\lvert \mathcal{S}\rvert}{13}\\
    \leq (1+{\epsilon})\lvert \mathcal{S}\rvert
\end{DispWithArrows*}
\end{proof}
\begin{remark}
    It is necessary that we reduce from the \textsc{SpecialRect-$3$SC} problem. If we view \textsc{SpecialRect-$3$SC} as a subproblem of $2$-\textsc{DimIntVGridBoxCover}. Without the ``planarity"  property of the instances of \textsc{SpecialRect-$3$SC}, the lower bound from Lemma \ref{lemma: lower-bound} would not hold in general. 

    Consider, for example, a rectangle on $[2N]^2, N \in \mathbb{N}$ with $\Omega(N)$ other rectangles that intersect with it. This results in the following:
    \begin{equation}
        \Omega(N) \lvert S^*\rvert \geq N \Rightarrow \lvert S^*\rvert \geq O(1)   
    \end{equation}
\end{remark}
\begin{corollary}
$2$-\textsc{DimIntVGridBoxCover} is APX-Hard.
\end{corollary}\label{cor: 2dimintvgrid-apxhard}
\begin{proof}
    This follows by Theorem \ref{thm: R2ZPTAS-reduction}.
\end{proof}
\bigskip
\subsection{\textsc{VGridSpecialRect-3SC}  $\leq_{PTAS} (k,D)$\textsc{-RectLossyVVFCompression}}\label{sec: dimintvgridboxcover2rectlossyvfcomp}
We show in the following the PTAS reduction from \textsc{VGridSpecialRect-3SC} to \\ $(k,D)$\textsc{-RectLossyVVFCompression}. This requires a reduction map that on an instance of \\\textsc{VGridSpecialRect-$3$SC} can: 
\begin{enumerate}
    \item Construct a voxelized vector field and
    \item Construct a piecewise multivariate polynomial.
\end{enumerate}

\textbf{1. Constructing the Voxelized Vector Field: }

Let the voxelized vector field be an inclusion of the $2$-dimensional voxel grid $[2N]^2$ from the instance $([2N]^2,\mathcal{R}')$ of \textsc{VGridSpecialRect-$3$SC} into $\mathbb{R}^2$.

\textbf{2. Constructing the Piecewise Multivariate Polynomial: }

If we can form a partition, or disjoint union, of shapes from the instance $([2N]^2,\mathcal{R}')$ of \\\textsc{VGridSpecialRect-3SC}. Certainly a piecewise multivariable polynomial could then be defined over such a partition of $[2N]^2 \subseteq \mathbb{R}^2$.

For a disjoint union of rectangles, any optimal cover of its union is itself. When the rectangles intersect, however, this may not be true.

    We know that the instance $([2N]^2,\mathcal{R}')$ of \textsc{VGridSpecialRect-$3$SC} has $\mathcal{R}'$ of the form:
    \begin{equation}
        \begin{split}
        \mathcal{R}'\triangleq \mathcal{R}''\cup H(\mathcal{R}'') \text{ s.t. }\mathcal{R}''\triangleq\{\phi(R) \cap \mathbb{Z}^2 :R \in \mathcal{R}\} \\\text{ and } H(\bullet) \text{ is computed by the sweep-line Algorithm \ref{alg: sweep-line-alg}.}
        \end{split}
    \end{equation}
    We also know that $H(\mathcal{R}'')$ is a disjoint union of rectangles that cover the complement of $\mathcal{R}''$. Thus, $\mathcal{R}''$ is the only subcollection of $\mathcal{R}'$ that has rectangles that can intersect. 
   Let:
   \begin{equation}
       \mathcal{U}(\mathcal{R}'')\triangleq \bigcup_{R \in \mathcal{R}''} R
   \end{equation}
We prove a relationship between the optimal cover of $\mathcal{U}(\mathcal{R}'')$ and the collection $\mathcal{R}''$:
\begin{proposition}\label{prop: subcover-opt}
Let $([2N]^2,\mathcal{R}'' \cup H(\mathcal{R}''))$ be an instance of \textsc{VGridSpecialRect-3SC}.

   The rectangle cover of minimum size for $\mathcal{U}(\mathcal{R}'')$ has the cardinality of the subcover $\mathcal{O}$ of $\mathcal{R}''$ consisting of rectangles whose corners are not covered by any other rectangle in $\mathcal{R}''$.
\end{proposition}
\begin{proof}
We state here Property 2b for an instance $(X,\mathcal{R})$ of \textsc{SpecialRect-$3$SC}. It states that the collection $\mathcal{R}$ satisfies the following property:

\textbf{Property 2b: }\label{prop: prop2b}
    ``\text{Every pair of rectangles $R,R' \in \mathcal{R}: R \neq R'$, either intersect at exactly four points } \\\text{between their boundaries or not at all."}
    
In the definition of an instance of \textsc{VGridSpecialRect-$3$SC}, because of the box collection embedding $\phi$ we must have that $\mathcal{R}''=\{\phi(R) \cap \mathbb{Z}^2: R \in \mathcal{R}\}$ of the instance $([2N]^2,\mathcal{R}''\cup H(\mathcal{R}''))$ satisfies Property 2b of a \textsc{SpecialRect-$3$SC} instance.

If a corner is covered, then in order to satisfy Property 2b, $R$ must be a strict subset of another rectangle $R'\in \mathcal{R}''$.

We have just shown that pairs of rectangles in $\mathcal{R}''$ either:
 \begin{enumerate}
     \item Do not intersect.
     \item Have one strictly contained in the other.
     \item Intersect at $4$ points between their boundaries.
 \end{enumerate}
We show that the collection $\mathcal{O}: \mathcal{O}\subseteq \mathcal{R}''$ consisting of all rectangles whose corners are not covered by any other rectangle in $\mathcal{R}''$ is a cover of $\mathcal{U}(\mathcal{R}'')$.

Since the collection $\mathcal{O}$ satisfies that all rectangles in $\mathcal{O}$ have their corners uncovered, we must have that it cannot have any rectangle satisfy Property 2 for pairs of rectangles in $\mathcal{R}''$.  In fact, since we want to minimize the size of the cover, the smaller of the two rectangles $R,R' \in \mathcal{R}'': R \subseteq R'$ cannot belong to $\mathcal{O}$.

\textbf{Minimality: }

    We prove that there is no rectangle cover of $\mathcal{U}(\mathcal{R}'')$ of smaller cardinality than $\mathcal{O}$.

    Say that the subcover $\mathcal{O}: \mathcal{O}\subseteq \mathcal{R}''$ is suboptimal in size. 
    This means some rectangle $R^* \in \mathcal{O}$ is not in any of the true optimal covers of $\mathcal{U}(\mathcal{R}'')$.

    Since $R^*$ does not belong to a true optimal cover, all the points $x \in R^*$ must now be covered by some alternative rectangles (which may or may not be in $\mathcal{R}''$). 

    None of the four corners of $R^*$ are covered by any other rectangle in $\mathcal{R}''$ besides $R^*$. Consider without loss of generality one of the corners $c \in R^*$ with: 
    \begin{equation}
        c=\min_{x \in R^*}\max_{j \in [2]} (x)_j
    \end{equation}
    We notice that the rectangle $R^*$ is maximal. This means that the largest rectangle, with respect to the inclusion operation, that contains corner $c \in R^*$ is $R^*$ itself.

    In order to form a rectangle in $\mathcal{U}(\mathcal{R}'')$ with $c$ as a bottom left corner, we must consider all points above and to the right of $c$. We notice the following:
    \begin{equation}
        (c)_j\leq (y)_j, \forall j\in [2] \Leftrightarrow y \in R^*, \forall y \in \mathcal{U}(\mathcal{R})
    \end{equation}
    Thus, we can say:
    \begin{equation}
        \textsf{BB}(c,y) \subseteq R^*, \forall y \in \mathcal{U}(\mathcal{R}) \text{ s.t. } (c)_j\leq (y)_j, \forall j\in [2]
    \end{equation} 
    In fact, if we maximize over all such $y$, there is a unique maximum. This maximum is $R^*$ itself.

    This means that any rectangle that covers the corner $c \in R^*$, can be replaced by $R^*$ without increasing the size of the cover.
    
    By symmetry, all such corners have this property. Thus, $R^*$ belongs to the optimal cover of $\mathcal{U}(\mathcal{R})$. Since $R^*$ was arbitrary from $\mathcal{O}$, it  must be that any optimal cover of $\mathcal{U}(\mathcal{R})$ has size atleast $\lvert \mathcal{O}\rvert $.
    \end{proof}
    We have just shown that the optimal covering of $\mathcal{U}(\mathcal{R}'')$ is a subcover of $\mathcal{R}''$. Since the complement $H(\mathcal{R}'')$ of $\mathcal{U}(\mathcal{R}'')$ with respect to $[2N]^2$ consists of disjoint rectangles, the optimal subcover of $\mathcal{R}$ for $[2N]^2$ is {$\mathcal{O}\cup H^*$} where $H^*$ is an optimal subcover of $H(\mathcal{R}'')$.
    
We can now design an algorithm that determines a piecewise multivariate polynomial on a voxel grid $[2N]^2$. We need this algorithm to design a PTAS reduction from \textsc{VGridSpecialRect-$3$SC} to $(k,D)$\textsc{-RectLossyVVFCompression}.
\stepcounter{algocf}
\begin{boxalgorithm}[label={alg: multi-poly-construction}]{A Piecewise Multivariate Polynomial to Encode the Solution of \textsc{VGridSpecialRect-$3$SC}}
    \begin{itemize}
        \item[] \textbf{Input: }An instance $([2N]^2,\mathcal{R}')$ of \textsc{VGridSpecialRect-$3$SC}
        \item[] \textbf{Output: } A piecewise multivariate polynomial $f: [2N]^2 \rightarrow \mathbb{R}$.
    \end{itemize}
    \begin{enumerate}
         \item \textbf{For }{$R \in \mathcal{R}'$: }
    \begin{enumerate}
        \item[] \textbf{If }{$\exists R' \in \mathcal{R}'$ and $R'\neq R$ and $R' \cap R \neq \emptyset$: }
        \begin{enumerate}
        \item[] $\mathcal{I}\gets \mathcal{I}\cup R$
        \end{enumerate}
        \item[] \textbf{Else }
        \begin{enumerate}
        \item[] $\mathcal{I}_{\mathcal{R}'}^c\gets \mathcal{I}_{\mathcal{R}'}^c\cup \{R\}$
        \end{enumerate}
    \end{enumerate}
        \item Let $f(x,y)\gets 0, \forall (x,y) \in {\mathcal{I}}$
        \item \textbf{For }{$(i,R_i) \in \mathbb{N}\times \mathcal{I}_{\mathcal{R}'}^c$ (In Order): }
        \begin{enumerate}
            \item[] $f\mid_{R_i}\gets i$
        \end{enumerate}
        \item \textbf{Return }f
    \end{enumerate}
\end{boxalgorithm}
Algorithm \ref{alg: multi-poly-construction} collects the following two disjoint sets:
\begin{enumerate}
    \item The set $\mathcal{I}$ is the union of all rectangles that intersect with some other rectangle from $\mathcal{R}'$.
    \item The collection $\mathcal{I}^c_{\mathcal{R}'}$ consists of all rectangles from $\mathcal{R}'$ not intersecting with $\mathcal{I}$.
\end{enumerate}
It then sets the multivariate piecewise polynomial $f$ to $0$ on $\mathcal{I}$ and to some different integer on each of the remaining rectangles. We have the following properties of this Algorithm.
\begin{proposition}
   Algorithm \ref{alg: multi-poly-construction} correctly outputs a piecewise multivariate polynomial. It has a run time of $O(N^2)$. 
\end{proposition}
\begin{proof}
\textbf{Correctness: }

By the If-Else statement of Algorithm \ref{alg: multi-poly-construction}, the set $\mathcal{I}$ and all rectangles from the collection $\mathcal{I}_{\mathcal{R}'}^c$ must be disjoint.

Furthermore, inside the For-loop the function $f$ is  either obtaining an integer value on a new coordinate $(x,y) \in [2N]^2$ or is overwriting  values already assigned by enumerating the rectangles in $\mathcal{I}_{\mathcal{R}'}^c$. Due to the overwriting, the function $f$ is well defined. Thus,  the function $f$ always maintains that it is defined on a partition of a subset of $[2N]^2$. 

We can prove this by induction:

\textbf{Base Case: }

The set $f(\mathcal{I})=0$ is a well defined partition.

\textbf{Induction Step: }

Let $I_n$ be the domain of $f$ defined after $n$ steps of the For-loop.

For the next rectangle $R_{n+1}$. Following the overwrite step inside the For-Loop, $R_{n+1}\gets n+1$ all previous values from $I_n\setminus R_{n+1}$ stay the same. Thus, we have that $I_{n+1}$ is a partition.

    We can check that the output $f$ is a piecewise polynomial since it is a constant integer on a partition.

\textbf{Complexity: }

The For loop takes $N$ time, the inner If takes $O(N)$ time without a data structure.  

For the assignment to $f$, we can view this as a concatentation of  $\lvert \mathcal{I}_{\mathcal{R}'}^c\rvert$ many piecewise functions to a string representation of the piecewise multivariate polynomial. This concatenation process takes time $O(1)$ for a rectangle and is performed $O(N)$ times.  

Thus the algorithm takes $O(N^2)$ in total.
\end{proof}
\begin{remark}
Algorithm \ref{alg: multi-poly-construction} can be run in time $O(N\log(N))$ with an appropriate data structure, such as an interval tree, for checking rectangle intersection queries. See~\cite{edelsbrunner1983new, bentley1980optimal} for rectangle intersection algorithms. 
\end{remark}
\begin{theorem}\label{thm: vgrid2rectlossyvfcomp}
\begin{equation}\textsc{VGridSpecialRect-3SC} \leq_{\text{PTAS}}(k,D)\textsc{-RectLossyVVFCompression} \text{ if $k,D\geq 2$.}
     \end{equation}
\end{theorem}
\begin{proof}
    We define the following PTAS reduction from \textsc{VGridSpecialRect-3SC} to \\$(2,2)\textsc{-RectLossyVVFCompression}$:
    \begin{enumerate}
        \item Let the reduction map $m: I_{\textsc{VGridSpecialRect-3SC}} \rightarrow I_{(2,2)\textsc{-RectLossyVVFCompression}}$ map as follows:
    \begin{equation}
        m: ([2N]^2,\mathcal{R}')\mapsto (\Psi, f_{\mathcal{R}'}, 0.1)
    \end{equation}
    where $\Psi: [2N]^2 \rightarrow [2N]^2$ is the identity map on $[2N]^2 \subseteq \mathbb{Z}^2$. The function $f_{\mathcal{R}'}: [2N]^2 \rightarrow \mathbb{R}$ is as determined by Algorithm \ref{alg: multi-poly-construction}. 
    \item Let $\alpha(\bullet)\triangleq {\bullet}$
    \item Let $x=([2N]^2,\mathcal{R}')$ be an instance of $(2,2)$\textsc{-RectLossyVVFCompression}. 
    
    Let $s=\mathcal{S}$ be a feasible solution to the $(2,2)$\textsc{-RectLossyVVFCompression} on instance $m(x)$. Let
    \begin{equation}
         g(x,s)\triangleq s
    \end{equation}
    \end{enumerate}

\textbf{The Optimal Solution Doesn't Change: }

Certainly the optimal solution to $x=([2N]^2,\mathcal{R}')$ from \textsc{VGridSpecialRect-$3$SC} is $\mathcal{O} \cup H^*$. According to Proposition \ref{prop: subcover-opt}, this consists of the union of the optimal cover of $\mathcal{U}(\mathcal{R}'')$ and $H(\mathcal{R}'')$. 

According to Algorithm \ref{alg: multi-poly-construction}, we have that we can construct a piece-wise multivariate polynomial where the rectangles from $\mathcal{I}$ take on value $0$ and all other rectangles are assigned some positive integer value.

According to Proposition \ref{prop: max-VR-threshold}, $\textsf{iVR}^{d_f}_{2\varepsilon}(\Psi)$ contains the optimal solution. The rectangles in the optimal solution cannot have diameter larger than $2\varepsilon$. We set the distortion bound to $\varepsilon=0.1$ through the reduction map $m$. Thus, $2\varepsilon=0.2$ is smaller than the closest function value difference of $1$ between any adjacent regions. We can thus say that the rectangles from  $H(\mathcal{R}'')$ cannot combine to form a feasible cluster of $\textsf{iVR}^{d_f}_{2\varepsilon}(\Psi)$. This means the optimal solution to \textsc{VGridSpecialRect-$3$SC} consists of the optimal cover of $\mathcal{I}$ and $H(\mathcal{R}'')$. 

Certainly $\mathcal{I} \subseteq \mathcal{U}(\mathcal{R}'')$ since 
$H(\mathcal{R}'')=[2N]^2\setminus \mathcal{U}(\mathcal{R}'') \subseteq [2N]^2\setminus \mathcal{I}$. Certainly all rectangles from $H(\mathcal{R}'')$ are disjoint and thus belong to $[2N]^2\setminus \mathcal{I}$. 

We also know that $\mathcal{U}(\mathcal{R}'') \setminus \mathcal{I}$ consists of disjoint rectangles according to the definition of $\mathcal{I}$. We know that a cover of a union of rectangles, no two intersecting except at their boundaries, is the collection of rectangles themselves. 

Thus, $\mathcal{O}$ consists of the union of an optimal cover $\mathcal{O}'$ of $\mathcal{I}$ and the collection \begin{equation}
    C(\mathcal{I})=\{R \text{ is a rectangle}:R \subseteq \mathcal{U}(\mathcal{R}'') \setminus \mathcal{I}\}
\end{equation}
Since $\mathcal{O}$ is of minimum size as a cover of $\mathcal{U}(\mathcal{R}'')$, $\mathcal{O}'$ is also of minimum size. 

Thus, the optimal cover of $m(x)=(\Psi,f_{\mathcal{R}'},0.1)$ is  
\begin{equation}
    \mathcal{O}' \cup C(\mathcal{I})\cup H(\mathcal{R}'')=\mathcal{O}\cup H(\mathcal{R}'')
\end{equation}
This is the same as the solution to $x$.

   \textbf{The Preservation of the PTAS from \textsc{SpecialRect-$3$SC} to \textsc{VGridSpecialRect-3SC}: }
    
    Let $\mathcal{A}$ be a $(1+\alpha(\epsilon))$-approximation algorithm to $(2,2)$\textsc{-RectLossyVVFCompression}. This means: 
\begin{equation}
    \lvert \mathcal{A}(m(x))\rvert \leq (1+\alpha(\epsilon))\lvert \mathcal{S}^*\rvert 
\end{equation}
where $\mathcal{S}^*$ is an optimal minimizing solution to $(2,2)$\textsc{-RectLossyVVFCompression}. 

We then have that:
\begin{equation}
    \lvert g(x, \mathcal{A}(m(x))) \rvert =\lvert \mathcal{A}(m(x))\rvert
    \leq (1+\epsilon)\lvert (\mathcal{S}')^*\rvert
\end{equation}
where $\lvert \mathcal{S}^*\rvert =\lvert (\mathcal{S}')^*\rvert $ according to what we have proven above.

According to Proposition \ref{prop: k'kD'Dsubprob}, the instances of $(2,2)$\textsc{-RectLossyVVFCompression} can be viewed as instances of $(k,D)$\textsc{-RectLossyVVFCompression} without changing the solution. Thus, there is a PTAS reduction from $(2,2)$\textsc{-RectLossyVVFCompression} to   $(k,D)$\textsc{-RectLossyVVFCompression}:
\begin{enumerate}
    \item Let the reduction map be the embedding map given in Proposition \ref{prop: k'kD'Dsubprob}.
    \item Let $\alpha(\bullet)\triangleq \bullet$
    \item Let $g(\bullet,s)\triangleq \{(\pi^{0}_1(R),\pi^{0}_2(R)): R \in s\}$. 
\end{enumerate}
Certainly for an algorithm $\mathcal{A}$ solving $(k,D)$\textsc{-RectLossyVVFCompression}, if:
\begin{equation}
    \lvert \mathcal{A}(m(x)) \rvert \leq (1+\alpha(\epsilon))\lvert \mathcal{S}^*\rvert
\end{equation}
then:
\begin{equation}
    \lvert g(x, \mathcal{A}(m(x))) \rvert =\lvert \mathcal{A}(m(x))\rvert
    \leq (1+\epsilon)\lvert (\mathcal{S}')^*\rvert
\end{equation}
where the first equality follows by injectivity of the projection map $g$ on $\mathcal{A}(m(x))$ and the last inequality follows from $\lvert \mathcal{S}^*\rvert =\lvert (\mathcal{S}')^*\rvert $. 

Composing the two PTAS reductions gives the desired result.
\end{proof}
We can now prove that $(k,D)$\textsc{-RectLossyVVFCompression} is APX-Hard. 
\begin{theorem}    $(k,D)$\textsc{-RectLossyVVFCompression} is APX-Hard for $k,D\geq 2$
\end{theorem}
\begin{proof}
This follows by composing the two reductions from Theorem \ref{thm: R2ZPTAS-reduction} and Theorem \ref{thm: vgrid2rectlossyvfcomp}  along with the fact that \textsc{SpecialRect-$3$SC} is APX-Hard.
\end{proof}
\ifshowline
\section*{Acknowledgment}
The author was supported by a postdoctoral appointment with Hanqi Guo.
\fi
 \bibliographystyle{alphaurl} 
\bibliography{citations} 
\clearpage
\appendix
\section{Appendix}
\subsection{Metric Spaces}\label{sec: metric-spaces}
\begin{definition}
    On a set $X$, a metric is defined as a function $d: X \times X \rightarrow \mathbb{R}$ with the following three properties:
\begin{enumerate}
    \item $d(x,x)=0, \forall x \in X$
    \item $d(x,y)=d(y,x), \forall x,y \in X$
    \item $d(x,z) \leq d(x,y)+d(y,z), \forall x,y,z \in X$
\end{enumerate}
\end{definition}
\begin{definition}
   For $p\in \mathbb{N}$, let the $l_p$-norm of $\mathbb{R}^{D}, D \in \mathbb{N}$ be defined as:
\begin{equation}
    \lvert (x_1,...,x_D) \rvert_p \triangleq \sqrt[{p}]{\sum_{i=1}^{D} x_i^p}
\end{equation} 
\end{definition}
\begin{observation}
The $l_p$-norm forms a metric called the $l_p$-metric as follows:
\begin{equation}
    (x,y) \mapsto \lvert x-y \rvert_p
\end{equation}
\end{observation}
\begin{example}
    The Euclidean distance on $\mathbb{R}^D$ is the metric induced by the $l_2$-norm.
\end{example}
\subsection{Vector Algebra}\label{sec: vector-algebra}
\begin{definition}
    An abelian group $(S,+)$ is a set $S$ with the following properties:
    \begin{enumerate}
        \item $0 \in S$
        \item $\forall x \in S, \exists! y \in S: x+y=0$
        \item $\forall x,y \in S: x+y \in S$
        \item $\forall x,y \in S: x+y=y+x$
    \end{enumerate}
\end{definition}
\begin{definition}
A $\mathbb{Z}$-module is a set $M$ with two operations $(+,\cdot)$ that satisfy the following two properties:
\begin{enumerate}
\item $M$ is an abelian group
\item $\cdot: \mathbb{Z} \times M \rightarrow M$ acts as follows:
\begin{equation}
    r\cdot x= \overbrace{x+\cdots+ x}^{r\text{ times}}
\end{equation}
\end{enumerate}
\end{definition}
\begin{definition}
Let $M$ be a $\mathbb{Z}$-module.

    A set of elements $m_1,...,m_k \in M$ are \textbf{dependent} if:
    \begin{equation}
        \exists r_1,...,r_k \in \mathbb{Z}: \sum_{i=1}^k r_i m_i=0
    \end{equation}
    If no such $r_i \in \mathbb{Z}$ exist, then the $m_i \in M: i=1,...,k$ are \textbf{independent}. 
\end{definition}
\begin{example}
    The integers form a $\mathbb{Z}$-module over itself.
\end{example}
\begin{proposition}\label{prop: span-Zk}
The coordinate space $\mathbb{Z}^k$ forms a $\mathbb{Z}$-module and is generated by forming integer combinations with the standard generating vectors $e_q, q \in [k]$ as follows:
    \begin{equation}
        \mathbb{Z}^k=\{\sum_{q\in [k]}v_q e_q: v_q \in \mathbb{Z}\}
    \end{equation}
\end{proposition}
Due to Proposition \ref{prop: span-Zk}, for each vector $v  \in \mathbb{Z}^k$ we denote $v_q \in \mathbb{Z}$ as the $q$-th entry of $v$.

\subsection{Languages and Problems}
\begin{definition}
A language $L$ is a set of strings in some alphabet $\Sigma$. If not specified, we assume $\Sigma=\{0,1\}$, the set of bits. We assume that any data can have some encoding into a bit string. 
\end{definition}
\begin{definition}
    A prefix-free language $L$ is a language where no string $x \in L: x=zy$ ($z$ concatenated with $y$ on the right hand side) has $z \in L$.
\end{definition}
\begin{definition}
    A Turing machine is a $7$ tuple $M=\langle Q,\Gamma ,b,\Sigma ,\delta ,q_{0},F \rangle$ where:
    \begin{enumerate}
        \item $Q$ is a finite, non-empty set of states.
        \item $\Gamma$ is a finite, non-empty set of tape alphabet symbols;
        \item $b\in \Gamma$ is the blank symbol.
        \item $\Sigma \subseteq \Gamma \setminus \{b\}$
        \item $\delta :(Q\setminus F)\times \Gamma \rightarrow Q\times \Gamma \times \{L,R\}$ maps a nonterminal state to another state and moves either left (L) or right (R) on its tape. If the input pair is not recognized, terminate. 
        \item $q_0 \in Q$ is the initial state. 
        \item $F\subseteq Q$ is the set of final states or accepting states. The initial tape contents is said to be accepted by $M$ if it eventually halts in a state from $F$.
    \end{enumerate}
\end{definition}
\begin{definition}
A decidable language $L$ is a language that has some Turing machine $M$ so that for any string $x$, membership of $x \in L$ can be decided in finite time.  
\end{definition}
\begin{definition}
An undecidable language is a language that is not decidable.
\end{definition}
\begin{definition}
    A \textbf{Turing reduction} from language $L_1$ to language $L_2$, denoted $L_1 \leq_T L_2$ means that to decide membership in $L_1$ by use of a membership oracle for deciding membership in $L_2$.
\end{definition}
\begin{definition}
    For two languages $L_1,L_2$ if $L_1\leq_T L_2$ and $L_2\leq_T L_1$, then $L_1$ and $L_2$ are called Turing equivalent.
\end{definition}
\subsection{Kolmogorov Complexity and Universal Turing Machines}
\begin{definition}
    A universal Turing machine $U_T$ is a Turing machine that accepts inputs of the form:
\begin{equation}
    P;x
\end{equation}
where $P$ is a program and $x$ is an input string and outputs the output of $P$ on $x$:
\begin{equation}
U_T(P;x)=P(x)
\end{equation}
\end{definition}
\begin{definition}
    A universal Turing machine is called a prefix universal Turing machine if the language it accepts is prefix-free.
\end{definition}
\begin{definition}
The conditional Kolmogorov complexity for a pair of bit strings $(x,y)$ is the shortest length code $Q^*$ where
\begin{equation}
    U_T(Q^*)=U_T(P;x)=y
\end{equation}
This is denoted by:
\begin{equation}
    K_T(y\mid x)\triangleq \textsf{size}(Q^*)
\end{equation}
\end{definition}
\begin{lemma}\label{lemma: undecidable-kolmogorov}
Let $\mathcal{X}$ be a set of binary codes and let $x \in \mathcal{X}$.

If $\mathcal{X}$ is infinite, it is undecidable to compute the optimal $Q^*$ so that:
\begin{equation}
    U_T(Q^*(x,\hat{x}))=U_T(P;x)=\hat{x}
\end{equation}
for some $\hat{x} \in \mathcal{X}$.
\end{lemma}
\begin{proof}
The length $\textsf{size}(Q^*(x,\hat{x}))$ is exactly $K_T(\hat{x}\mid x)$. 
Thus we only require a standard textbook proof of undecidability of computing $K_T(y\mid x)$ for infinite number of $y$:

    For any $n \in \mathbb{N}$, consider the following set:
    \begin{equation}
        S=\{\hat{x} \in \mathcal{X}: \textsf{size}(Q^*(x,\hat{x}))<n\}
    \end{equation}
    We know that:
    \begin{equation}
        \lvert S\rvert \leq 2^n-1
    \end{equation}
     since for all  $\hat{x} \in S$ there is a code $Q^*(\hat{x}\mid {x})$ of length  strictly less than $n$. Counting all such codes gives the inequality.

     There are exactly $2^n$ strings of length $n$. Thus, by Pigeonhole on $\mathcal{X}$ there is a string $\hat{x}_n$ with $\textsf{size}(\hat{x}_n)=n$. Thus:
     \begin{equation}
         Q^*(x,\hat{x}_n)\geq n
     \end{equation}

     Consider the following program on an input $n \in \mathbb{N}$:
     
     \begin{enumerate}
     \item $z\gets \epsilon$ (empty string)
         \item For all binary strings  $\hat{x}$ up to length $n$:
         \begin{enumerate}
             \item[] If $Q>Q^*(x,\hat{x}))$ then $Q \gets Q^*(x,\hat{x}), z \gets \hat{x}$
         \end{enumerate}
         \item Output $z$
     \end{enumerate}

     This program computes $\hat{x}_n$ and can be encoded with a constant $C$ bits with $O(\log(n))$ bits to encode $n$.

     Thus: 
     \begin{equation}
         Q^*(x,\hat{x}_n)\leq C+O(\log(n))
     \end{equation}
     This contradicts that for any $n \in \mathbb{N}$ there exists $\hat{x}_n$ with 
     \begin{equation}
        Q^*(x,\hat{x}_n)\geq n
     \end{equation}
\end{proof}

\subsection{Asymptotic Analysis}\label{sec: asymptotics-appendix} 
We would like to measure the large scale behavior of a function. We can do this with big-O notation:
 \begin{definition}
 \textbf{Big-O Asymptotics}
 
 For two functions $g: \mathbb{N} \rightarrow \mathbb{N}$ and $f: \mathbb{N} \rightarrow \mathbb{R}$ we have that:
 \begin{equation}
     g(n)=O(f(n)) \text{ iff }
 \end{equation}
 \begin{equation} \exists C \in \mathbb{R}: C>0, \exists N \in \mathbb{N}, \forall n \geq N, g(n) \leq C f(n)
 \end{equation}
  \end{definition}
    We would like the function $f: \mathbb{N} \rightarrow \mathbb{R}$ in the Big-O notation to be simple. We will use polynomials to do this.
  \begin{definition}
In asymptotics, we assume positive coefficient \textbf{polynomials}  of the form:
\begin{equation}
    f(x)= \sum_{k_i \in \mathbb{R}_{> 0}: i=0,...,d} c_i x^{k_i}, c_i \in \mathbb{R}_{\geq 0}
\end{equation}
The positive real number $\max_{i=1}^d k_i$ is called the \textbf{degree} of the polynomial $f(x)$. 
\end{definition}

 A polynomial composed with a logarithm is called a ``\textbf{polylog}."
\begin{definition}
 The notation 
 \begin{equation}
     g(n)=\tilde{O}(f(n))
 \end{equation} means:
 \begin{equation} \exists \text{a polylog }p, \exists C \in \mathbb{R}: C>0, \exists N \in \mathbb{N}, \forall n \geq N, g(n) \leq C f(n) p(n)
 \end{equation}
 \end{definition}
\subsubsection{Terminology and Notation for Algorithms and Complexity Theory}\label{sec: algorithms-appendix}
 \begin{definition}
An \textbf{optimization problem} $\mathcal{P}=(I,S,c)$ is determined by a space of input instances $I_P$, a space of feasible solutions $S_P$, and a cost function $c_P: S_P \rightarrow \mathbb{N}$. 

The maximization/minimization objective is for every instance $x \in I_P$ to find some solution $s_P(x) \in S_P$  where:
\begin{enumerate}
    \item $s_P: I_P \rightarrow S_P$ is some (nonunique) solution map from instances to solutions and
    \item $c_P(s_P(x))$ is maximized/minimized. 
\end{enumerate}
\end{definition}
\begin{definition}
   A \textbf{decision problem} $\mathcal{P}=(I_P,S_P)$ is determined by a space of input instances $I_P$ and a space of feasible solutions $S_P$ where every $x \in I_P$ has some (atleast one) corresponding solution $s_P(x) \in S$ for a (nonunique) map $s_P: I_P \rightarrow S_P$.
   
   The decision objective is for every instance $x \in I$ to decide (true or false) whether there is a solution $s_P(x) \in S$ corresponding to $x$.
\end{definition}
\begin{proposition}(Problem Language Correspondence)
    Any decision problem $\mathcal{P}=(I_P,S_P)$ is Turing equivalent to deciding membership in the language:
\begin{equation}
    L_{\mathcal{P}}:=\{x \in I_P: s_P(x) \in S_P\text{ for some }s_P: I_P\rightarrow S_P\}
\end{equation}
\end{proposition}
\begin{proof}
   We know that:
    \begin{equation}
        x\in L_{\mathcal{P}} \text{ iff } (x\in I_P\text{ and } s_P(x) \in S_P)
    \end{equation}
We now have that:

     For any instance $x \in I_P$, any membership oracle for $L_{\mathcal{P}}$ decides the decision objective of $\mathcal{P}$.
     
    For any input string $x$, membership in $L_{\mathcal{P}}$ can be decided by a membership oracle for the decision objective $\mathcal{P}$.
\end{proof}
\begin{proposition}\label{prop: NP-Reduce2-NPO}
Let $\mathcal{P}=(I,S,c)$ be a maximization problem. 
Deciding membership in the  language:
\begin{equation}
    L_{\mathcal{P}_{\max}}:=\{(x,K) \in I\times \mathbb{N}: s_P(x) \in S_P, c_P(s_P(x)) \geq K \text{ for some }s_P: I_P\rightarrow S_P\}
\end{equation}
 Turing reduces to $\mathcal{P}$.
 
Let $\mathcal{P}=(I,S,c)$ be a minimization problem. 
Deciding membership in the  language:
\begin{equation}
    L_{\mathcal{P}_{\min}}:=\{(x,K) \in I\times \mathbb{N}: s_P(x) \in S_P, c_P(s_P(x)) \leq K \text{ for some }s_P: I_P\rightarrow S_P\}
\end{equation}
 Turing reduces to $\mathcal{P}$.
\end{proposition}
\begin{proof}
    The proof is identical to that from Proposition \ref{prop: NP-Reduce2-NPO}. 
\end{proof}
\begin{definition}
    A \textbf{subproblem} of a decision problem $\mathcal{P}=(I_P,S_P)$ is a decision problem $\mathcal{Q}=(J_Q,T_Q)$ where:
    \begin{enumerate}
        \item $J_Q\subseteq I_P, T_Q \subseteq S_P$ and
        \item $s_P(x) \in T_Q,  \forall x \in J_Q$.
    \end{enumerate} 
    Similarly, a \textbf{subproblem} of an optimization problem $\mathcal{P}=(I_P,S_P,c_P)$ is the optimization problem $\mathcal{Q}=(J_Q,T_Q,c_P)$ where:
    \begin{enumerate}
        \item $J_Q\subseteq I_P, T_Q \subseteq S_P$ and
        \item $s_P(x) \in T_Q,  \forall x \in J_Q$ where $s_P$ satisfies (nonuniquely) the optimization objective of $\mathcal{P}$.
    \end{enumerate}
\end{definition}
\begin{proposition}\label{prop: subproblem-Treduce}
    For decision problem $\mathcal{P}=(I_P,S_P)$ any subproblem $\mathcal{Q}=(I_Q,S_Q)$ is Turing reducible to $\mathcal{P}$.
\end{proposition}
\begin{proof}
    With a membership oracle for $L_{\mathcal{P}}=\{x \in I_P: \exists s_P, s_P(x) \in S_P$, for any input $x \in I_Q$ we can decide $x \in L_{\mathcal{P}}$ using the membership oracle from $L_{\mathcal{P}}$ since $L_{\mathcal{Q}} \subseteq L_{\mathcal{P}}$.
\end{proof}
\begin{definition}
    For problem $\mathcal{P}=(I,S)$, the \textbf{input size} $\textsf{size}(x)$ of an instance $x \in I$ is determined by the number of bits to represent the input.
\end{definition}
For any problem, we are interested in designing \textbf{algorithms} that take input instances of a problem and compute some feasible solution, called the output, for the input instance.
\begin{definition}
      A \textbf{deterministic algorithm} $\mathcal{A}$ is a sequence of operations that computes a function from a set $I$ to a set $S$ where each $x \in I$ is called an input and $\mathcal{A}(x) \in S$ is called the output computed by the algorithm.
\end{definition}
We can measure the computational ``cost" of an algorithm, which we call its complexity.
\begin{definition}
The \textbf{complexity} of a deterministic algorithm is the number of operations that the algorithm performs in order to compute its output, as a function of the input size.

The \textbf{asymptotic complexity} is the complexity in terms of asymptotic behavior such as Big-O.
\end{definition}
\begin{definition}
    A deterministic algorithm \emph{runs in polynomial time} if its asymptotic complexity is $O(q(n))$ for $q$ a polynomial.
\end{definition}
\begin{proposition}\label{prop: NP-Reduce2-NPO}
   If the corresponding decision problem $\mathcal{P}_D$ of an NPO problem $\mathcal{P}$ is NP-Hard, then $\mathcal{P}$  is NP-Hard.
\end{proposition}
\begin{proof}
   Consider an  intermediate problem $\mathcal{P}_D^L$ which has the same instance space $I\times \mathbb{N}$ of $\mathcal{P}_D$. 
   
   For problem $\mathcal{P}_D^L$, given an instance $(x,k) \in I \times \mathbb{N}$ its solution is a $O(\log(k))$ lengthed  tuple:\begin{equation}
       \tilde{s}_k\triangleq (s_k) \| (s_{k'}(x))_{k' \in T: T \subseteq \mathbb{N}}
   \end{equation}
   of the computational history  $T \subseteq \mathbb{N}$ of the parameter $k$ when  running a binary search on:
   \begin{enumerate}
       \item  Minimization: 
       \begin{enumerate}
           \item $[k^*,k^*+k]$ if $k<k^*$
           \item $[1,k]$ if $k>k^*$ 
       \end{enumerate}
       \item Maximization:
       \begin{enumerate}
           \item $[k^*-k,k^*]$ if $k<k^*$
           \item $[1,k]$ if $k>k^*$ 
       \end{enumerate}
   \end{enumerate}
   using the membership oracle from $\mathcal{P}_D$
   until the optimal solution $k^*$ from an oracle from $\mathcal{P}$ is reached.
   
  1. We have that $\mathcal{P}_D \leq_p \mathcal{P}_D^L$ by the identity reduction on the instance spaces:
  
  $(\Rightarrow): $ If the feasible solution for instance $(x,k)$ of $\mathcal{P}_D$ is a true solution, (e.g.  $c(s_k(x))\geq k$ for maximization), then the first entry of the corresponding solution to $(x,k)$ in $\mathcal{P}_D^L$ is a true solution. All subsequent solutions $s_{k'}(x)$ must be true solutions since the binary search maintains the truthiness of the solutions in its computational history.

  $(\Leftarrow): $ By projection on the first entry of $\tilde{s}_k$.
  
  2. We also have that $\mathcal{P}_D^L\leq_p \mathcal{P}$ by the reduction $(x,k) \mapsto x$:
  
  $(\Rightarrow):$ 
  If the feasible solution $\tilde{s}_k$ for instance $(x,k)$ in $\mathcal{P}_D^L$ is a true solution, then the binary search must deterministically converge to the optimal solution $s_{k^*}(x)$ in $O(\log(k))$ number of calls to $\mathcal{P}_D$. 

  $(\Leftarrow): $ 
  By projection on the last entry of $\tilde{s}_k$.
  \end{proof}
We discuss the terminology for algorithms and complexity theory. For notations such as asymptotics and problem formulations, see Section \ref{sec: asymptotics-appendix} and Section \ref{sec: algorithms-appendix} in the Appendix.
\begin{definition}
    We say a function $f: A \rightarrow B$ is \emph{polynomial time computable} if there exists  a deterministic algorithm running in polynomial time that computes $f$.
\end{definition}
A problem can also be \textbf{verified}. This means for a given instance-feasible solution pair $(x,s) \in I\times S$, we check (true/false) if $s$ is a corresponding solution to $x$ according to the problem.  
\begin{definition}
    A problem $\mathcal{P}=(I,S, \bullet)$ is in nondeterministic polynomial time (NP) if there is a polynomial time deterministic algorithm that can verify whether for a pair $(x,s) \in I\times S$ that there is some solution $s_P(x)$ with $s_P(x)=s$.

    We call the set of all such decision problems the \emph{complexity class} NP. 
    
    The set of all such optimization problems forms the \emph{complexity class} NPO.
\end{definition}
For any minimization/maximization problem $\mathcal{P}=(I,S,c)$ in NPO there is a \emph{corresponding decision problem} of the form $(I\times \mathbb{N},S)$ where true solutions to instances $(x,k) \in I \times \mathbb{N}$ satisfy:
\begin{enumerate}
    \item $c(s(x)) \leq k$ for a minimization problem 
    \item $c(s(x)) \geq k$ for a maximization problem. 
\end{enumerate}
\begin{proposition}
   If the corresponding decision problem $\mathcal{P}_D$ of an NPO problem $\mathcal{P}$ is NP-Hard, then $\mathcal{P}$  is NP-Hard.
\end{proposition}
For Proof, see Section \ref{sec: algorithms-appendix} in the Appendix.
\begin{definition}
    A Karp reduction \cite{garey2002computers} from decision problem $A=(I_A,S_A)$ to decision problem $B=(I_B,S_B)$, denoted $A \leq_p B$ is a polynomial time algorithm that computes a function $m: I_A\rightarrow I_B$ that takes instances of $A$ to instances of $B$ so that the following diagram \emph{commutes}:
    \begin{equation}
    \begin{tikzcd}
	{I_A} & {I_B} \\
	{S_A} & {S_B}
	\arrow["m", from=1-1, to=1-2]
	\arrow["{s_A}"', from=1-1, to=2-1]
	\arrow["{s_B}", from=1-2, to=2-2]
	\arrow["i", tail reversed, from=2-1, to=2-2]
\end{tikzcd}
\end{equation}
where $i: S_A \rightarrow S_B$ is an injection and $s_A: I_A \rightarrow S_A,s_B: I_B \rightarrow S_B$ are the (nonunique) solution maps.

This means:
\begin{equation}
    \exists i: S_A \rightarrow S_B \text{ injective }, \forall x \in I_A:
    s_B(m(x))=i(s_A(x))
\end{equation}
\end{definition}
We can define two more complexity classes. These are of importance for determining the theoretical intractability of a problem.
\begin{definition}
The decision problem $\mathcal{P}$ is NP-Complete if any problem in NP can be Karp reduced to $\mathcal{P}$.

A decision problem $\mathcal{P}$ is NP-Hard if there is a Karp reduction from a NP-Complete problem to $\mathcal{P}.$
\end{definition}
It is conjectured that NP-Hard problems are not feasible in practice to compute.

We can also design algorithms that approximate problems. We formalize this in the following:
\begin{definition}
    For a problem $\mathcal{P}=(I_P,S_P,c_P)$ in NPO, a deterministic algorithm $\mathcal{A}$ approximates a maximization/minimization problem $\mathcal{P}$ with \textbf{approximation factor} $f(n), n \in \mathbb{N}$ if for every $x \in I$ with $\textsf{size}(x)=n:$
\begin{equation}
    \frac{c_P(s_P(x))}{c_P(\mathcal{A}(x))}\leq f(n)
\end{equation}
If such an algorithm $\mathcal{A}$ exists, then we say that $\mathcal{P}$ is $f(n)$-APX. The algorithm is called an \emph{approximation algorithm}.
\end{definition}
\begin{definition}
A problem in NPO is \textbf{approximable (APX)} if there is a polynomial time deterministic algorithm with constant approximation factor.
\end{definition}
If we introduce a parameter $\epsilon>0$ so that the approximation factor of an approximation algorithm is $(1+\epsilon)$, we form another complexity class:
\begin{definition}
     A problem  in NPO having a polynomial time deterministic algorithm with an approximation factor $(1+\epsilon)$ is called a \textbf{polynomial time approximation scheme (PTAS)}.
\end{definition}
\begin{proposition}
The following tower of complexity classes holds:
    \begin{equation}
        \text{PTAS}\subseteq\text{APX}\subseteq\text{NPO}
    \end{equation}
\end{proposition}
\begin{proof}
    Check the definition of the classes.
\end{proof}
We define a reduction for the APX complexity class that is used analogously to the Karp reductions for the complexity class NP.
\begin{definition}
Let $\mathcal{P}=(I_P,S_P,c_P),\mathcal{Q}=(I_Q,S_Q,c_Q)$ be two problems in NPO.

    A PTAS reduction from $\mathcal{P}$ to $\mathcal{Q}$, denoted: $\mathcal{P}\leq_{PTAS}\mathcal{Q}$ has:
    \begin{enumerate}
        \item (Reduction Map) A polynomial time computable function $m: I_P\rightarrow I_Q$.
        \item (Approximation Factor) $\alpha: \mathbb{R}_{> 0} \rightarrow \mathbb{R}_{> 0}$ 
        \item (Solution Map) A polynomial time computable function $g: I_P \times S_Q \rightarrow S_P$ so that the following diagram:
        \begin{equation}
           \begin{tikzcd}
	{I_P} & {I_Q} && {} \\
	{S_P} & {I_P \times S_Q} & {S_Q}
	\arrow["m", from=1-1, to=1-2]
	\arrow["{{{s_P}}}"', from=1-1, to=2-1]
	\arrow["{\text{id}\times \vardiamond}"{description}, hook, from=1-1, to=2-2]
	\arrow["{{{s_Q}}}", from=1-2, to=2-3]
	\arrow["{{\mathcal{A}}}"{description}, shift left=3, curve={height=-12pt}, from=1-2, to=2-3]
	\arrow["g", from=2-2, to=2-1]
	\arrow["{{\vardiamond\times \text{id}}}", hook', from=2-3, to=2-2]
\end{tikzcd}
        \end{equation}
        has the following relationship:
    \begin{equation}
    \begin{split}
       \text{$\forall x\in I_P$, If there is an Algorithm $\mathcal{A}$ satisfying: }\\
       \lvert \mathcal{A}(m(x))\rvert  \leq  (1+\alpha(\epsilon))\lvert s_Q(m(x))\rvert \\
       \text{ Then }
          \lvert g(x,\mathcal{A}(m(x)))\rvert  \leq (1+\epsilon) \lvert s_P(x) \rvert
        \end{split}
        \end{equation}
        for $s_Q, s_P$ (nonunique) solution maps for $\mathcal{P},\mathcal{Q}$ and where $\vardiamond$ is a placeholder map (any map can go here), and $\text{id}$ is the identity map.
    \end{enumerate}
\end{definition}
It is known that a composition of PTAS reductions is a PTAS reduction.
\begin{definition}
    We say that a problem is \textbf{APX-Hard} if any problem in APX PTAS reduces to it.

    A problem is \textbf{APX-Complete} if it is APX-Hard and belongs to APX.
\end{definition}
It is known that if $P\neq NP$ then any APX-Hard problem cannot be in PTAS by using the PCP theorem~\cite{dinur2007pcp} in a reduction from \textsc{$3$-SAT} to the APX-Complete problem \textsc{MAXSAT}.
 \subsection{The Voxel Grid Index Space}\label{sec: more-connected-shapes}
We define here an integer index space in $k$-dimensions. 
\begin{definition}\label{def: ortho-orientability}
A $k$-dimensional \textbf{voxel grid} $M$ is a set of points with the following Euclidean properties:
\begin{enumerate}
    \item All points $x \in M$ can be given a unique coordinate via the injective coordinate map: 
    \begin{equation}
        I: x \mapsto (i_1,...,i_k) \in \mathbb{Z}^k
    \end{equation}
    \item Each point $x \in M$ with coordinate $I(x)=(i_1,...,i_k)$ has a neighborhood of coordinates
    \begin{equation}
        N^M_{I(x)}\triangleq \{(i_1,...,i_j',..,i_k): i_j'=i_j+s, s \in \{-1,1\}, j=1,...,k\}
    \end{equation}
\end{enumerate}
\end{definition}
\begin{example}\label{ex: 2d-grid}
A simple example of a $2$-dimensional voxel grid $M$ is $[n_1]\times [n_2] \subseteq \mathbb{N}^2$.

Each point $(i,j) \in [n_1]\times [n_2]$ is its own coordinate. It has a neighborhood:
\begin{equation}
\{(i,j),(i-1,j),(i+1,j), (i,j-1), (i,j+1)\}
\end{equation}
The points $(0,\bullet_2), (n_1,\bullet_2), (\bullet_1, 0), (\bullet_1,n_2), \bullet_1 \in [n_1], \bullet_2 \in [n_2]$ are virtual points not on the grid but belonging to the neighborhood of each point on the boundary of the grid.
\end{example}
 \begin{definition}
    In a $k$-dimensional voxel grid $M$, we denote 
    \begin{equation}
        e_q\triangleq(\overbrace{0,...,0,1}^{q},\overbrace{0,...,0}^{k-q}), q \in [k]
    \end{equation}
    as a standard generating vector of $\mathbb{Z}^k$ in $k$ dimensions.
\end{definition}
We have shown that $\mathbb{Z}^k$ is a $\mathbb{Z}$-module. In particular, the vectors in $\mathbb{Z}^k$ can be added together just like those from $\mathbb{R}^k$.
 \subsubsection{Connected Shapes In Voxel Grids}\label{sec: Connected-Shapes}

We define a connected shape in terms of a sequence of points.
\begin{definition}
Let $M$ be a $k$-dimensional voxel grid.

    A sequence of points $(x_1,...,x_n) \in M^n,n \in \mathbb{N}$ is a \textbf{walk} on $M$ if \begin{equation}
        I(x_{i+1}) \in N_{I(x_i)}^M, \forall i=1,...,n-1
    \end{equation}
\end{definition}
For a set of points, we can use walks to define connectivity in $j$ dimensions for $j\leq k$:
\begin{definition}
Let $M$ be a $k$-dimensional voxel grid.

A nonempty set of points $S \subseteq M$ is \textbf{connected in $j$ dimensions}, $j\leq k$, if for any pair of points $x,y \in S$ there are $j$ sequences of integer coordinates $(I(z^p_i))_{i=1}^{n_p},  p=1,...,j, n_p \in \mathbb{N}$ satisfying the following:
\begin{enumerate}
\item (Endpoints at $x,y$) $z^p_1=x, z^p_{n_p}=y, z_i^p \in S, \forall p\in [j]$
    \item (Connected) The sequence $(z^p_i)_{i=1}^{n_p}$ is a walk on $M$ for all $p \in [j]$. 
    \item (Distinguishable) No pair of sequences intersect by sharing any points except $x$ and $y$. 
\end{enumerate} 
\end{definition}
We identify a connected set in one dimension and parameterize it by the dimension it has a nontrivial range in:
\begin{definition}
    A \textbf{line segment $L$ in dimension $q$} in a $k$-dimensional voxel grid $M$ is a nonempty set of points connected in one dimension where all of its indices only differ in a single dimension $q \in [k]$.
\end{definition}
Since a line segment in dimension $q$ can be parameterized by only a single variable, we can denote a line segment by its two endpoints:
\begin{subequations}
    \begin{equation}
        [a,b]_q=\{x \in M: (I(a))_q\leq (I(x))_q\leq (I(b))_q\}, a, b \in M; 
    \end{equation}
    \begin{equation}
        \text{ s.t. }I(b)-I(a)=d e_q, d \in \mathbb{N}\text{ for some } q \in [k]
    \end{equation}
    \end{subequations}
\subsubsection{Bounding Box}\label{sec: Boxes}
Given a set of points, we can define a simple set that covers all of these points through maximizing its size along all the dimensions.
\begin{definition}
Let $S$ be a subset of a $k$-dimensional \textbf{voxel grid} $M$. 

A \textbf{bounding box} of $S$ is the following region:
\begin{equation}
    \textsf{BB}(S)\triangleq \{x \in M: (\min_{s \in S} I(s))_q \leq (I(x))_q \leq (\max_{s \in S} I(s))_q, \forall q \in [k]\}
\end{equation}
A bounding box can also be defined by the same construction in Euclidean space $\mathbb{R}^k$.
\end{definition}
The dimensions of nontrivial length for the bounding box of a set of points $S \subseteq M$ are called the \textbf{axes} of $S$, denoted $J(S)$. More formally:
\begin{definition}
Let $S$ be a subset of a $k$-dimensional \textbf{voxel grid} $M$. 
    \begin{equation}
        J(S)=\{q \in [k]: (\min_{s \in \textsf{BB}(S)} I(s))_q<(\max_{s \in \textsf{BB}(S)} I(s))_q\}
    \end{equation}
\end{definition}
\begin{definition}\label{def: box}
    We say that a $j$-dimensional \textbf{box} is a set of points connected in $j$ dimensions whose bounding box is itself.
\end{definition}
\begin{definition}
    A $2$-dimensional box is called a \textbf{rectangle}.
\end{definition}
We define a ``corner" on a subset of a voxel grid. 
\begin{definition}
    Let $P$ be a set of $j$-connected points  in a $k$-dimensional \textbf{voxel grid} $M$. 
    
    A point $x \in P$ is called a $j$-dimensional \textbf{corner} if there is another unique point $y \in M$, called the \textbf{witness point}, whose index is reachable by a sum of $j$ independent unit vectors from $I(x) \in \mathbb{Z}^k$  satisfying:
    \begin{equation}
    \exists! y \in M \text{ with } I(y)=I(x)+\sum_{q \in J(P)} v_qe_q; v_q \in \{ -1,1\}
    \end{equation}
    \begin{enumerate}
        \item The corner $x \in P$ is \textbf{convex} if $\textsf{BB}(\{x,y\}) \subseteq P$  and 
        \item The corner $x \in P$ is \textbf{concave} if $\textsf{BB}(\{x,y\}) \subseteq P$ but $y \notin P$.
    \end{enumerate}
    
\end{definition}
We can now recover the usual definition of a corner from a three dimensional cube:
\begin{example}
    A three dimensional cube, which can be defined by the voxel grid $C:=(\{0,1\})^3$, has all of its $2^3$ points as corners where for corner $(c_1,c_2,c_3) \in C$ its witness is $(1-c_1,1-c_2,1-c_3) \in C$  
\end{example} 
Generalizing the three dimensional cube to $j$ dimensions in a voxel grid of dimension $k$, we can define unit boxes:
\begin{definition}
    We call a $j$-dimensional box with exactly $2^j$ convex corner points a $j$-dimensional \textbf{unit box}.
    \end{definition}
This allows us to define a more general shape:
\begin{definition}
For voxel grid $M$, a union of $j$-dimensional unit boxes $P \subseteq M$ connected in $j$ dimensions is called a \textbf{$j$-dimensional general orthogonal polyhedron}.   

We denote $\text{dim}(P)$ as the dimension $j$ of the general orthogonal polyhedron.
\end{definition}
In two dimensions, we call these orthogonal polygons.
When the dimension $j$ is arbitrary, we just say \textbf{general orthogonal polyhedron}. 

A general orthogonal polyhedron may have partial disconnections such as ``holes." These can be defined formally in the sense of homotopy~\cite{munkres2000topology}, namely through the ``fundamental group." We will not define this in general and will allow for such partial disconnections.

\textbf{Spanning a Box: }

According to Proposition \ref{prop: box-corner-properties} in the Appendix, the $2^j$ corners of a $j$-dimensional box completely determine it.  We show that only two corners are actually needed. 
\begin{observation}
    A \textbf{line segment} in a $k$-dimensional voxel grid $M$ is a $1$-dimensional box.
\end{observation}
 \begin{proposition}\label{prop: 2-corner-box-appendix}
For a $j$-dimensional box $R$ in a $k$-dimensional voxel grid $M$, there are exactly two opposing convex corners, denoted $c_1(R),c_2(R)$, in $R$ that can be used to recover all other points in $R$: 
\begin{equation}
    \textsf{BB}(c_1(R),c_2(R))=R
\end{equation}
with:
\begin{equation}\label{eq: min-corner}
        c_1(R)= \arg\min_{x \in R} \max_{q \in J(R)} (I(x))_q
    \end{equation}
    \begin{equation}\label{eq: max-corner}
        c_2(R)= \arg\max_{x \in R} \min_{q \in J(R)} (I(x))_q
    \end{equation}
\end{proposition}
 \begin{proof}
\textbf{1. We first show that $c_1(R)$ and $c_2(R)$ are unique solutions to their minimax and maximin equations respectively: } 

Without loss of generality, say that there are two solutions $c_1',c_1'' \in R, c_1'\neq c_1''$ to the minimax equation, Equation \ref{eq: min-corner}. 

Because $R$ is a box, we can take the max over every dimension of the two points $c_1',c_1'' \in R$ and still get a point in $R$: 
\begin{equation}
    c_1 \in R: (I(c_1))_q=\max_{x \in \{c_1',c_1''\}} (I(x))_q, \forall q \in [k]
\end{equation}
Since $c_1'\neq c_1''$ and each are minimax optima, we must have that there are two $q_1,q_2 \in [k], q_1\neq q_2$ with:
\begin{equation}
    (I(c_1'))_{q_1}< (I(c_1''))_{q_1}, (I(c_1'))_{q_2}> (I(c_1''))_{q_2}
\end{equation}
Thus, $c_1\neq c_1'$ and $c_1 \neq c_1''$. 

We can find $c_1$ over any pair of minimax optima $c_1',c_1''$ so long as there is more than one solution. This is a contradiction. 

Thus, there exists a unique solution to the maximin equation. Similarly, there is a unique solution to the maximin equation as well.

\textbf{2. We show that $c_1(R), c_2(R)$ are both convex corners: }

Certainly these are convex corners by the proof of Proposition \ref{prop: box-corner-properties}.

\textbf{3. We show that $c_1(R)$ and $c_2(R)$ can recover all of $R$: }

    We have that:
    \begin{equation}
        (I(c_1(R)))_q \leq (I(z))_q \leq (I(c_2(R)))_q, \forall q \in J(R)
    \end{equation}
    by definition of $c_1(R)$ and $c_2(R)$.
    
    Since these inequalities hold for each dimension separately, this means we can obtain $I(z)$ by means of addition with standard generating vectors with $I(c_1(R))$. 
    \begin{equation}
    \forall z \in R, \exists (d_q)_{q \in [k]} \in \mathbb{Z}^{k}:  d_q \leq  (I(c_2(R)))_q-(I(c_1(R)))_q  
\end{equation}
\begin{equation}
    \text{ s.t. }I(z)=I(c_1)+\sum_{q \in J(R)}d_qe_q
\end{equation}
\end{proof}
 \begin{proposition}\label{prop: box-corner-properties}
    For a $j$-dimensional box $R \subseteq M$ in a $k$-dimensional voxel grid $M$, every corner is \textbf{convex} and there are exactly $2^j$ of them.
\end{proposition}
\begin{proof}
    \textbf{1. We claim that there are atleast $2^j$ convex corners for a box $R$: }

    Since $R$ is a box, $BB(R)=R$. We can thus use the axes of $R$ in formulating the  following extreme points parameterized by subsets of $J(R)$:  
    \begin{equation}\begin{split}
        C(R)=\{x_Q \in R: \min_{z_Q \in R} (I(z_Q))_q=(I(x_Q))_q, \forall q \in Q \text{ and }\\\max_{z_Q \in R} (I(z_Q))_p=(I(x_Q))_p, \forall p \in J(R) \setminus Q; \forall Q \subseteq J(R)\}
        \end{split}
    \end{equation}
    These points have their coordinates take on  either mins or maxes at each dimension from $J(R)$.
    
    We can check that every point in $C(R)$ is a convex corner.
    
    Let $Q \subseteq J(R)$ be a subset of the axes of $R$. For each point $x_Q \in C(R)$ depending on the subset $Q \subseteq J(R)$, consider the construction of the following witness point depending on $Q$:
    \begin{equation}
    y_Q \in R: 
    \begin{cases}
        (I(y_Q))_q=(I(x_Q))_q+1 & q\in Q\\
        (I(y_Q))_q=(I(x_Q))_q-1 & q \in J(R)\setminus Q
        \end{cases}
    \end{equation}
    By definition of $R$ being a box, all boxes spanned by subsets of $R$ still belong to $R$. Thus, 
    \begin{equation}
        \textsf{BB}(\{y_Q,x_Q\})\subseteq R
    \end{equation}
    This is unique since all other points of the form $y_{Q'} \in R$ depending on a subset $Q' \subseteq J(R), Q' \neq Q$ 
    result in $y_{Q'}$ coming outside the box $R$. 
    
    We have two cases:
    \begin{enumerate}
        \item $\exists q \in Q' \setminus Q$:
        
        By the definition of $Q$ providing the dimensions for minimizing $I(y_{Q})$ and $J(R)\setminus Q$ providing the dimensions for maximizing $I(y_{Q})$. This means that $(I(y_{Q'}))_q$ will surpass the maximum possible value in $R$.
    \begin{equation}
        (I(y_{Q'}))_q >(\max_{z_Q \in R} (I(z_Q))_q
    \end{equation}
    \item $\exists q \in Q \setminus Q'$:

    Similarly, $(I(y_{Q}))_q$ will surpass the minimum possible value in $R$.
    \begin{equation}
        (I(y_{Q'}))_q <\min_{x_Q \in R} (I(x_Q))_q
    \end{equation}
    \end{enumerate}
    \textbf{2. We prove that there are no more convex corners: }

    We consider all the remaining points in $R$:

    Every remaining point $ x\in R \setminus C(R)$ satisfies the following inequalities:
    \begin{equation}
        \exists q^* \in J(R): \min_{x' \in R} (I(x'))_{q^*}<(I(x))_{q^*} <\max_{x' \in R} (I(x'))_{q^*}
    \end{equation}
    For any of these points $x \in R \setminus C(R)$, there are atleast two different points $y_1,y_2 \in R$ acting as potential witnesses to $x \in R$ being a convex corner:
    \begin{subequations}\label{eq: different witnesses}
    \begin{equation}
    (I(y_1))_{q^*}=(I(x))_{q^*}-1; 
    \end{equation}
    \begin{equation}
        (I(y_2))_{q^*}=(I(x))_{q^*}+1
    \end{equation}
    \begin{equation}
        (I(y_{\bullet}))_{q}=\begin{cases}
            (I(x))_{q}+1 & \text{ if }(I(x))_q<\max_{x' \in R} (I(x'))_q\\
            (I(x))_{q}-1& \text{ else }
        \end{cases}: q \neq q^*
    \end{equation}
    \end{subequations}
    This shows that $y_1\neq y_2$, thus $x$ cannot be a convex corner. A convex corner has a unique witness.
    
    Since $x$ was arbitrarily chosen from $R\setminus C(R)$, we must have that none of these points are convex corners.
    
\textbf{3. We prove that there are no more concave corners: }

    If $x \in R \setminus C(R)$ is a concave corner, then it has a unique witness $y \in M\setminus R$ where the witness must satisfy:
    \begin{subequations}
    \begin{equation}
        \textsf{BB}(\{x,y\}) \setminus\{y\} \subseteq R
    \end{equation}
    \begin{equation}
        I(y)=I(x)+\sum_{q \in J(R)}d_qe_q; d_q  \in \{-1,1\}
    \end{equation}
    \end{subequations}
    According to Equations \ref{eq: different witnesses}, for $ x\in R\setminus C(R)$ there are two different points $y_1,y_2 \in R$ that also satisfy this equation. They cannot be witnesses to concavity since they belong to $R$. We pick these two points $y_1,y_2 \in BB(\{x,y\})$ to have coordinates $I(y_1),I(y_2)$ that differ from $I(y)$ in the dimensions $q_1,q_2 \in J(R)$.
    \begin{subequations}
    \begin{equation}
        I(y_1)=I(x)+\sum_{q \in J(R)\setminus \{q_1\}}d_{q}e_q+d_{q_1}'e_{q_1}; d_{q_1}'\neq d_{q_1} 
    \end{equation}
    \begin{equation}
        I(y_2)=I(x)+\sum_{q \in J(R)\setminus \{q_2\}}d_{q}e_q+d_{q_2}'e_{q_2}; d_{q_2}'\neq d_{q_2}, q_2\neq q_1 
    \end{equation}
    \end{subequations}
    Because $R$ is a box, we know that all points $z \in R$ must satisfy: 
    \begin{equation}
        \min_{z' \in R}(I(z'))_q\leq (I(z))_q\leq \max_{z' \in R}(I(z'))_q, \forall q \in J(R)
    \end{equation}
    However, $y \in M \setminus R$ has that:
    \begin{subequations}
    \begin{equation}
        \min_{z' \in R}(I(z'))_q\leq \min((I(y_1))_q,(I(y_2))_q,(I(y))_q) 
    \end{equation}
    \begin{equation}
        \leq (I(y))_q  \leq \max((I(y_1))_q,(I(y_2))_q,(I(y))_q)
    \end{equation}
    \begin{equation}
          \leq \max_{z' \in R} (I(z'))_q, \forall q \in J(R)
    \end{equation}
    \end{subequations}
    This means that $y \in R$, this is a contradiction. Thus, $x \in R \setminus C(R)$ is not a concave corner.

    We have shown that all $x \in R \setminus C(R)$ cannot be either a convex or concave corner.

    Thus, $C(R)$ is the complete set of convex corners of $R$.
    \end{proof}
    \subsection{Properties of $(k,D)\textsc{-RectLossyVVFCompression}$}
    We can show the following hardness relationship amongst the parameter pairs $(k,D)$, namely that both the $(k,D)\textsc{-RectLossyVVFCompression}$ problem and the $(k,D)\textsc{-RectLossyVVFCompression}$  problem are atleast as hard as their lower dimensional version. 
\begin{proposition}\label{prop: k'kD'Dsubprob}
    Any instance of the $(k',D')$\textsc{-RectLossyVVFCompression} problem can be viewed as an instance of the $(k,D)$\textsc{-RectLossyVVFCompression} problem if $k\geq k', D\geq D'$.

    Similarly, this relationship is true between the $(k,D)$\textsc{-RectLossyVVFCompression} problem and the $(k',D')$\textsc{-RectLossyVVFCompression} problem if $k\geq k', D\geq D'$. 
\end{proposition}
\begin{proof}
    For instance $(\tilde{X},\tilde{f},\varepsilon)$, with $\tilde{X}: [n_1]\times ... \times [n_{k'}] \rightarrow \mathbb{R}^{D'}$, $\tilde{f}: \mathbb{R}^{D'} \rightarrow \mathbb{R}$ we can form an instance $(X,f,\varepsilon)$ of  $(k,D)$\textsc{-RectLossyVVFCompression} with $X: [n_1]\times ... \times [n_{k}] \rightarrow \mathbb{R}^{D}$, $f: \mathbb{R}^{D} \rightarrow \mathbb{R}$ by the following two injective embeddings:
    \begin{equation}
        X_{(i_1,...,i_k)}:=\begin{cases}
            \tilde{X}_{(i_1,...,i_{k'})} \|\overbrace{(0,...,0)}^{D-D'} \ & \text{ if } i_j=1, k'<j\leq k
            \\
            \overbrace{(0,...,0)}^{D}& \text{ else}
        \end{cases}
    \end{equation}
    \begin{equation}
        f(x):=\begin{cases}
            \tilde{f}(x)& \text{ if }x\in \mathbb{R}^{D'}\times \{0\}^{D-D'}\\
            0 & \text{else}
        \end{cases}
    \end{equation}
    Certainly the optimal solution in the $(k,D)$\textsc{-RectLossyVVFCompression} problem cannot involve the extended dimensions $k-k'$ and $D-D'$. Thus  there is an injection between it and the solution to the instance $(\tilde{X},\tilde{f},\varepsilon)$ of the $(k',D')$\textsc{-RectLossyVVFCompression} problem
\end{proof}
\begin{proposition}\label{prop: recbitsize}
\begin{equation}
    \textsf{size}(\textsf{RecTupleBit}(L_1,...,L_n))=r+\sum_{i=1}^n(\textsf{RecTupleBit}(L_i)+r)
\end{equation}
where $C$ is the number of bits to encode the delimiters. (e.g. $C=2$ if the delimeters are the comma and two parentheses.
\end{proposition}
\begin{proof}
    This follows by the recursive definition of \textsf{RecTupleBit} and the delimiting structure of a length $n$ tuple as $n-1$ commas and $2$ parentheses.
\end{proof}
\begin{proposition}\label{prop: appendix-sizes}
\begin{equation}
    \textsf{size}(\mathcal{C}(X,f,\varepsilon))= \lvert \mathcal{S}(X,f,\varepsilon)\rvert(10C+4\textsf{size}(n)+2v)+5C+\textsf{size}(\varepsilon^*)+\textsf{size}(\varepsilon)+\textsf{size}(f)
    \end{equation}
\end{proposition}
\begin{proof}
By repeated use of Propostion \ref{prop: recbitsize}:
\begin{subequations}
    \begin{equation}
        \textsf{size}(\textsf{Encode}(R))= C+2(\text{size}(n)+C)
        \end{equation}
        \begin{equation}
        \textsf{size}((\textsf{Encode}(R),\bullet_R))=C+2(v+\textsf{size}(\textsf{Encode}(R))+C)
    \end{equation}
    \begin{equation}
        \textsf{size}((\varepsilon^*,\varepsilon,f))=C+\textsf{size}(\varepsilon^*)+\textsf{size}(\varepsilon)+\textsf{size}(f)+3C
    \end{equation}
    \begin{equation}
    \begin{split}
        \textsf{size}(\mathcal{C}(X,f,\varepsilon))=(\lvert \mathcal{S}(X,f,\varepsilon)\rvert+1)C+\lvert \mathcal{S}(X,f,\varepsilon)\rvert \textsf{size}((\textsf{Encode}(R),\bullet_R))\\+4C+\textsf{size}(\varepsilon^*)+\textsf{size}(\varepsilon)+\textsf{size}(f)
        \end{split}
    \end{equation}
    \end{subequations}
    \begin{subequations}
    \begin{equation}
    \therefore \qquad \textsf{size}(((\textsf{Encode}(R),\bullet_R), \cdots )_{R \in \mathcal{S}(X,f,\varepsilon)})\end{equation}
    \begin{equation}
    \begin{split}
        =(\lvert \mathcal{S}(X,f,\varepsilon)\rvert+1)C+\lvert \mathcal{S}(X,f,\varepsilon)\rvert (C+2(v+C+2(\text{size}(n)+C)+C))\\
        +4C+\textsf{size}(\varepsilon^*)+\textsf{size}(\varepsilon)+\textsf{size}(f)
        \end{split}
    \end{equation}
    \begin{equation}
    =\lvert \mathcal{S}(X,f,\varepsilon)\rvert(10C+4\textsf{size}(n)+2v)+5C+\textsf{size}(\varepsilon^*)+\textsf{size}(\varepsilon)+\textsf{size}(f)
    \end{equation}
\end{subequations}
\end{proof}
    \subsection{Interval Tree Data Structure}\label{sec: interval-tree}
    An interval tree \cite{edelsbrunner1980dynamic} is a dynamic data structure that stores a collection of intervals on the real line. It has three operations that can be performed on it. Let $\mathcal{T}$ denote an interval tree. The three operations and their complexity are shown below:
    \begin{enumerate}
        \item (Insert):  $\textsf{Insert}(\mathcal{T},I): I= [a,b], a,b \in \mathbb{R}$
        \begin{enumerate}
            \item Adds an interval to the collection of intervals stored by $\mathcal{T}$. 
            \item This takes time $O(\log(n))$.
        \end{enumerate}
        \item (Delete):  $\textsf{Delete}(\mathcal{T},I): I= [a,b], a,b \in \mathbb{R}$
        \begin{enumerate}
            \item Delete an existing interval from $\mathcal{T}$. If there is none, then do nothing.
        \item This takes time $O(\log(n))$.\end{enumerate} 
        \item (Query) $\textsf{Query}(\mathcal{T},I): I= [a,b], a,b \in \mathbb{R}$
        \begin{enumerate}
            \item Queries for all possible intervals  in $\mathcal{T}$ intersecting with $I$ an existing interval from $\mathcal{T}$. 
            \item This takes time $O(\log(n)+k)$ where $k$ is the number of intervals returned by the query.
        \end{enumerate}
    \end{enumerate}
    \clearpage
\subsection{Subroutines for the the Sweep-Line Algorithm \ref{alg: sweep-line-alg} Implementation}
The first algorithm determines the adjacent neighboring intervals around an interval $I$ as determined by the nearest two intervals from $\mathcal{I}_1$. This is a subroutine of the Sweep Line Algorithm \ref{alg: sweep-line-alg}. 
\begin{algorithm}[!h]
\SetAlgoLined
\SetKwComment{Comment}{/* }{ */}
\caption{Adjacent Neighbors in Dimension $1$ of an Interval}\label{alg: adj-neighbors}
\SetKwFunction{FAdjNeigh}{AdjacentNeighbors}
\SetKwProg{Fn}{Function}{:}{}
\KwData{$(I,\mathcal{I}_1): $ an interval $I$ and a collection of intervals $\mathcal{I}_1$}
\KwResult{$(C_l,C_r): $ a pair of intervals. The interval $C_l$ is adjacent to $I$ on the left and in between $I$ and an interval from $\mathcal{I}_1$. The interval $C_r$ is adjacent to $I$ on the right and in between $I$ and an interval from $\mathcal{I}_1$.}
\Fn{\FAdjNeigh{$I,\mathcal{I}_1$}}{
$i_{\min}\gets \min(I)-1$\\
$i_{\max}\gets \max(I)+1$\\
$(I_{l},\text{``sign}_o^l\text{"}) \gets$ \text{BinarySearch}($1,N,i_{\min}, \mathcal{I}_1, -1$)\Comment{Binary Search for Closest Integer Key to $i_{\min}$ in Dimension $1$.}
\uIf{$\text{``sign}_o^l\text{"}=+1$}{
$C_l \gets [\max(I_l)+1,i_{\min}]$\Comment*[r]{The Interval Strictly in between the Left Endpoint of $I$ and the Right Endpoint $I_r$}
}\Else{
$C_l \gets [i_{\min}+1,i_{\min}+1]$ \Comment*[r]{The Left Endpoint of $I$}
}
$(I_{r},\text{``sign}_o^r\text{"}) \gets$ \text{BinarySearch}($1,N,i_{\max}, \mathcal{I}_1, +1$)\Comment{Binary Search for Closest Integer Key to $i_{\max}$ in Dimension $1$.}
\uIf{$\text{``sign}_o^l\text{"}=-1$}{
$C_r \gets [i_{\max}, \min(I_r)-1]$\Comment*[r]{The Interval Strictly in between the Right Endpoint of $I$ and the Left Endpoint $I_r$}
}\Else{
$C_r \gets [i_{\max}-1,i_{\max}-1]$ \Comment*[r]{The Right Endpoint of $I$}
}\Return $(C_l,C_r)$
}
\end{algorithm}
\clearpage
The second subroutine is the binary search subroutine of the Adjacent Neighbors Algorithm \ref{alg: adj-neighbors}
\begin{algorithm}[!h]
\SetAlgoLined
\SetKwComment{Comment}{/* }{ */}
\caption{Binary Search within dimension $1$}\label{alg: binary-search}
\SetKwFunction{FBS}{BinarySearch}
\SetKwProg{Fn}{Function}{:}{}
\KwData{$a_1,b_1 \in \mathbb{N}: a_1<b_1$ are the endpoints to search in dimension $1$. $T_1 \in \mathbb{N}$ is a target integer. $\mathcal{I}_1$ is a sorted array of intervals. $\text{``sign"} \in \{-1,1\}$ determines whether to search to left or right of target. }
\KwResult{$(I,\text{``}\text{sign}_o\text{"})$: The interval $I$ with endpoint $x$ closest to $T_1$ searched for in the direction of $\text{``sign"}$. $\text{``}\text{sign}_o\text{"}=-1$ if $x$ is a left endpoint; $\text{``}\text{sign}_o\text{"}=+1$ if $x$ is a right endpoint.}
\Fn{\FBS{$a_1$,$b_1$,$T_1$,$\mathcal{I}_1$,$\text{``sign"}$}}{
$m_1\gets \lfloor \frac{a_1+b_1}{2}\rfloor$\\
\uIf{$\text{``sign"}=-1$}{
$\text{``}\text{sign}_o\text{"} \gets -1$\Comment*[r]{Default to a Left Endpoint}
\While{$a_1<m_1$}{
\uIf{$\max(\mathcal{I}_1[m_1]) < T_1$}{
$a_1 \gets m_1$\\
$\text{``}\text{sign}_o\text{"} \gets +1$
}\uElseIf{$\min(\mathcal{I}_1[m_1]) < T_1$}{
$a_1 \gets m_1$\\
$\text{``}\text{sign}_o\text{"} \gets -1$
}\Else{
$b_1 \gets m_1$
}
$m_1\gets \lfloor \frac{a_1+b_1}{2}\rfloor$
}}\uElseIf{$\text{``sign"}=1$}{
$\text{``}\text{sign}_o\text{"} \gets +1$\Comment*[r]{Default to a Right Endpoint}
\While{$m_1< b_1$}{
\uIf{$T_1<\min(\mathcal{I}_1[m_1])$}{
$b_1 \gets m_1$\\
$\text{``}\text{sign}_o\text{"} \gets -1$
}\uElseIf{$T_1<\max(\mathcal{I}_1[m_1])$}{
$b_1 \gets m_1$\\
$\text{``}\text{sign}_o\text{"} \gets +1$
}\Else{
$a_1 \gets m_1$
}
$m_1\gets \lfloor \frac{a_1+b_1}{2}\rfloor$
}
}
\Return $(\mathcal{I}_1[m_1],\text{``}\text{sign}_o\text{"})$
}
\end{algorithm}
\end{document}